\documentclass[11pt]{article}

\usepackage[english]{babel}

\usepackage[letterpaper,top=2cm,bottom=2cm,left=3cm,right=3cm,marginparwidth=1.75cm]{geometry}

\usepackage{amsthm,amssymb,amsmath}  
\usepackage{mathtools}
\usepackage[T1]{fontenc}
\usepackage{graphicx}
\usepackage[colorlinks=true, allcolors=blue]{hyperref}
\usepackage[capitalize]{cleveref}
\usepackage{xspace}
\usepackage{thmtools}
\usepackage{thm-restate}
\usepackage{url}
\usepackage{algorithm}
\usepackage[noend]{algpseudocode}
\usepackage{comment}
\usepackage{caption}
\usepackage{authblk}

\usepackage[utf8]{inputenc}
\usepackage{subcaption}
\usepackage[short,nocomma]{optidef}

\usepackage{cite}
\graphicspath{{images}}

\newtheorem{theorem}{Theorem}
\newtheorem{lemma}{Lemma}
\newtheorem{fact}{Fact}
\newtheorem{remark}{Remark}
\newtheorem{definition}{Definition}
\newtheorem{example}{Example}
\newtheorem{proposition}{Proposition}

\DeclareMathOperator*{\argmin}{arg\,min}

\usepackage{tikz}
\usetikzlibrary{calc,shapes,shapes.misc,calc,arrows,positioning,trees,decorations.pathmorphing,decorations.markings, decorations.pathreplacing, patterns, calligraphy, backgrounds}

\def\dd{\mathinner{.\,.}}
\newcommand{\cO}{\mathcal{O}}
\newcommand{\occ}{\textsf{occ}}
\newcommand{\dna}{\textsc{DNA}\xspace}

\newcommand{\EBOL}{\textsc{EBOL}\xspace}
\newcommand{\COR}{\textsc{COR}\xspace}
\newcommand{\xml}{\textsc{XML}\xspace}
\newcommand{\english}{\textsc{ENGLISH}\xspace}
\newcommand{\sources}{\textsc{SOURCES}\xspace}
\newcommand{\proteins}{\textsc{PROTEINS}\xspace}
\newcommand{\boost}{\textsc{BOOST}\xspace}
\newcommand{\wiki}{\textsc{WIKI}\xspace}

 \newcommand{\defDSproblem}[3]{
  \vspace{2mm}
\noindent\fbox{
  \begin{minipage}{0.96\columnwidth}
  \textsc{#1}\\
  {\bf{Input:}} #2  \\
  {\bf{Output:}} #3
  \end{minipage}
  }
  \vspace{2mm}
}

\newcommand{\BASELINE}{\textsc{RPM-DP}\xspace}
\newcommand{\RPM}{\textsc{RPM}\xspace}
\newcommand{\RPMST}{\textsc{RPM-ST}\xspace}
\newcommand{\RPMESA}{\textsc{RPM-ESA}\xspace}
\newcommand{\RFR}{RFR\xspace}
\newcommand{\FC}{\textsc{FC}\xspace}
\newcommand{\LR}{LR\xspace}
\newcommand{\per}{\textsf{per}}

\title{Resilient Pattern Mining}
\author[1]{Pengxin Bian}
\author[1]{Panagiotis Charalampopoulos}
\author[2]{Lorraine A. K. Ayad}
\author[3]{\\Manal Mohamed}
\author[4,5]{Solon P.\ Pissis}
\author[1]{Grigorios Loukides}
\affil[1]{King's College London, UK}
\affil[2]{Brunel University of London, UK}
\affil[3]{Birkbeck, University of London, UK}
\affil[4]{CWI, Amsterdam, The Netherlands}
\affil[5]{Vrije Universiteit, Amsterdam, The Netherlands}

\date{\vspace{-.5cm}}

\begin{document}

\maketitle

\thispagestyle{empty}

\begin{abstract}
Frequent pattern mining is a flagship problem in data mining.
In its most basic form, it asks for the set of substrings of a given string $S$ of length $n$
that occur at least $\tau$ times in $S$, for some integer $\tau\in[1,n]$.
We introduce a \emph{resilient} version of this classic problem, which we term the \textsc{$(\tau, k)$-Resilient Pattern Mining} (\RPM) problem. Given a string~$S$ of length $n$ and two integers $\tau, k\in[1,n]$, \RPM asks for the set of substrings of $S$ that occur at least $\tau$ times in~$S$, even when the letters at \emph{any} $k$ positions of $S$ are substituted by other letters. Unlike frequent substrings, resilient ones account for the fact that changes to string $S$ are often expensive to handle or are unknown.

We propose an exact $\cO(n\log n)$-time and $\cO(n)$-space algorithm for \RPM, which employs advanced data structures and combinatorial insights. We then present experiments on real large-scale  datasets from different domains demonstrating that: (I) The notion of resilient substrings is useful in analyzing genomic data and is more powerful than that of frequent substrings, in scenarios where resilience is required, such as in the case of versioned datasets; (II) Our algorithm is \emph{several orders of magnitude} faster and more space-efficient than a baseline algorithm that is based on dynamic programming; and (III) Clustering based on resilient substrings is effective. 
\end{abstract}

\clearpage
\setcounter{page}{1}

\section{Introduction}
Given a string $S=S[0\dd n-1]$ of length $n$ and an integer $\tau\in[1,n]$, the classic frequent pattern mining (FPM) problem~\cite{DBLP:conf/icdm/FischerHK05, DBLP:conf/pkdd/FischerHK06, DBLP:journals/tkde/DhaliwalPT12}
asks for the set of substrings of $S$ that occur at least $\tau$ times in~$S$. These substrings are referred to as \emph{$\tau$-frequent}.  
We introduce a \emph{resilient} version of FPM, which we term the \textsc{$(\tau, k)$-Resilient Pattern Mining} (\RPM) problem. Informally stated,  
given a string~$S$ of length~$n$ and integers $\tau, k\in[1,n]$, \RPM asks for the set of substrings of $S$ that occur at least $\tau$ times in $S$, even when \emph{any} $k$ positions of~$S$ undergo letter substitutions.
Clearly, \RPM is a generalization of the FPM problem; it corresponds to the $(\tau, 0)$-RPM case. 

\begin{example}
Let $S=\texttt{aaabaaaabbaaa}$,
$\tau=2$, and $k=1$. 
The output of \textsc{$(\tau, k)$-Resilient Pattern Mining} is $\{\texttt{aaa}, \texttt{aa}, \texttt{a}, \texttt{b}\}$. Each such substring is $(2,1)$-resilient as it has (at least) two occurrences in~$S$ even after the letter in any single position of $S$ is substituted.
For example, the substring \texttt{aaa} of $S$ is $2$-frequent in $S'=\texttt{\underline{aaa}ba\textcolor{red}{b}aabb\underline{aaa}}$, as even when letter $\textcolor{red}{\texttt{b}}$ substitutes~\texttt{a} at position $5$ of $S$, two occurrences of \texttt{aaa} (those underlined in $S'$) still exist.
\end{example}

Here, we consider \RPM under letter substitutions;
considering other edit operations is an intriguing research direction.

\paragraph*{Motivation.}~Unlike frequent substrings, resilient ones account for the fact that changes to a string are often expensive to handle (e.g., we may need to perform 
FPM on the new string from scratch even if it is highly similar to the original one) 
or are unknown (e.g., when caused by adversaries~\cite{clelland1999hiding} or they occur in the future).  
\RPM is motivated by applications in:

{\bf 1.} \emph{Bioinformatics:} In this application domain, there are large collections of very similar strings (also referred to as \emph{repetitive})~\cite{acmsurv1}.   
For example, conserved regions in DNA are stretches of nucleotides that remain relatively unchanged across different individuals, species, or time periods, indicating that these regions likely play essential biological roles~\cite{kryukov2005small}. 
However, even in conserved regions, there may be slight variations due to mutations, sequencing errors, or other factors~\cite{kryukov2005small,stojanovic1999comparison}. Similarly, a pangenome is a collection of DNA sequences each corresponding to the same part of DNA of a different individual~\cite{computational2018computational}. To find conserved regions in a DNA sequence collection, or motifs of interest in a pangenome, a straightforward approach is to solve FPM in each sequence of the collection and then output only the substrings that are frequent in all sequences. This, however, takes time proportional to the size of the collection, which can be in the order of TBs~\cite{boucher2021phoni,turnbull2018100}. 
An alternative approach is to exploit the fact that a bound on \emph{the number} of changes (letter substitutions) of the strings in the collection is known or can be estimated~\cite{acmsurv1,var} 
and apply \RPM into any of these strings. This approach is much  faster and approximates the set of patterns that are frequent in all strings remarkably well, as shown in \cref{example:case}.  

\begin{example}\label{example:case}
    Consider the collection $\mathcal{C}$ of $143,588$ SARS-CoV-2 DNA sequences from~\cite{ncbiGenomeDatasets} with total length $4.176\cdot 10^9$. \cref{fig:runtime:casestudy} shows that mining $(\tau, k)$-resilient substrings from an arbitrarily selected sequence in~$\mathcal{C}$ is more than \emph{four orders of magnitude} times faster on average than mining the $\tau$-frequent substrings that appear in all sequences in~$\mathcal{C}$.
    \cref{fig:jaccard:casestudy} shows the Jaccard Similarity~\cite{buyya2016big} $J_R$ between the set of $\tau$-frequent substrings in all sequences in~$\mathcal{C}$ and the set of $(\tau,k)$-resilient substrings. $J_R$ is $0.96$ on average, implying that the set of $(\tau,k)$-resilient substrings approximates the set of $\tau$-frequent ones very well (i.e., the former substrings are frequent in almost all sequences in $\mathcal{C}$).  
    \cref{fig:jaccard:casestudy} also shows the Jaccard Similarity $J_F$ between 
     the set of $\tau$-frequent substrings in all sequences and the set of $\tau$-frequent substrings in an arbitrarily selected sequence from $\mathcal{C}$. $J_F$ is $0.5$ on average and no more than $0.68$, implying that the latter substrings are not frequent in most of the sequences in $\mathcal{C}$.  
\begin{figure}[!ht]
    \centering
    \begin{subfigure}{0.45\textwidth}
        \includegraphics[width=\textwidth]{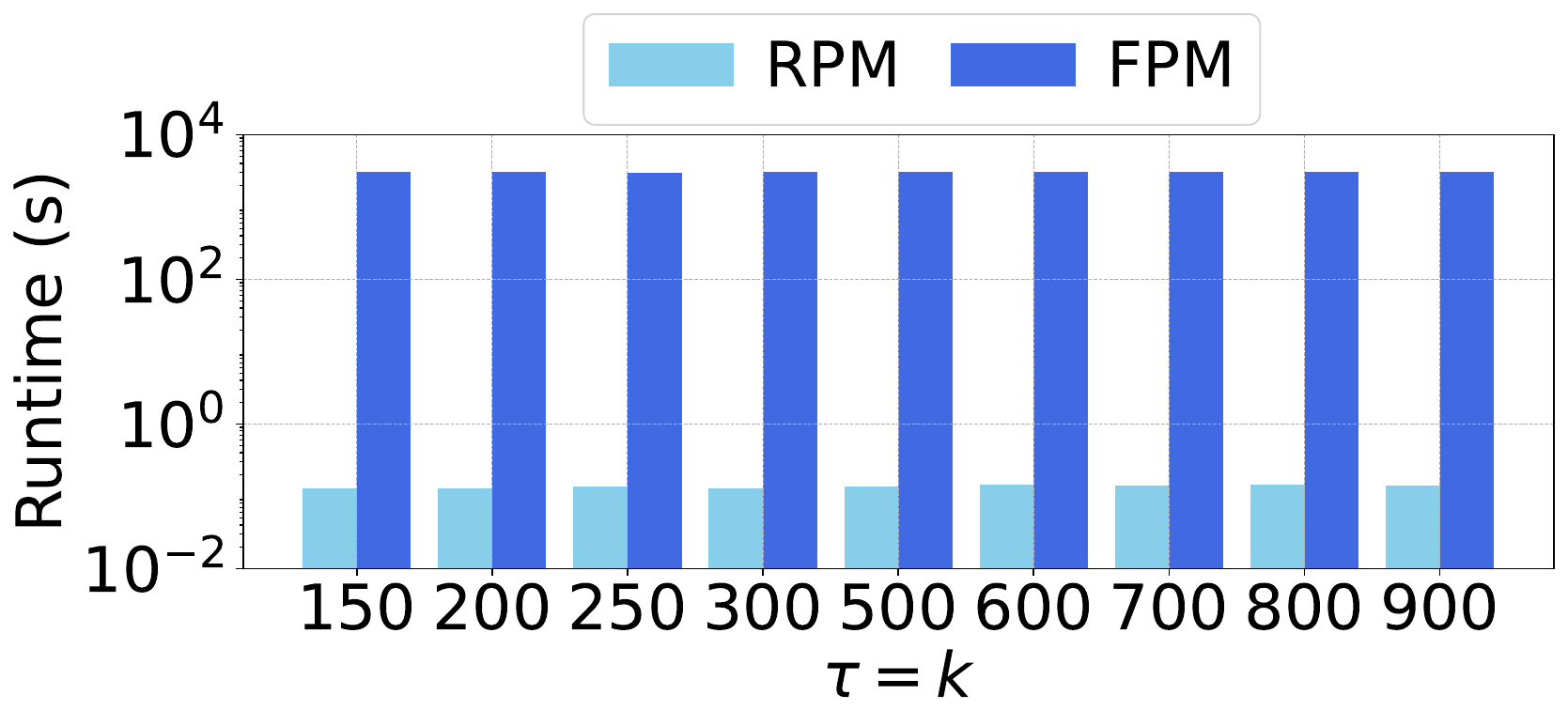}
        \caption{}\label{fig:runtime:casestudy}
    \end{subfigure}
    \hspace{0.05\textwidth}
    \begin{subfigure}{0.45\textwidth}
        \includegraphics[width=\textwidth]{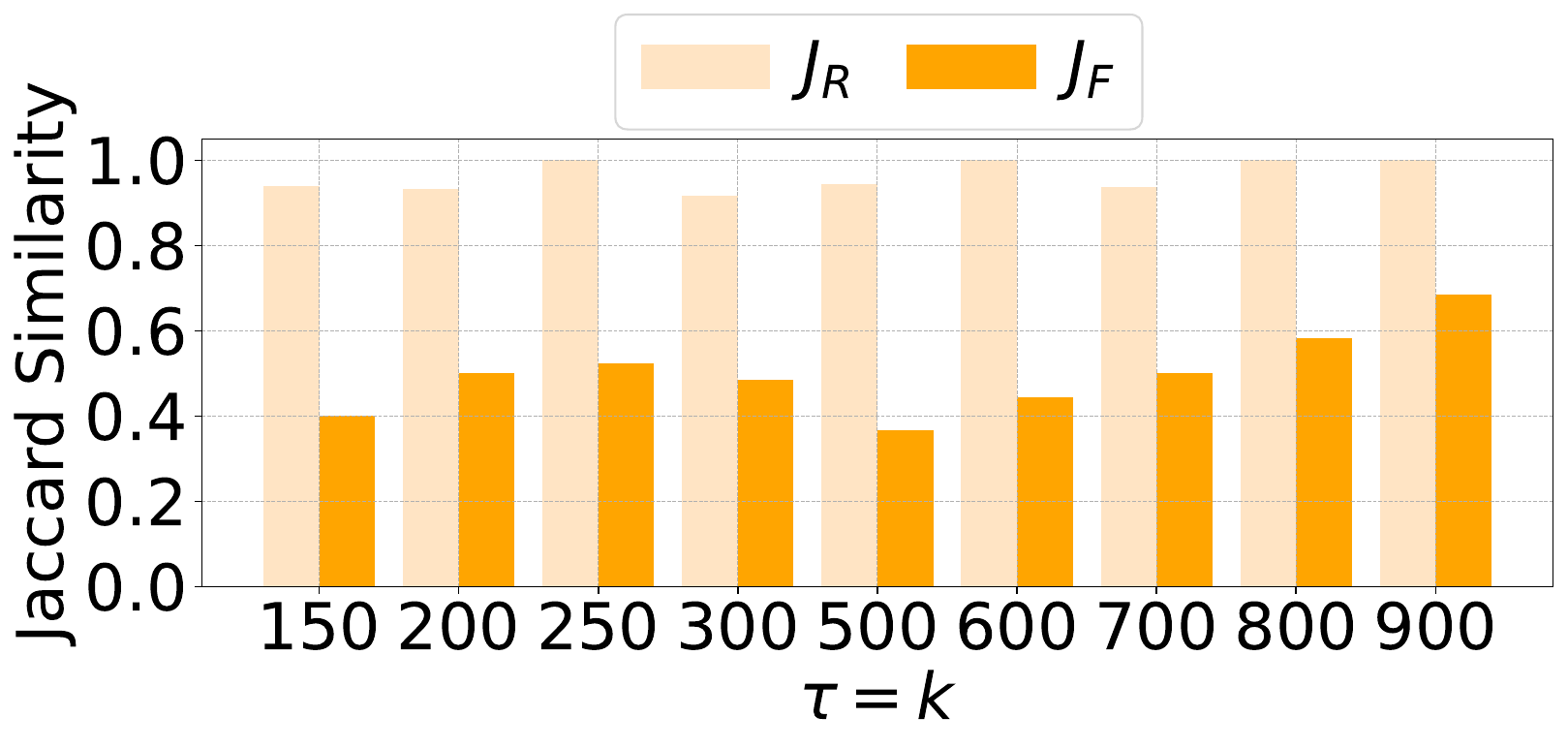}
        \caption{}\label{fig:jaccard:casestudy}
    \end{subfigure}
    \caption{(a) Running time and (b) Jaccard similarity.}
    \label{fig:casestudy}
\end{figure}
\end{example}

{\bf 2.} \emph{Text Analytics:} 
Text data representing webpages~\cite{herz}, tweets~\cite{tao13}, or code in software repositories~\cite{softref} often undergo changes due to updates, leading to slightly different versions of the same text~\cite{soda18}. Such updates may increase or decrease the substring frequencies.
Applying \RPM to a given version of the data ensures that the output substrings remain frequent even after $k$ letter substitutions may occur in later versions. This is useful when outsourcing frequent substrings~\cite{icdeout} or using them in downstream analyses~\cite{dmbook}, as $(\tau,k)$-resilient substrings  remain truthful for longer in these datasets.

{\bf 3.} \emph{Recommender Systems:} String letters are often substituted by adversaries  to force a recommender that is trained on the string to favor their items~\cite{sdm24recom}. Frequent 
 substrings mined from such modified strings may be spurious and others that are indeed frequent may be missed, as letter substitutions can increase or decrease substring frequencies~\cite{tkddsanit,hideandmine}. 
On the other hand, \RPM ensures that spurious substrings are not mined and that the mined ones are indeed frequent, provided that $k$ is set to a sufficiently large value.  

From a technical perspective, \RPM is motivated as follows: 
For any value of $k$, a brute-force algorithm for \RPM considers all possible $\Omega({n \choose k})$ substitutions in $S$, and, for each new instance $S'$ of~$S$, it computes the substrings occurring at least $\tau$ times in~$S'$ and outputs the substrings that are frequent in all $S'$ instances. Clearly, this algorithm is impractical for strings of realistic length~$n$ even for constant $k$. Furthermore, finding frequent patterns with Hamming (or edit) distance at most $k$ is NP-hard~\cite{DBLP:journals/tcs/EvansSW03}, and thus exact algorithms for this problem also have some exponential dependency on $k$~\cite{DBLP:conf/latin/PisantiCMS06}. It is thus interesting to investigate whether the \RPM problem can be solved in polynomial time when $k$ is superconstant.

\paragraph*{Contributions and Paper Organization.}
In this work, we introduce \RPM and make the following contributions:

{\bf 1.} We design a simple exact $\cO(n^3k\log n)$-time and $\cO(n^2)$-space algorithm for \RPM. The main idea is to efficiently find the longest $(\tau,k)$-resilient substring starting at each position of~$S$. These substrings form a compact $\cO(n)$-size representation of 
the $(\tau, k)$-resilient substrings. 
At the heart of our algorithm is a monotonicity property that allows us to apply binary search on the length of the substring starting from a certain position of $S$ and a reduction that allows us to use a dynamic programming  algorithm from~\cite{DBLP:journals/scheduling/ChrobakGLN21} that solves a scheduling problem. We call our algorithm \BASELINE. See \cref{sec:Baseline}.

{\bf 2.} We design an exact $\cO(n \log n)$-time and $\cO(n)$-space  algorithm for \RPM which employs advanced data structures and combinatorial insights. The algorithm outputs the same 
$\cO(n)$-size representation as \BASELINE but is theoretically and practically more time- and space-efficient. Surprisingly, the complexity of our algorithm is independent of $\tau$ and $k$. Its efficiency stems from: (I) the use of an index of $S$ to compactly represent and efficiently explore the space of solutions; and
(II) novel algorithms based on combinatorial insights for efficiently checking whether a \emph{periodic} (see \cref{sec:preliminaries} for a definition) or an \emph{aperiodic} substring is resilient. Interestingly, these checks are performed efficiently based on counting arguments; and not by modifying the string and then inferring the $\tau$-frequent substrings.
We provide two implementations of this algorithm: one called \RPMST that is based on a suffix tree~\cite{DBLP:books/daglib/0020103}; 
and another one called \RPMESA that is based on an enhanced suffix array~\cite{DBLP:journals/jda/AbouelhodaKO04}. See \cref{sec:fast-algo}. For conceptual simplicity, we henceforth assume that all substitutions in~$S$
are done using a letter $\#\notin\Sigma$.
However, in \cref{sec:alphabet}, we prove that the substituting letters can always be chosen from alphabet $\Sigma$, for all alphabets with $|\Sigma|\geq 4$.

{\bf 3.} We present experiments on \emph{seven real datasets} from different domains showing that: (I) Resilient substrings are a fundamentally different notion than frequent substrings, as frequent substrings are often not resilient and, unlike resilient substrings, they do not remain frequent for long in versioned datasets. (II) \RPMST and \RPMESA significantly outperform \BASELINE in terms of running time and space consumption
and, as expected, their running time is not affected by $\tau$ or $k$. For instance, \RPMST is more than \emph{two orders of magnitude} faster and more space-efficient than \BASELINE.
Also, \RPMESA is up to about an order of magnitude faster than \RPMST and it uses about \emph{three times} less space. (III) Clustering based on resilient substrings is effective, as it outperforms a widely-used clustering algorithm~\cite{DBLP:journals/bioinformatics/VingaA03} and a clustering algorithm~\cite{DBLP:journals/tkde/WuZLGZFW23} based on $\tau$-frequent patterns, achieving results very close to the ground truth.   
See \cref{sec:experiments}. 

Related work is discussed in \cref{sec:related}, and \cref{sec:conclusion} concludes the paper. 

\section{Preliminaries}\label{sec:preliminaries}

\paragraph*{Strings.} An \emph{alphabet} $\Sigma$ is a finite set of elements (e.g., integers,  characters, etc.), which we call \emph{letters}. 
A \emph{string} $S=S[0\dd n-1]$ of \emph{length} $|S|=n$ is a sequence of $n$ letters from $\Sigma$, where $S[i]$ denotes the $i$-th letter of the sequence. We refer to each $i\in [0,n)$ as a \emph{position} of~$S$.
We consider throughout, an alphabet $\Sigma=[0,\sigma)$ with $\sigma=n^{\cO(1)}$, which captures virtually any realistic scenario. 
A substring $R$ of $S$ may occur multiple times in $S$.
The set of the starting positions of its \emph{occurrences} in $S$ is denoted by $\occ_S(R)$; we define the \textit{frequency} of~$R$ to be $|\occ_S(R)|$ and may omit the subscript $S$ when it is clear from the context.  
An occurrence of $R$ in~$S$ starting at position $i$ is referred to as a \emph{fragment} of $S$ and is denoted by $R=S[i\dd i+|R|-1]$. Thus, different fragments may correspond to different occurrences of the same substring. 
 A \emph{prefix} of $S$ is a substring of the form $S[0\dd j]$, and a  \emph{suffix} of $S$ is a substring of the form $S[i \dd n-1]$. A prefix $S[0\dd j]$ (resp., suffix $S[i \dd n-1]$) is \emph{proper} if $j<n-1$ (resp., $i>0$). Thus any fragment of $S$ is a prefix of some suffix of $S$.
 The concatenation of two strings $S$ and  $S'$ is denoted by $S\cdot S'$ and the concatenation of $k$ copies of $S$ is denoted by $S^k$.  The \emph{Hamming distance} $d_H(S,S')$ of two strings $S, S'\in \Sigma^n$ of equal length $n$ is $|\{i \in [0, n) : S[i] \neq S'[i]\}|$.

\begin{definition}[Period]
An integer $p>0$ is a \emph{period} of a string $P$ if $P[i] = P[i + p]$ for all $i \in [0,|P|-p)$. The smallest period of $P$ is referred to as \emph{the period} of $P$ and is denoted by $\textsf{per}(P)$. String $P$ is called \emph{periodic} if and only if $\textsf{per}(P)\leq |P|/2$ and \emph{aperiodic} otherwise.\label{def:periodic}
\end{definition}

\begin{example}
For $P=\texttt{abaaabaaabaaaba}$, $\textsf{per}(P)=4$ as $P[i]=P[i+4]$ for all $i\in [0, 15-4)$; $p=8$ is also a period of $P$. Since $\textsf{per}(P)=4\leq |P|/2=15/2$, $P$ is periodic.   
\end{example}

\begin{definition}[Run]\label{def:run}
A fragment $F$ of a string $S$ is a \emph{run} of $S$ if and only if 
$\textsf{per}(F) \leq |F|/2$ and $F$ cannot be extended in either direction with its period remaining the same. 
\end{definition}

\begin{example}
Consider 
$S=\texttt{\underline{abaaabaaabaaaba}b}$. The \emph{underlined} fragment $F = S[0\dd 14]$ is a run, as its period $\textsf{per}(F)=4\leq |F|/2=15/2$ and it cannot be extended to the right by one letter with its period remaining the same, as then it would become equal to $S$ and $\textsf{per}(S)=14\neq 4$. 
\end{example}

\begin{definition}[Lyndon Root]
The \emph{Lyndon root} of a run $F$ is the lexicographically smallest rotation of $F[0 \dd \textsf{per}(F)-1]$.
\end{definition}

\begin{definition}\label{def:canonical}
A run $F$ with Lyndon root $L$ has a unique \emph{canonical representation} $L_1\cdot L^k\cdot L_2$, where $L_1$ is a proper suffix of $L$, $L_2$ is a proper prefix of $L$, and $k \geq 0$ is a nonnegative integer; either of $L_1$ and $L_2$ may be the empty string.
\end{definition}

\begin{example}
The run $F=\texttt{\textcolor{red}{ab}\underline{aaabaaabaaab}\textcolor{red}{a}}$ has Lyndon root $\texttt{aaab}$, as the rotations of $F[0\dd 4-1]$ are $\texttt{abaa}$, $\texttt{baaa}$, $\texttt{aaab}$, and $\texttt{aaba}$, and the lexicographically smallest one is $\texttt{aaab}$. 
The canonical representation of $F$ is $\texttt{ab}\cdot (\texttt{aaab})^{3}\cdot \texttt{a}$.
\end{example}

\paragraph*{Indexes.}
For a set $\mathcal{S}$ of strings, the \emph{trie} $\textsf{T}(\mathcal{S})$ is a rooted tree whose nodes are in one-to-one correspondence with the prefixes of the strings in $\mathcal{S}$~\cite{DBLP:books/daglib/0020103}. 
The edges of $\textsf{T}(\mathcal{S})$ are labeled with letters from $\Sigma$; the prefix corresponding to node $v$ is denoted by $\textsf{str}(v)$ and equals the concatenation of the letters labeling the edges on the root-to-$v$ path.
The node $v$ is called the \emph{locus} of $\textsf{str}(v)$. 

The order on $\Sigma$ induces an order on the edges outgoing from any node of $\textsf{T}(\mathcal{S})$.
A node $v$ is \emph{branching} if it has at least two children and \emph{terminal} if $\textsf{str}(v) \in \mathcal{S}$. 

A \emph{compacted trie} is obtained from $\textsf{T}(\mathcal{S})$ by dissolving all nodes except the root, the branching, and the terminal nodes. The dissolved nodes are called \emph{implicit}; the preserved ones are called \emph{explicit}.
The edges of the compacted trie are labeled by strings. The \emph{string depth} $\textsf{sd}(v)=|\textsf{str}(v)|$ of node $v$ is the length of the string it represents, i.e., the total length of the strings labeling the path from the root to~$v$. 
The compacted trie takes $\cO(|\mathcal{S}|)$ space provided that edge labels are stored as fragments of strings in $\mathcal{S}$ (a fragment $S[i\dd j]$ can be stored in $\cO(1)$ space~\cite{DBLP:conf/focs/Farach97}).  
Given the lexicographic order on~$\mathcal{S}$ along with the lengths of the longest common prefixes between any two consecutive (in this order) elements of~$\mathcal{S}$, one can construct the compacted trie of $\mathcal{S}$ in $\cO(|\mathcal{S}|)$ time~\cite{DBLP:conf/cpm/KasaiLAAP01}.  

The \emph{suffix tree} $\textsf{ST}(S)$ of a string $S$ is the compacted trie of the set of all suffixes of $S$; see \cref{fig:esa1}. 
Each terminal node $v$ of $\textsf{ST}$ is labeled by $n-\textsf{sd}(v)$, i.e., the starting position of the suffix it represents.
$\textsf{ST}(S)$ can be constructed in $\cO(n)$ time for any string $S$ of length~$n$ over $\Sigma$~\cite{DBLP:conf/focs/Farach97}.
After an $\cO(n)$-time preprocessing of $\textsf{ST}(S)$, we can compute the length of the longest common prefix of any two suffixes in $\cO(1)$ time.
We denote the length of the longest common prefix of $S[i\dd n-1]$ and $S[j\dd n-1]$ by $\textsf{lcp}(i,j)$.  
The \emph{suffix array} $\textsf{SA}(S)$ of $S$~\cite{DBLP:journals/siamcomp/ManberM93} is the permutation of $[0,n)$ such that $\textsf{SA}[i]$ is the starting position of the $i$-th lexicographically smallest suffix of $S$; see \cref{fig:esa2} and \cref{fig:esa3}.
It can be constructed in $\cO(n)$ time for any string~$S$ of length $n$ over $\Sigma$~\cite{DBLP:conf/focs/Farach97}.
The $\textsf{LCP}(S)$ array~\cite{DBLP:journals/siamcomp/ManberM93} of $S$ stores the length of the longest common prefix of lexicographically adjacent suffixes. For $j>0$, $\textsf{LCP}[j] = \textsf{lcp}(j-1,j)$  and $\textsf{LCP}[0]=0$; see \cref{fig:esa4}. Given \textsf{SA}$(S)$, $\textsf{LCP}(S)$ can be constructed in $\cO(n)$ time~\cite{DBLP:conf/cpm/KasaiLAAP01}.

\paragraph*{Problem Definition.} Let us now formalize the studied problem.

\begin{definition}[$\tau$-frequent substring]
A substring $X$ of $S$ is $\tau$-\emph{frequent} in $S$ if it occurs at least $\tau$ times in $S$.
\end{definition}

\begin{definition}[$(\tau,k)$-resilient substring]\label{def:resilient}  
A substring $X$ of $S\in \Sigma^n$ is $(\tau,k)$-\emph{resilient} in $S$ if $X$ is $\tau$-frequent in all strings $S'\in \Sigma^n$ that satisfy $d_H(S,S') \leq k$.
\end{definition}

\defDSproblem{$(\tau, k)$-Resilient Pattern Mining (\RPM)  }
{A string $S\in \Sigma^n$ and two integers $\tau, k \in [1,n]$.}
{An array $\textsf{OUTPUT}$ of size $n$ such that $\textsf{OUTPUT}[i]$ is the length of the longest prefix of $S[i\dd n-1]$ that is $(\tau,k)$-resilient.}

Parameter $k$ can be set based on domain knowledge (e.g., in bioinformatics, we can use the substitution rate~\cite{var}). The array \textsf{OUTPUT} is a \emph{compact representation} of the set of $(\tau,k)$-resilient substrings of $S$. It is based on \cref{fct:mono}.

\begin{fact}[Monotonicity]\label{fct:mono}
  If string $Z=X\cdot Y$ is $(\tau,k)$-resilient in $S$, where $X,Y$ are strings, then so are $X$ and $Y$.
\end{fact}
\begin{proof}
Assume, toward a contradiction, that $Z$ is $(\tau,k$)-resilient and either $X$ or $Y$ is not $(\tau,k)$-resilient. Say $X$ is not $(\tau,k)$-resilient (the proof for $Y$ is analogous). Then there exists a string $S'$ that is obtained from $S$ after substitution $k$ letters with $\#$ such that $X$ is not $\tau$-frequent in $S'$. However, $Z$ is by definition $\tau$-resilient in $S'$ and each occurrence of $Z$ in $S'$ implies at least one occurrence of $X$ in $S'$ because~$X$ is a prefix of $Z$, this is a contradiction.
\end{proof}

In \cref{app:output}, we explain how to explicitly construct the set of $(\tau,k)$-resilient substrings of $S$ using the $\textsf{OUTPUT}$ array in $\cO(n)$ time plus time linear in the number of these substrings.

\begin{figure} 
\begin{minipage}[b]{.44\columnwidth}
\hspace{-3mm}
\scalebox{0.9}{
\subfloat[]{\label{fig:esa1}
 \begin{tikzpicture}[- , level distance=1.5cm]
\filldraw[black] (0,0) circle (2pt); 
\filldraw[black] (-2,-1) circle (2pt);
\filldraw[black] (1.59,-1) circle (2pt);
\filldraw[black] (-1.67,-2) circle (2pt);
\filldraw[black] (-0.67,-1) circle (2pt);
\filldraw[black] (-0.47,-2) circle (2pt);
\filldraw[black] (-1.2,-2.8) circle (2pt);
\filldraw[black] (0.55,-2.2) circle (2pt);
\filldraw[black] (-0.35,-3) circle (2pt);
\filldraw[black] (1.29,-2.2) circle (2pt);
\filldraw[black] (2.29,-2.2) circle (2pt);
  \draw (0,0)   -- (-2,-1) node[midway,sloped,above] {$\texttt{\$}$};
  \draw (0,0)   -- (-0.67,-1) node[midway,sloped,above] {$\texttt{a}$};
  \draw (0,0)   -- (0.55,-2.2) node[midway,sloped,above] {$\texttt{banana\$}$};
  \draw (0,0) -- (1.59,-1)   node[midway,sloped,above] {$\texttt{na}$};
    \draw (-1.67,-2)   -- (-0.67,-1) node[midway,sloped,above] {$\texttt{\$}$};
        \draw (-0.47,-2)   -- (-0.67,-1) node[midway,sloped,above] {$\texttt{na}$};
           \draw (-0.47,-2)   -- (-1.2,-2.8) node[midway,sloped,above] {$\texttt{\$}$};
              \draw (-0.47,-2)   -- (-0.35,-3) node[midway,sloped,above] {$\texttt{na\$}$};
    \draw (1.59,-1) -- (1.29,-2.2)   node[midway,sloped,above] {$\texttt{\$}$};
    \draw (1.59,-1) -- (2.29,-2.2)   node[midway,sloped,above] {$\texttt{na\$}$};

\node[] at (-1.65,-2.3) {$l_5$};
\node[draw, rectangle, fill=white, inner sep=1pt] at (-1.96, -2) {5};
\node[] at (-2,-1.3) {$l_6$};
\node[draw, rectangle, fill=white, inner sep=1pt] at (-2.3, -1) {6};
\node[] at (-1.2,-3.1) {$l_3$};
\node[draw, rectangle, fill=white, inner sep=1pt] at (-1.5, -2.8) {3};
\node[] at (-0.1,-3.1) {$l_1$};
\node[draw, rectangle, fill=white, inner sep=1pt] at (-0.63, -3) {1};
\node[] at (0.6,-2.5) {$l_0$};
\node[draw, rectangle, fill=white, inner sep=1pt] at (0.3, -2.2) {0};
\node[] at (1.3,-2.5) {$l_4$};
\node[draw, rectangle, fill=white, inner sep=1pt] at (1.02, -2.2) {4};
\node[] at (2.3,-2.5) {$l_2$};
\node[draw, rectangle, fill=white, inner sep=1pt] at (2, -2.2) {2};

\node[draw, circle, fill=white, inner sep=0.2pt] at (-0.36, -0.95) {3};

\node[] at (-0.7,-2) {$v$};
\node[draw, circle, fill=white, inner sep=0.2pt] at (-0.16, -1.95) {2};

\node[draw, circle, fill=white, inner sep=0.2pt] at (1.89, -0.95) {2};

\end{tikzpicture}  
}}
\end{minipage}
\hspace{+1.6mm}
\begin{minipage}[b]{.14\columnwidth}
\scalebox{0.93}{
\subfloat[]{\label{fig:esa2}
\setlength{\fboxsep}{1.5pt}
\definecolor{aqua}{rgb}{0.0, 1.0, 1.0}
\begin{tabular}{l|l|}
\cline{2-2}
\textcolor{gray}{0} & 6 \\
\textcolor{gray}{1} & 5 \\
\textcolor{gray}{2} & 3 \\
\textcolor{gray}{3} & 1 \\
\textcolor{gray}{4} & 0 \\
\textcolor{gray}{5} & 4 \\
\textcolor{gray}{6} & 2 \\
\cline{2-2}
\end{tabular}
}}
\end{minipage}
\begin{minipage}[b]{.14\columnwidth}
\centering
\scalebox{0.93}{
\subfloat[]{\label{fig:esa3}
\centering
\setlength{\fboxsep}{0.2pt}
\definecolor{aqua}{rgb}{0.0, 1.0, 1.0}
\begin{tabular}{l}
\hspace{-3mm} \colorbox{white}{\texttt{\$}}\\
\hspace{-3mm}  \colorbox{aqua}{\texttt{a}}\texttt{\$} \\
\hspace{-3mm}  \colorbox{aqua}{\texttt{a}}\colorbox{yellow}{\texttt{na}}\texttt{\$} \\
\hspace{-3mm} \colorbox{yellow}{\texttt{ana}}\texttt{na\$} \\
\hspace{-3mm} \colorbox{white}{\texttt{banana\$}}\\
\hspace{-3mm} \colorbox{green}{\texttt{na}}\texttt{\$}\\
\hspace{-3mm} \colorbox{green}{\texttt{na}}\texttt{na\$}\\
\end{tabular}
}}
\end{minipage}
\begin{minipage}[b]{.14\columnwidth}
\scalebox{0.93}{
\subfloat[]{\label{fig:esa4}
\setlength{\fboxsep}{1.5pt}
\definecolor{aqua}{rgb}{0.0, 1.0, 1.0}
\begin{tabular}{l|l|}
\cline{2-2}
\textcolor{gray}{0} & 0 \\
\textcolor{gray}{1} & 0 \\
\textcolor{gray}{2} & \colorbox{aqua}{1} \\
\textcolor{gray}{3} & \colorbox{yellow}{3} \\
\textcolor{gray}{4} & 0 \\
\textcolor{gray}{5} & 0 \\
\textcolor{gray}{6} & \colorbox{green}{2} \\
\cline{2-2}
\end{tabular}
}
}
\end{minipage}
\caption{$S=\texttt{banana\$}$: (a) \textsf{ST}$(S)$; each branching non-root node is labeled by its frequency (circle) and each terminal node $l_i$ by its associated starting position $i$ (square). Node $v$ is the locus of $\textsf{str}(v)=\texttt{ana}$ and $\textsf{sd}(v)=|\textsf{str}(v)|=3$. (b) \textsf{SA}$(S)$. (c) The lexicographically sorted suffixes of $S$. (d) \textsf{LCP}$(S)$.} \label{fig:esa}
\end{figure}
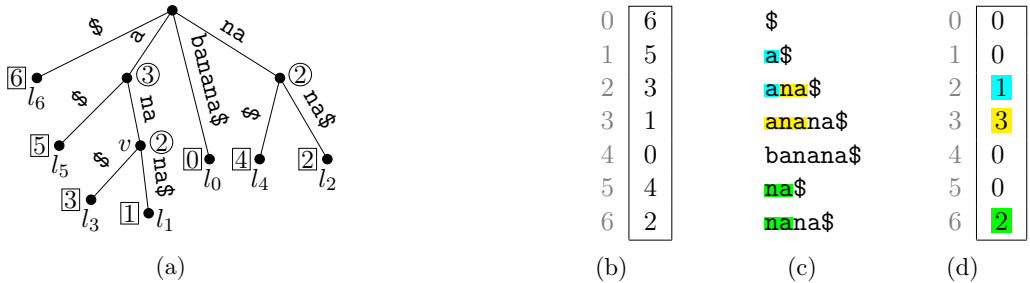

\section{A DP-based Polynomial Algorithm for \RPM} \label{sec:Baseline}

We design \BASELINE, a dynamic programming (DP) based polynomial-time algorithm for  \textsc{$(\tau,k)$-Resilient Pattern Mining}, for any integer $\tau,k\in[1,n]$. \BASELINE is founded on the following result.

\begin{theorem}[\hspace{0.2mm}\cite{DBLP:journals/scheduling/ChrobakGLN21}, Section 3.1]\label{the:intervals}
Given a family $\mathcal{F}$ of $f$ arbitrary intervals on the real line and a positive integer $k\leq f$, we can find a set of at most $k$ points on the line that intersect a maximum number of intervals from $\mathcal{F}$ in $\cO(kf^2)$ time and $\cO(f^2)$ space.
\end{theorem}

By \cref{fct:mono}, we can binary search on the length of the substring starting at position $i$, for all $i\in[0,n)$, to find the longest substring $S[i\dd j]$ that is $(\tau,k)$-resilient.
Thus, it suffices to check $\cO(n \log n)$ substrings of $S$.

Let us explain how we check whether a fixed substring $Z$ of~$S$ is $(\tau,k)$-resilient.  
We first find  $\occ_S(Z)$, the set of occurrences of $Z$ in $S$, in $\cO(n)$ time using the KMP algorithm~\cite{DBLP:journals/siamcomp/KnuthMP77}.
If $|\occ_S(Z)|$ is at least~$\tau$ we apply the DP algorithm from~\cite{DBLP:journals/scheduling/ChrobakGLN21}; we also describe it in \cref{app:DP}. 
This DP algorithm underlies \cref{the:intervals}; it finds a set of at most $k$ points on the line that intersect a maximum number $M$ of intervals in a given family of $f$ intervals, in $\cO(kf^2)$ time and $\cO(f^2)$ space.   
Specifically, for each occurrence 
$i\in \occ_S(Z)$, the interval $[i,i+|Z|-1]$ 
represents 
the fragment $S[i\dd i+|Z|-1]$ 
and the points that intersect the intervals represent the letters of the fragments that, if substituted by $\#\notin\Sigma$, change the frequency of 
$Z$ as much as possible. 
Hence, $f = |\occ_S(Z)|$ and $Z$ is \emph{not} $(\tau,k)$-resilient if and only if $f - M < \tau$ (in which case $Z$ would cease being $\tau$-frequent after the letter substitutions). 

Since we check $\cO(n\log n)$ substrings, and each check takes $\cO(kn^2)$ time and $\cO(n^2)$ space, as $f= |\occ_S(Z)|\leq n$, 
we obtain an $\cO(n^3 k\log n)$-time and $\cO(n^2)$-space algorithm for \RPM. 
The algorithm returns the $\textsf{OUTPUT}$ array as required.

\begin{proposition}
    For any string $S\in \Sigma^n$, with $|\Sigma|\geq 4$, and integers $\tau, k \in [1,n]$, the \textsc{$(\tau,k)$-Resilient Pattern Mining} problem can be solved in $\cO(n^3 k \log n)$ time and $\cO(n^2)$ space.
\end{proposition}

\section{An Efficient Algorithm for \RPM}\label{sec:fast-algo}

\paragraph*{High-Level Idea.} The main idea is to \emph{implicitly} consider all substrings of $S$. Each substring is either \emph{aperiodic} or \emph{periodic}: (I) The occurrences of an aperiodic substring $R$ of $S$ do not overlap a lot, so we can use a simple greedy algorithm to count the maximal number of occurrences of $R$ that may be lost after $k$ substitutions in $S$.
If the frequency of $R$ in $S$ minus this number is at least $\tau$, then~$R$ is $(\tau,k)$-resilient; otherwise, it is not. (II) The occurrences of a periodic substring~$P$ of $S$ can overlap a lot but have a predictable structure: they form compact clusters of occurrences of~$P$ 
and the clusters do not overlap much. The technical difficulty is to retrieve these clusters efficiently, in time proportional to their number and \emph{not} to the total number of occurrences of $P$. After retrieving these clusters, we use a more involved version of the greedy algorithm to check whether $P$ is $(\tau,k)$-resilient in $S$ or not using counting arguments similarly to the aperiodic case. 

The total number of substrings of $S$ is $\Theta(n^2)$.
By \cref{fct:mono}, it suffices to check $\cO(n \log n)$ substrings: $\cO(\log n)$ for each starting position of $S$ in a binary-search fashion. This way we obtain a compact representation of the set of $(\tau,k)$-resilient substrings of $S$ in $\cO(n \log n)$ time using $\cO(n)$ space.

\paragraph*{Roadmap.} The main algorithm is described in \cref{sec:main}. 
It relies on two combinatorial lemmas (\cref{sec:lemmas}) and three functions (\cref{sec:functions}). 
We present the algorithm using a suffix tree~\cite{DBLP:conf/focs/Farach97}, for simplicity. It is well-known that the suffix tree functionality (e.g., bottom-up traversal) can be simulated using an enhanced suffix array~\cite{DBLP:journals/jda/AbouelhodaKO04} (see \cref{app:practical}), 
resulting in time and space improvements in practice.

\subsection{Combinatorial Lemmas}\label{sec:lemmas}

We give two crucial lemmas 
for our algorithm's efficiency; 
they are folklore but we 
give their proofs for completeness. 
\cref{lem:aperiodic} is useful for handling the aperiodic substrings of $S$. Intuitively, 
it  guarantees that any position of $S$ is contained in at most $2$ occurrences of an aperiodic substring $R$; see also \cref{fig:aperiodic}.

\begin{restatable}{lemma}{aper}\label{lem:aperiodic}
Let $S$ be a string and $R$ be an aperiodic substring of $S$.
For any three distinct occurrences of $R$ in $S$ at positions $x$, $y$, and $z$ with $x<y<z$, we have $z>x+|R|$.
\end{restatable}
\begin{proof}
If two occurrences of $R$ in $S$ started $p \leq |R|/2$ positions apart, then $\textsf{per}(R) \leq p \leq |R|/2$ and hence $R$ would be periodic, a contradiction. 
Thus, we have $z > y + |R|/2 > x +|R|$.
\end{proof}

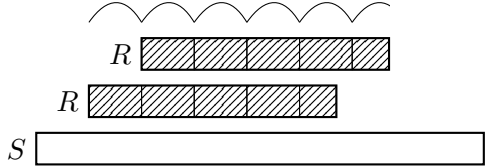
\begin{figure}[ht]
 \begin{center}
 \begin{tikzpicture}
      \begin{scope}[yshift=0cm, xshift=0cm, scale=.7]
				\draw[pattern = north east lines] (0,0.3) rectangle (1,0.9);
				\draw[pattern = north east lines] (1,1.2) rectangle (2,1.8);
				\draw[pattern = north east lines] (1,0.3) rectangle (2,0.9);
				\draw[pattern = north east lines] (2,1.2) rectangle (3,1.8);
				\draw[pattern = north east lines] (2,0.3) rectangle (3,0.9);
				\draw[pattern = north east lines] (3,1.2) rectangle (4,1.8);
				\draw[pattern = north east lines] (3,0.3) rectangle (4,0.9);
				\draw[pattern = north east lines] (4,1.2) rectangle (5,1.8);
				\draw[pattern = north east lines] (4,0.3) rectangle (4.7,0.9);
				\draw[pattern = north east lines] (5,1.2) rectangle (5.7,1.8);
            \draw[thick] (-1,0) rectangle (7.5,-0.6);
            \draw (-1,-0.3) node[left] {$S$};
            \draw[thick] (0,0.3) rectangle (4.7,0.9);
            \draw (0,0.6) node[left] {$R$};
                \draw[thick] (1,1.2) rectangle (5.7,1.8);
                \draw (1,1.5) node[left] {$R$};
                \clip (0,1) rectangle (5.7,2.5);
                \foreach \x in {0,1,2,3,4,5}{
                    \draw[xshift=\x cm,yshift=1.1cm] (0,1) .. controls (0.35,1.5) and (0.7,1.5) .. (1,1);
                }
    \end{scope}   
 \end{tikzpicture}
 \end{center}
 \caption{The periodic structure of $R$ in the case when it has two occurrences in $S$ that start at most $|R|/2$ positions apart.}
 \label{fig:aperiodic}
\end{figure}

\cref{lem:periodic} is useful for handling the periodic substrings of $S$.
Intuitively, it guarantees that any position of $S$ is contained in at most two runs with the same period.

\begin{restatable}{lemma}{periodiclem}\label{lem:periodic}
Let $S$ be a string.
Suppose that for positions $i<j<k$ there are runs $S[i\dd i']$, $S[j\dd j']$, and $S[k\dd k']$ in $S$ that have the same period.
We then have $k>i'$.
\end{restatable}
\begin{proof}
Let the common period of the considered runs be $p$.
First, note that $j' \geq j+2p-1$ since, for any run $R$, we have $|R| \geq 2\textsf{per}(R)$.
Now, two runs with the same period $p$ cannot overlap by $x \geq p$ positions; if this were the case, the portion of $S$ spanned by these runs would have period $p$ and this would contradict the maximality of at least one of the considered runs.
Due to the overlap constraint, we have $k > j' - p+1$ and $j > i' - p + 1$.
Putting everything together, we obtain
$k > j' - p +1 \geq (j+2p-1) - p+1 \geq (i'- p + 1) +p = i'$.
\end{proof}

\subsection{Main Algorithm}\label{sec:main}

The main algorithm has two phases and returns \textsf{OUTPUT}, a compact representation of the $(\tau,k)$-resilient substrings. 

\paragraph*{Preprocessing.} We construct the suffix tree $\textsf{ST}(S)$ of $S$ in $\cO(n)$ time and space. By construction (inspect \cref{fig:esa1} and \cref{fig:esa2}), the leaves of \textsf{ST}$(S)$ of $S$, read from left to right (e.g., using a post-order traversal), precisely form the suffix array \textsf{SA}$(S)$ of $S$.
Using such a traversal, for each branching node $u$ in \textsf{ST}$(S)$, we record the \textsf{SA} indices of the leftmost and rightmost leaves in the subtree rooted at $u$, denoted by $l_u$ and~$r_u$, respectively.
Then, $u$ stores interval $[l_u,r_u]$
where, $\textsf{str}(u) =S[\textsf{SA}[l_u]\dd \textsf{SA}[l_u]+\textsf{sd}(u) -1]=\ldots=S[\textsf{SA}[r_u]\dd \textsf{SA}[r_u]+\textsf{sd}(u) -1]$, and $r_u-l_u+1$ is the frequency of $\textsf{str}(u)$.
Lastly, we associate each node of \textsf{ST}$(S)$ with a flag initialized  to \textsf{FAIL}. This preprosessing takes linear time in the size of \textsf{ST}$(S)$, i.e., $\cO(n)$ time.

\paragraph*{Phase 1.} We perform a bottom-up (DFS) traversal of \textsf{ST}$(S)$ to find a \emph{cut} of \textsf{ST}$(S)$ (i.e., a set of nodes that partitions the suffix tree into disjoint subtrees) where the nodes that form the cut and their ancestors are exactly the nodes that represent $(\tau,k)$-resilient substrings.
When some node $u$ is explored during the traversal, we check its flag.
If the flag value is \textsf{FAIL}, then we read the \textsf{SA} interval $[l_u,r_u]$ stored at $u$, check whether $r_u-l_u+1\geq \tau + k$, and, if so, apply the \textsc{IsPeriodic} function to $\textsf{str}(u)$, which determines whether $\textsf{str}(u)$ is periodic or not. If \textsc{IsPeriodic} answers \textsf{YES}, we apply the \textsc{PeriodicResilient} function to $\textsf{str}(u)$; otherwise, we apply the \textsc{AperiodicResilient} function.
If the applied function answers \textsf{YES}, $\textsf{str}(u)$ is $(\tau,k)$-resilient, and the flag of $u$ is changed to \textsf{SUCCESS}; otherwise, $\textsf{str}(u)$ is not $(\tau,k)$-resilient and the flag of $u$ is kept unchanged as \textsf{FAIL}. Whenever a node $u$ has its flag changed  to \textsf{SUCCESS}, all the ancestors of $u$ should also have their flags changed to \textsf{SUCCESS}; we thus propagate the \textsf{SUCCESS} flag upward before proceeding with our traversal. The end of the traversal gives a \emph{preliminary cut}. The upper part including the cut is composed  of \textsf{SUCCESS} nodes and the lower part is composed of \textsf{FAIL} nodes.

\paragraph*{Phase 2.} 
We use \cref{fct:mono} to refine the preliminary cut from Phase~1. We try to extend the $(\tau,k)$-resilient substrings we found in Phase~1 as much as possible to compute the final output.  For each edge $(u,v)$ in \textsf{ST}$(S)$, such that: (a) $\textsf{str}(u)$ is $(\tau,k)$-resilient (i.e., the flag of $u$ is equal to \textsf{SUCCESS} and $u$ is on the cut); (b) $\textsf{str}(v)$ is not $(\tau,k)$-resilient (the flag of $v$ is equal to \textsf{FAIL} and $v$ is below the cut); and (c) $\textsf{str}(v)$ has at least $\tau+k$ occurrences in $S$, we perform binary search on the interval $[\textsf{sd}(u), \textsf{sd}(v)]$ to compute \emph{the longest prefix of $\textsf{str}(v)$ that is $(\tau,k)$-resilient}. (Note that if condition (c) does not hold, we can already conclude that the sought prefix of $\textsf{str}(v)$ is $\textsf{str}(u)$.)
In particular, in each iteration of the binary search, we first call \textsc{IsPeriodic}  and then either \textsc{PeriodicResilient} or  \textsc{AperiodicResilient} (similarly to Phase 1). 
In the end, we obtain our \emph{refined final cut}.
All in all, for each edge $(u,v)$ such that~$u$ is on the preliminary cut and $v$ is not, we have computed the longest prefix of $\textsf{str}(v)$ that is $(\tau,k)$-resilient; let us denote the length of this prefix by $m(u,v)$.
Then, for each such edge $(u,v)$,
we read the \textsf{SA} interval $[l_v,r_v]$ from~$v$,
and set $\textsf{OUTPUT}[\textsf{SA}[i]]:=m(u,v)$, for all $i\in[l_v,r_v]$, in an initially all-zeroes array $\textsf{OUTPUT}[0\dd n-1]$. 
 
As $\textsf{ST}(S)$ has $\cO(n/(\tau+k))$ nodes, each with at least $\tau + k$ leaf descendants, that call the three functions, we directly obtain:

\begin{restatable}{lemma}{lemred}\label{lem:reduction}
    The \textsc{$(\tau,k)$-Resilient Pattern Mining} problem can be reduced in $\cO(n \log n)$ time, using $\cO(n)$ space, to
    $\cO(n \log n/(\tau+k))$ calls to
    \textsc{IsPeriodic,} \textsc{PeriodicResilient}, and \textsc{AperiodicResilient}.
\end{restatable}

\begin{remark}
Conceptually, one can consider a simpler algorithm that performs binary search to compute, for each node $u$ of \textsf{ST}$(S)$ with at least $\tau+k$ leaf descendants, the longest prefix of $\textsf{str}(u)$ that is $(\tau,k)$-resilient.
The problem is that in order to do this, we would need random access to the sorted list of occurrences of $\textsf{str}(v)$ at each node~$v$ of \textsf{ST}$(S)$, which is costly. 
However, as we show below, we can maintain a sorted list of occurrences during a DFS traversal.
\end{remark}

\subsection{The Functions of the Algorithm}\label{sec:functions}

In this section, we present the theoretically efficient implementations of the \textsc{IsPeriodic}, \textsc{AperiodicResilient}, and \textsc{PeriodicResilient} functions. For each function, we also discuss our practical implementation when it differs from the theoretically efficient one. These functions assume that, during the bottom-up (DFS) traversal of $\textsf{ST}(S)$, at each node $v$,  we can retrieve a sorted list of the occurrences of $\textsf{str}(v)$ in $S$.
Any such list is the \emph{sorted} union of the disjoint lists of the children of $v$.
By implementing these lists using mergeable AVL trees, 
we can compute all of them in $\cO(n\log n)$ total time.
This is rather standard but we also describe it in \cref{app:AVL_trees}.

\paragraph*{The \textsc{IsPeriodic} Function.}~For any string $S$ of length $n$ and a fragment $F$ of~$S$, a \emph{periodic extension} query~\cite{DBLP:journals/siamcomp/KociumakaRRW24} decides whether $F$ is periodic, and, if so, returns  $\per(F)$. We rely on the following result.

\begin{theorem}[\hspace{0.2mm}\cite{DBLP:journals/siamcomp/KociumakaRRW24}]\label{the:2-period}
For any string of length $n$, after an $\cO(n)$-time preprocessing, we can answer periodic extension queries in $\cO(1)$ time.
\end{theorem}

Unfortunately, there is no available implementation of the data structure underlying \cref{the:2-period}, and it seems that an implementation of it would require heavy engineering in order to be practically efficient; \cref{the:2-period} relies on an intricate string sampling scheme and other sophisticated (auxiliary) data structures.
We thus propose a practical implementation that relies on first computing the runs of~$S$~\cite{DBLP:conf/icalp/Ellert021} and then indexing them using an interval tree~\cite{DBLP:books/daglib/0023376}.

We first compute the runs of~$S$ in $\cO(n)$ time using the algorithm of~\cite{DBLP:conf/icalp/Ellert021}. 
Since there are $\cO(n)$ runs~\cite{DBLP:journals/siamcomp/BannaiIINTT17},  
each represented by a subinterval of $[0,n)$, we index them in an interval tree in $\cO(n\log n)$ time and $\cO(n)$ space~\cite{DBLP:books/daglib/0023376}. 
Given a substring $S[i\dd j]$, we return all $m$ runs that contain $S[i\dd j]$ as a substring in $\cO(m + \log n)$ time using the interval tree. 
By iterating over these runs, 
we check if $S[i\dd j]$ is periodic (\cref{def:periodic}) and, if so, find its smallest period. 
The query time is thus $\cO(m + \log n)$; in real-world strings, each position is expected to be in a few runs, so $m$ is expected to be small.

\paragraph*{The \textsc{AperiodicResilient} function.}~ 
Let $R $ be an aperiodic substring of~$S$ that satisfies $|\occ_S(R)|\geq  \tau+k$. Function \textsc{AperiodicResilient} uses a greedy algorithm (see  \cref{alg:greedy}) to compute the maximal number of occurrences of $R$ that may be lost if $k$ letters of $S$ are substituted  with a letter~$\#\notin\Sigma$.
\cref{lem:aperiodic} proves that, for any aperiodic substring $R$ of~$S$, such a substitution could result in the loss of \emph{at most two occurrences} of~$R$.
We refer to each of the $k$ substitutions, with $\#\notin\Sigma$, in our budget as a \emph{token} and we say that a token \emph{destroys} one (or more) occurrence(s) of $R$ if the frequency of $R$ is reduced  by one (or more) as a result of (using) this token.
Our algorithm is greedy in the sense that it first computes the tokens needed for the destruction of as many pairs of overlapping occurrences as possible spending one token per pair of overlapping occurrences (Lines~\ref{alg:greedy:line4}-\ref{alg:greedy:line9}).
Note that an occurrence of~$R$ may overlap with at most two other occurrences of $R$ (the one that precedes it and the one that succeeds it) and in this case only one pair of overlapping occurrences can be destroyed by a single token (Line~\ref{alg:greedy:line8}).
The algorithm then uses the remaining tokens to destroy single occurrences (Line~\ref{alg:greedy:line10}).
Clearly, \cref{alg:greedy} takes linear time in the number of occurrences, if they are given sorted. 

\begin{algorithm}[t]\footnotesize
\caption{Compute the number of occurrences of a length-$\lambda$ string that can be destroyed with~$t$ tokens in a sorted list $\cO$ of occurrences}
\begin{algorithmic}[1]
    \Function{\textsc{Greedy}}{$\cO, t, \lambda$}
    \State $k'' \gets 0$ \label{alg:greedy:line2} \Comment{Max num.~of tokens that can each destroy 2 distinct occs}
     \State $i \gets 0$
        \While{$i < |\cO|$ and $k''<t$} \Comment{Iterate over all occurrences}        
        \label{alg:greedy:line4}
            \If{$i\leq |\cO|-2$ and $\cO[i+1]-O[i]\leq \lambda-1$ } 
                \State $k'' \gets k'' + 1$
                \State $i\gets i+1$ \Comment{Skip the $(i+1)$-st occurrence of $\mathcal{O}$} \label{alg:greedy:line8}
            \EndIf
             \State $i\gets i+1$ \label{alg:greedy:line9}
        \EndWhile    
        \State $k' \gets t - k''$     \Comment{Remaining number of tokens} \label{alg:greedy:line10}
        \State {\bf return} $k'+2k''$ \Comment{Max number of occs that can be destroyed}
    \EndFunction
\end{algorithmic}\label{alg:greedy}
\end{algorithm}

\cref{alg:handle-aperiodic} presents our implementation of \textsc{AperiodicResilient}.
For an aperiodic substring $R=S[i\dd j]$, we compute its  frequency $|\occ_S(R)| =r_u-l_u+1$ (Line~\ref{alg:aperiodic:linef}), where $[l_u,r_u]$ is the \textsf{SA} interval stored for node $u$ with $\textsf{str}(u)= R$. Note that $S[i\dd j]$ is   only  considered if its frequency is greater than or equal to $\tau +k$ (see Phase~1). 
If the frequency  is greater than or equal to $\tau + 2k$ then $R$ is $(\tau,k)$-resilient (Lines~\ref{alg:aperiodic:line2}-\ref{alg:aperiodic:line3});  this is because  at most $2k$ occurrences of $R$ can be destroyed by $k$ tokens (\cref{lem:aperiodic}), and hence at least $\tau$ occurrences survive.
Otherwise (Lines~\ref{alg:aperiodic:line6}-\ref{alg:aperiodic:line11}), we employ the greedy algorithm (\cref{alg:greedy}) after retrieving the sorted list of the occurrences of $R$ in $S$ (Line~\ref{alg:aperiodic:line9}). 
Recall that we maintain  
AVL trees that allow us to access the  sorted list of occurrences for each node $u$ considered by the DFS on $\textsf{ST}(S)$.
Finally, we check whether the remaining number of occurrences (after \cref{alg:greedy} has been applied) suffices for $R$ to be $(\tau,k)$-resilient and return \textsf{YES} if that is the case and \textsf{NO} otherwise (Lines~\ref{alg:aperiodic:line10}-\ref{alg:aperiodic:line11}).
Applying the greedy algorithm (\cref{alg:greedy}) in Line~\ref{alg:aperiodic:line9} costs an additional $\cO(\tau + k)$ time. 

\begin{algorithm}\footnotesize
\caption{Check if an aperiodic substring of $S$ is $(\tau,k)$-resilient, given access to its sorted list $\cO$ of occurrences in $S$}\label{alg:aperiodicsurvive}
\begin{algorithmic}[1]
\Function{AperiodicResilient}{$i,j, \tau, k, \textsf{SA}, l_u, r_u, \cO$}
     \State $|\occ_S(S[i\dd j])| \gets r_u - l_u + 1$\label{alg:aperiodic:linef}
    \If{$|\occ_S(S[i\dd j])| \geq \tau + 2k$} \label{alg:aperiodic:line2} 
        \State \Return \textsf{YES} \label{alg:aperiodic:line3}
    \Else \Comment{In this case, $|\occ_S(S[i\dd j])| =\cO( \tau + k)$} \label{alg:aperiodic:line6}
        \If{$|\occ_S(S[i\dd j])| - \textsc{Greedy}(\cO, k, j-i+1) \geq \tau$} \label{alg:aperiodic:line9}
            \State \Return \textsf{YES} \label{alg:aperiodic:line10}
        \EndIf
        \State \Return \textsf{NO}\label{alg:aperiodic:line11}
    \EndIf
\EndFunction
\end{algorithmic}\label{alg:handle-aperiodic}
\end{algorithm}

We thus obtain \cref{lem:aperes}. 
\begin{restatable}[\textsc{AperiodicResilient}]{lemma}{lemaper}\label{lem:aperes}
    Given the occurrences of an aperiodic substring $P$ of a string $S$ sorted, we can check whether $P$ is $(\tau,k)$-resilient in $\cO(\tau+k)$ time.
\end{restatable}

\paragraph*{The \textsc{PeriodicResilient} function.}~Let $P$ be a periodic substring of $S$.
We say that an occurrence $S[i\dd j]$ of $P$ \emph{belongs} to a run $R = S[x \dd y]$ of $S$ when $\textsf{per}(R)=\textsf{per}(S)$ and $x \leq i \leq j \leq y$.
By \cref{lem:periodic}, we know that a single token can destroy occurrences of $P$ that belong to \emph{at most two different runs} of $S$.
Hence, \textsc{PeriodicResilient} relies on an efficient non-trivial preprocessing of $S$ to retrieve (at most) $\tau+2k$ runs of $S$ with period $\textsf{per}(P)$ that have $P$ as a substring.
If the number of such runs is greater than or equal to $\tau + 2k$, then we cannot destroy sufficiently many occurrences of $P$, and thus $P$ is $(\tau,k)$-resilient; otherwise, we employ a more involved version of the greedy algorithm discussed before to check whether $P$ is $(\tau,k)$-resilient.
We start by describing this preprocessing (and our practical implementation of it) before providing the full implementation of the \textsc{PeriodicResilient} function. 

\paragraph*{Preprocessing.} Let a periodic substring $P$ of $S$ with fewer than $\tau+2k$ runs to which its occurrences belong. The set $\mathcal{H}_P$ is defined as the set of all such runs.
For any periodic substring~$Q$ that does not satisfy this condition, we set $\mathcal{H}_Q := \varnothing$.
Recall that any run of $S$, being a fragment of $S$ (see \cref{def:run}), can be represented in $\cO(1)$ space by its endpoints.
We can thus represent $\mathcal{H}_P$ in $\cO(\tau+k)$ space.
The following lemma encapsulates a data structure that, given a periodic substring~$P$ of $S$, efficiently computes $\mathcal{H}_P$.

\cref{lem:pst} encapsulates a restricted variant of priority search trees~\cite{DBLP:journals/siamcomp/McCreight85}, which we use in the proof of \cref{lem:clusters}.

\begin{lemma}[\hspace{0.2mm}\cite{DBLP:journals/siamcomp/McCreight85}]\label{lem:pst}
Given $n$ points in $[1,n] \times [1,n]$, we can construct in $\cO(n \log n)$ time and $\cO(n)$ space, a data structure that supports the following type of queries (known as 3-sided range reporting queries) in $\cO(\log n + |\textsf{output}|)$ time: given $a,b,c,d \in \mathbb{Z} \cup \{-\infty, \infty \}$, such that $\{a,b,c,d\} \cap \{-\infty , \infty\} \neq \emptyset$, report all the points in $(a,b) \times (c,d)$.
\end{lemma}

\begin{restatable}[\textsc{ComputeHP}$(i,j,\tau,k)$]{lemma}{runsds}\label{lem:clusters}
For any string $S$ of length $n$ and integers $\tau,k \in [1,n]$, we can construct, in $\cO(n\log n)$ time using $\cO(n)$ space, a data structure that, for any periodic substring $P$ of $S$, computes the set $\mathcal{H}_P$ in $\cO(\log n + \tau + k)$ time.
\end{restatable}
\begin{proof}
We first describe the preprocessing and then the query.

\paragraph*{Preprocessing.} We compute all runs of $S$ in $\cO(n)$ time using the algorithm of~\cite{DBLP:conf/icalp/Ellert021}.
We then group the runs by Lyndon root in $\cO(n)$ time~\cite[Theorem 6]{DBLP:journals/tcs/CrochemoreIKRRW14}.
Each group is further divided into subgroups according to the number of Lyndon root occurrences in the runs; these subgroups are stored in descending order with respect to this number.
Creating these subgroups and sorting them can be done in $\cO(n)$ time using radix sort since the numbers to be sorted are from $[1,n]$.
A subgroup containing runs with Lyndon root $V$ such that $V$ occurs $t$ times in each run is labeled with $(V,t)$.
For each subgroup $(V,t)$, we create a two-dimensional grid $\mathcal{G}_{(V,t)}$ where we insert, for each run $V_1\cdot V^t\cdot V_2$ in $(V,t)$, the point $(|V_1|,|V_2|)$ --- note that we may have a multiset of points.
Further, for each subgroup $(V,t)$, after reducing to rank space, we insert each point of $\mathcal{G}_{(V,t)}$ into a priority search tree $\mathcal{T}_{(V,t)}$ (cf.~\cref{lem:pst}).

Since we have no more than $n$ runs in $S$ in total~\cite{DBLP:journals/siamcomp/BannaiIINTT17}, we can construct all such trees~\cite{DBLP:journals/siamcomp/McCreight85} in $\cO(n \log n)$ time using $\cO(n)$ space.
Each tree supports $3$-sided range reporting queries in $\cO(\log n + f)$ time, where $f$ is the number of points in the output~\cite{DBLP:journals/siamcomp/McCreight85}.
We also preprocess $S$, in $\cO(n)$ time for constant-time periodic extension queries (cf.~\cref{the:2-period}) 
and constant-time minimal rotation queries~\cite{DBLP:conf/cpm/Kociumaka16} (that return the lexicographically smallest rotation of an arbitrary fragment of~$S$). The preprocessing time is $\cO(n\log n)$ and the space usage is $\cO(n)$.

\paragraph*{Querying.} Suppose that we are given a periodic substring $P = V_1\cdot V^h\cdot V_2$ of $S$, where~$V$ is the lexicographically smallest rotation of $P[0 \dd \textsf{per}(P)-1]$ and $|V_1|,|V_2| \leq \textsf{per}(P)$ (see Definition 4).

The period of $P$ of $S$ can be computed in $\cO(1)$ time using a periodic extension query.
The Lyndon root $V$ of $P$ can then be computed in $\cO(1)$ time using a minimal rotation query.

Recall that we would like to retrieve $\mathcal{H}_P$:
a set of (less than $\tau+2k$) runs such that each contains some occurrence of $P$ and has period $\textsf{per}(P)$.
If, during the course of the query procedure, we compute $\tau+2k$ such runs, we terminate the algorithm and return $\emptyset$.
In the remainder of the proof, we assume that this is not the case.
We first retrieve all runs in subgroups $(V,t)$ with $t\geq h+2$. 
Each of these runs contains at least one occurrence of $P$.

\setcounter{example}{7}
\begin{example}
Let $P=\texttt{bc}(\texttt{abc})^2\texttt{ab}$ with Lyndon root $V=\texttt{abc}$ and $h=2$.
The run $(\texttt{abc})^3$ does not have any occurrence of $P$.
The run $(\texttt{abc})^4$ has an occurrence of $P$ at position 1.
\end{example}

We retrieve all $f$ runs $U_1\cdot V^{h+1}\cdot U_2$ from $(V,h+1)$ such that either
$|V_1|\leq |U_1|$ or
$|V_2|\leq |U_2|$.
These runs are retrieved by computing the points of $\mathcal{G}_{(V,h+1)}$ in the union of $[|V_1|, \infty) \times (-\infty,\infty)$ and
$(-\infty, |V_1|-1] \times [|V_2|,\infty)$; these points can be reported using $\mathcal{T}_{(V,h+1)}$ in time $\cO(\log n + f)$.

Finally, we retrieve all $f'$ runs from $(V,h)$ such that
$|V_1|\leq |U_1|$ and $|V_2|\leq |U_2|$. We employ a range reporting query using $\mathcal{T}_{(V,h)}$ to compute the points that lie in the quarterplane $[|V_1|, \infty) \times [|V_2|, \infty)$ in $\mathcal{G}_{(V,h)}$.
This takes $\cO(\log n + f')$ time.

The total query time is $\cO(\log n +  \tau + k)$.
\end{proof}

The implementation of \cref{lem:clusters} is highly unlikely to be efficient in practice; see our previous discussion on \cref{the:2-period}. In practice, however, we can settle for a \emph{simple} $\cO((\tau +k)\log n)$-time algorithm as the total number of times we call the function \textsc{PeriodicResilient} in real-world strings is very small. In our experiments, this number was about $n/7000$ even for the smallest tested $\tau$ and $k$ values; specifically, it was about $29\cdot 10^3$ for $n=2\cdot 10^8$.

To realize our simple algorithm, we will view each run in $\mathcal{H}_P$ as a cluster of occurrences of $P$, and from now on we will be abusing the definition of $\mathcal{H}_P$ as follows: 
each run in $\mathcal{H}_P$ is represented by a pair $(f,m)$, where $f$ is the first occurrence of $P$ in the run and $m \geq 1$ is the total number of occurrences of $P$ in this run.
Hence, if $ (f,m) \in \mathcal{H}_P$ then there is a run in $S$ with period $\textsf{per}(P)$ such that the set of the starting positions of occurrences of~$P$ in $S$ that lie within this run is $\{f + i \cdot \textsf{per}(P): i \in [0,m)\}$. \cref{alg:clusters} presents our practical implementation of the preprocessing. 
It considers the occurrences of $P$ in $S$ in the left-to-right order and efficiently merges them into runs.

During the bottom-up traversal of  the suffix tree $\textsf{ST}(S)$ and at node $u$ such that $\textsf{str}(u) = P$ and $P$ is periodic,  if $P= S[i\dd j]$ then \cref{alg:clusters} takes as input the set $\cO$ of occurrences of $P$ in $S$ sorted from the left to the right.
 The set $\cO$ is then partitioned  in inclusion-maximal sequences of occurrences such that the occurrences in each sequence are at distance $p=\textsf{per}(S[i\dd j])$ or there exists a single occurrence of $P$ (Line~\ref{alg:clusters:line4}).
 We only need to compute at most $\tau+2k$ such sequences.
 If retrieving more sequences is possible, then $P$ is ($\tau,k$)-resilient in $S$ due to \cref{lem:periodic}.

 Otherwise, we construct $\mathcal{H}_P$ (Line~\ref{alg:clusters:line5}). This can be done  in $\cO((\tau +k)\log n)$ time.
 Before explaining this time bound, observe that, if $P$ occurs at position~$i'$ then  all occurrences of~$P$ that start in  $[i' ,i' + \textsf{lcp}(i',i'+\textsf{per}(P))]$  belong to one run.
 We achieve the stated time bound by iterating over the occurrences in $\cO$ in increasing order as follows, starting from the first such occurrence. 
 We process the $i$-th occurrence (if we have to) as follows; let $f$ be the corresponding starting position.
 We compute $m = \lfloor \textsf{lcp}(f ,f +  \textsf{per}(P))/\textsf{per}(P)  \rfloor +1$  and insert~$(f,m)$ to $\mathcal{H}_P$.
 The computed $m$ occurrences of $P$ belong to one run and can be skipped.
 We thus proceed to processing the $(i+m)$-th occurrence of $P$.
 Retrieving the starting position of the $j$-th occurrence for any~$j$ can be done in  $\cO(\log n)$ time from the sorted list $\cO$ (maintained as an AVL tree).
 The described process is repeated $\cO(\tau +k)$ times.
 Therefore the time complexity is $\cO((\tau +k)\log n)$.

\setcounter{algorithm}{2}
\begin{algorithm}\small
\caption{Create clusters for periodic fragment $P=S[i\dd j]$, given random access to $P$'s sorted list $\cO$ of occurrences in $S$}\label{alg:createclusters} 
\begin{algorithmic}[1]
\Function{CreateClusters}{$i,j,p, k,\tau, \cO$}   
     \State Construct a partition $\mathcal{P}$ of $\cO$, such that $|\mathcal{P}| < \tau+2k$,  in every set  $C\in \mathcal{P}$, $|C|=1$ or any two consecutive elements of $C$ are at distance $p=\textsf{per}(S[i\dd j])$ and $C$ is inclusion-maximal. \label{alg:clusters:line4}

     \If{ $\sum_{C\in \mathcal{P}} |C| < |\occ_S(S[i\dd j]|$} 
      \State \Return $\mathcal{H}_P =\varnothing$ 
     \Else \State \Return $\mathcal{H}_P = \{( C[0], |C|)  \text{ for all } C \in  \mathcal{P} \}$ \label{alg:clusters:line5}
     \EndIf
    \EndFunction 
    \end{algorithmic}\label{alg:clusters}
\end{algorithm}

We give a high-level description of function \textsc{PeriodicResilient}; see \cref{alg:psurvive} for pseudocode. 
\textsc{PeriodicResilient} first computes $\mathcal{H}_P$ ($|\mathcal{H}_P|< \tau+2k$)  either by \cref{lem:clusters} or by \cref{alg:clusters}.  
We start with two key concepts: 

\begin{enumerate}
    \item The maximum number of occurrences of $P$ that can be destroyed using a single token, denoted by $\alpha := \lceil \frac{|P|}{\textsf{per}(P)} \rceil$.
\begin{example}
Let $P=\texttt{aabaaba}$ and the run $\texttt{aabaabaabaaba}$ that contains 3 occurrences of $P$.
We have $\alpha = \lceil \frac{|P|}{\textsf{per}(P)} \rceil=3$ because all occurrences can be destroyed with a single token: $\texttt{aabaab\#abaaba}$.   
\end{example}

\item  The subset $\mathcal{H}^{(1,2)}_P$ of $\mathcal{H}_P$ defined as follows:
\[\mathcal{H}^{(1,2)}_P = \{(f,m) \in \mathcal{H}_P:  m \bmod \alpha =  1 \text{ OR } 2 \}.\]  
\end{enumerate}
These runs are treated specially by \textsc{PeriodicResilient}. 

Function \textsc{PeriodicResilient} maintains two counters $t$ and $d$, where~$t$ is the number of available tokens (initially, $t:=k$), and $d$ is the maximum number of occurrences   (initially, $d:=0$) that may be destroyed using $k-t$ tokens.

The algorithm has three main stages: 

In the first stage, we only want to spend a token if it results in the destruction of $\alpha$ occurrences of $P$. We process the set of runs $\mathcal{H}_P$ in the left-to-right order and, for each run $(f, m) \in \mathcal{H}_P$, we compute $\lfloor m/\alpha \rfloor$, which equals the maximum number of tokens that we can use on this run such that each token destroys $\alpha$ (unique) occurrences of $P$.

Specifically, while we have tokens available, we process $(f, m)$ as follows: first, we use $y:=\min\{t, \lfloor m/ \alpha \rfloor\}$ tokens on this run, destroying a total of $y\cdot \alpha$ occurrences of $P$. The counters $d$ and $t$ are updated accordingly: $d := d + y\cdot \alpha$ and $t := t -  y$.
\begin{enumerate}
   \item If $m \bmod \alpha \geq 3$, 
   $m \bmod \alpha$ occurrences of $P$ are yet to be destroyed. Their number is inserted into a global max-heap $\mathcal{R}$ along with its 
   counter that is incremented every time the number is encountered. These remaining occurrences are to be processed in the second stage.
    \item If $1 \leq m \bmod \alpha \leq 2$, we set $\mathcal{H}^{(1,2)}_P := \mathcal{H}^{(1,2)}_P \cup \{(f,m)\}$ ($\mathcal{H}^{(1,2)}_P$ is initially empty). The remaining at most two occurrences of $P$ are processed in the third stage.
\end{enumerate} 

After the first stage is finished, if $\occ_S(P) - d \geq \tau$ (that is, $P$ is still $\tau$-frequent) and $t > 0$ (that is, we still have tokens available), we proceed to the second stage. In this stage, we use the max-heap~$\mathcal{R}$ in order to maximize the number of occurrences destroyed by each token.
In particular, if the maximum element in $\mathcal{R}$ is $r$ with counter~$c_r$, we use $\min\{t, c_r\}$ tokens to destroy $r$ occurrences of $P$ with each of them.
(Observe that by construction $r <\alpha$.)

We next proceed to the third stage if $\occ_S(P) - d \geq \tau$ and $t > 0$.
In this stage, we process the runs of $\mathcal{H}^{(1,2)}_P$; each of them has one or two surviving occurrences of $P$.
Due to \cref{lem:periodic}, any position in $S$ is contained in at most two runs of $\mathcal{H}_P$.
We refer to a run $(f,m) \in \mathcal{H}^{(1,2)}_P$ as type one or type two depending on the value of $m\bmod \alpha$.
We can assume that the runs in~$\mathcal{H}^{(1,2)}_P$ are sorted with respect to their starting positions (due to how they are computed).
We process the remaining occurrences of $P$, which all belong to runs in $\mathcal{H}^{(1,2)}_P$, using \cref{alg:greedy}: this algorithm first destroys pairs of overlapping occurrences and then non-overlapping singleton occurrences, similarly to the aperiodic case.

Finally, the function outputs whether $P$ is $(\tau,k)$-resilient.

\begin{example}
Consider $\alpha=10$ and that for all runs $(f,m) \in \mathcal{H}_P$, $m \in \{61,60,43, 23,12,7\}$.
The max-heap $\mathcal{R}$ contains $[7,1],[3,2]$: 
$1$ occurrence as the remainder from $7$
and $2$ occurrences as the remainder from $43$ and~$23$. 
The remainders from $61$ and $12$ are processed by \cref{alg:greedy}.
\end{example}

$\mathcal{H}_P$ is computed in $\cO(\log n+\tau+k)$ time (\cref{lem:clusters}) or $\cO((\tau+k)\log n)$ time (\cref{alg:clusters}).
For each of the $\cO(\tau+k)$ runs in $\mathcal{H}_P$, we spend $\cO(\log n)$ time for a constant number of max-heap operations. Moreover, $\cO((\tau+k)\log n)$ time is needed for retrieving and deleting elements from $\mathcal{R}$ and $\cO(\tau+k)$ time analogously to \cref{alg:greedy}. We obtain the following result.

\begin{restatable}[\textsc{PeriodicResilient}]{lemma}{lemper}\label{lem:peres} 
    After an $\cO(n \log n)$-time preprocessing of a string $S$ of length $n$,
    we can check whether any given periodic substring $P$ of $S$ 
    is $(\tau,k)$-resilient in $\cO((\tau+k) \log n)$ time. 
\end{restatable}

\begin{proof}
A pseudocode implementation of our algorithm is provided as \cref{alg:psurvive}.
The only lines that do not take $\cO(1)$ time are decorated with their time complexity. Recall that the preprocessing takes $\cO(n\log n)$ time. Thus, the time complexity follows from that $\mathcal{H}_P$ is computed in $\cO(\log n+\tau+k)$ time
(\cref{lem:clusters}) and that each while loop is executed $\cO(\tau+k)$ times, each of which needs $\cO(\log n)$ time. 

\begin{algorithm}\footnotesize
\caption{Check if a periodic substring of $S$ is $(\tau,k)$-resilient}\label{alg:psurvive}

\begin{algorithmic}[1]  
\Function{PeriodicResilient}{$i,j, l_u,r_u ,\tau, k,
\textsf{per}(S[i\dd j])$}
    \State $\mathcal{H}_P \gets \textsc{ComputeHP}(i,j,\tau,k)$ \Comment{See \cref{lem:clusters}; $\cO(\tau+k+\log n)$ time}
    \If{$\mathcal{H}_P = \varnothing$}
        \State {\bf return} \textsf{YES}
    \EndIf
    \State $\occ \gets r_u - l_u + 1$ \label{alg:periodic:line2}
    \State $\alpha \gets \lceil \frac{j-i+1}{\textsf{per}(S[i\dd j])} \rceil$ \Comment{Max no.~of occs that a token can destroy}
    \State  $\mathcal{H}_P^{(1,2)}\gets \varnothing$ \Comment{A subset of $\mathcal{H}_P$}
    \State $d\gets 0$ \Comment{Destroyed occurrences} 
    \State $t \gets k$ \Comment{Available tokens}
    \State $\mathcal{\mathcal{R}} \gets$ empty max-heap of positive integers  \Comment{With multiplicity counters}
    \State $(f,m) \gets$  the first   element in $\mathcal{H}_P$
    
    \While{ $(\occ -d \geq \tau) ~\text{and}~  (t>0) ~\text{and}~ (f,m) \neq nil$}    \Comment{Iterate over $\mathcal{H}_P$}
            \State $y \gets \min \{ t,\lfloor m/\alpha \rfloor \}$
            \State $d \gets d+y\cdot\alpha$
            \State $t \gets t- y $
            \If{$m \bmod \alpha \geq 3$} 
                \If {$m \bmod \alpha  \notin \mathcal{R}$ }
                \State Insert 
                $[m\bmod \alpha, 1]$ into $\mathcal{R}$ \Comment{$\cO(\log n)$ time}
                \Else  ~~Increment the  counter of $m \bmod \alpha$ in $\mathcal{R}$ by one \Comment{$\cO(\log n)$ time}
                \EndIf
        \ElsIf{$m \bmod \alpha = 1$~\text{or}~$m \bmod \alpha = 2$ }
            \State Insert $(f,m)$ into $\mathcal{H}_P^{(1,2)}$
        \EndIf    
        \State $(f,m) \gets$ the next element of $\mathcal{H}_P$ \Comment{$nil$ if we have processed all of  $\mathcal{H}_P$}\label{alg:periodic:line27}
    \EndWhile

    \If{$\occ -d \geq \tau ~\text{and}~ t>0$} 
         \While{$\mathcal{R} \neq \varnothing ~\text{and}~ \occ -d   \geq \tau ~\text{and}~ t > 0$}  \Comment{Iterate over $\mathcal{R}$} \label{alg:periodic:line21}
            \State  $[r,c_r] \gets  \mathcal{R}.\textsf{top}()$
            \State Pop $\mathcal{R}.\textsf{top}()$ \Comment{$\cO(\log n)$ time}
             \State $y \gets \min \{ t,c_r \}$
            \State $d \gets d+y\cdot r$ \label{alg:periodic:line24}
            \State $t \gets t- y$ \label{alg:periodic:line25}
        \EndWhile \label{alg:periodic:line26}
    \EndIf
         \State $\kappa \gets 0$ \label{alg:periodic:line30}\Comment{Max number of tokens that can each destroy two distinct occs} 
    \If{$\occ -d \geq \tau ~\text{and}~ t>0$} 
        \State $z \gets -\infty$ \Comment{Position of rightmost processed and surviving occ} 
         \ForAll{$(f,m) \in \mathcal{H}_P^{(1,2)}$} \Comment{In increasing order}
            \If{$f - z < j-i+1$} \Comment{Overlap between two runs} \label{alg:periodic:line34}
                 \State $\kappa \gets \kappa + 1$ \Comment{Destroy a pair of occs} \label{alg:periodic:line35}
                \If{$m \bmod \alpha = 2$}\label{alg:periodic:line36}
                    \State $z \gets f +(m-1) \textsf{per}(S[i\dd j])$\label{alg:periodic:line37}
                \EndIf
            \ElsIf{$m \bmod \alpha = 2$}\label{alg:periodic:line38}
                    \State $\kappa \gets \kappa + 1$\label{alg:periodic:line39} \Comment{Destroy the remaining pair of occs of $(f,m)$}
            \Else  \label{alg:periodic:line40}
            ~~\State $z \gets f +(m-1) \textsf{per}(S[i\dd j])$\label{alg:periodic:line41}
            \EndIf
        \EndFor
    \EndIf

        \If{$\occ - d - \min\{\kappa,t\} - \min\{\textsf{occ}-d,t\} \geq \tau$}\label{alg:periodic:line42}
        \State {\bf return} \textsf{YES}
        \EndIf
        \State {\bf return} \textsf{NO}
\EndFunction
\end{algorithmic}
\end{algorithm}

The rest of the proof is for correctness.
Our algorithm first constructs $\mathcal{H}_P$ and returns \textsf{YES} if it is $\varnothing$.
In the remainder of this section, we assume that this is not the case.

Our algorithm maintains the number $t$ of available tokens and a counter $d$ of the maximum number of occurrences of $P$ that can be destroyed with the number of tokens that have been used so far.
In Lines~\ref{alg:periodic:line2}-\ref{alg:periodic:line27}, we use tokens to destroy $\alpha$ occurrences of $P$ as long as this is possible, while 
populating
max-heap $\mathcal{R}$ and 
$\mathcal{H}^{(1,2)}_P$.
Then, we start popping elements from $\mathcal{R}$
(Lines~\ref{alg:periodic:line21}-\ref{alg:periodic:line26}).
If $\mathcal{R}.\textsf{top}()= [r,c_r]$, we use $y=\min\{t, c_r\}$ tokens to destroy a total of $y \cdot c_r$ occurrences of~$P$; the values of $d$ and $t$ are updated accordingly (Lines~\ref{alg:periodic:line24}-\ref{alg:periodic:line25}).
Elements are popped from~$\mathcal{R}$ as long as $|\occ_S(P)|-d \geq \tau$ and $ t>0$ (Line~\ref{alg:periodic:line21}).
If $\mathcal{R}$ is empty, $|\occ_S(P)|-d \geq \tau$, and $t>0$, then the remaining (one or two) occurrences of $P$
in each of the runs of type one and two are processed with a greedy procedure, analogously to 
Algorithm 1.
This is done as follows (Lines~\ref{alg:periodic:line30}-\ref{alg:periodic:line41}). 
We process the elements of $\mathcal{H}^{(1,2)}_P$ in a left-to-right manner, and maintain a counter~$\kappa$ for the maximum number of pairs of remaining occurrences that can be destroyed.
When we process a run $(f,m)$, we check whether this run overlaps with an occurrence of $P$ in an already processed run of  $\mathcal{H}^{(1,2)}_P$ that has not been destroyed yet:
\begin{itemize}
\item If that is the case, we increment $\kappa$ (Lines~\ref{alg:periodic:line34}-\ref{alg:periodic:line35}).
Additionally, if $m \bmod \alpha = 2$, we set the starting position of the rightmost occurrence that has been processed but has not been destroyed to $f +(m-1) \textsf{per}(S[i\dd j])$ (Lines \ref{alg:periodic:line36}-\ref{alg:periodic:line37}). 
\item Otherwise, when $m \bmod \alpha = 2$, we destroy the two remaining occurrences of $(f,m)$ (Lines \ref{alg:periodic:line38}-\ref{alg:periodic:line39}); in the remaining case when $m \bmod \alpha = 1$,  we set the starting position of the rightmost occurrence that has been processed but has not been destroyed to $f +(m-1) \textsf{per}(S[i\dd j])$ (Lines \ref{alg:periodic:line40}-\ref{alg:periodic:line41}).
\end{itemize}
Observe, that as we have $\textsf{occ}-d$ occurrences remaining when the greedy procedure commences (if it does), the total number of occurrences that we can destroy with the $t$ available tokens is $\min\{\kappa,t\} + \min\{\textsf{occ}-d,t\}$. Thus, the test of Line~\ref{alg:periodic:line42} allows us to conclude whether $P$ is $(\tau,k)$-resilient.

The correctness of \textsc{PeriodicResilient} is based on the fact that it first destroys the batches of occurrences of~$P$ of size $\alpha$; then it processes $\mathcal{R}$, destroying numbers of occurrences of $P$ greater than~2 (and less than $\alpha$) \emph{in decreasing order}, and then it processes the remaining occurrences (from $\mathcal{H}^{(1,2)}_P$) using a variant of 
Algorithm 1;
the correctness of the latter follows from 
\cref{lem:periodic}.    
\end{proof}

\paragraph*{Wrapping up.}~We obtain our main result by plugging \cref{the:2-period} and \cref{lem:aperes,lem:peres} to \cref{lem:reduction}.
However, to shave a $\log n$ factor, we replace the max-heap implementation in the proof of \cref{lem:peres} with \emph{global sorting} and \emph{avoid recomputing the set} $\mathcal{H}_P$ in each step of a binary search along an edge when \textsc{PeriodicResilient} is called; see below for the proof.

\begin{restatable}{theorem}{thmmain}\label{the:main}
For any string $S\in \Sigma^n$, with $|\Sigma|\geq 4$, and integers $\tau, k \in [1,n]$, the \textsc{$(\tau,k)$-Resilient Pattern Mining} problem can be solved in $\cO(n \log n)$ time using $\cO(n)$ space.
\end{restatable}
\begin{proof}
    In Phase 1, we perform a DFS of $\textsf{ST}(S)$, using AVL trees to maintain the sorted list of occurrences of $\textsf{str}(u)$ for each node $u$ of $\textsf{ST}(S)$; this requires $\cO(n\log n)$ time and $\cO(n)$ space; see \cref{app:AVL_trees}.
For each node of $\textsf{ST}(S)$ with at least $\tau + k$ descendants, we first call \textsc{IsPeriodic} which costs $\cO(1)$ time 
 (\cref{the:2-period}) 
and then call either \textsc{AperiodicResilient} or \textsc{PeriodicResilient}.
Each call to \textsc{AperiodicResilient} costs $\cO(\tau +k)$ time 
(\cref{lem:aperes}), 
while each call to function \textsc{PeriodicResilient} costs $\cO((\tau +k)\log n)$ time (\cref{lem:peres}). Therefore, the total time for Phase~1 is $\cO(n \log n)$.

In Phase 2, we perform a DFS of $\textsf{ST}(S)$ maintaining AVL trees as in Phase 1.
For each node $v$, such that edge $(u,v)$ is on the preliminary cut computed in Phase 1, we compute the sorted list of occurrences of $\textsf{str}(v)$ and then we perform binary search for this edge.
In total, we perform $\cO(n\log n/ (\tau+k))$ calls to each of the functions \textsc{IsPeriodic}, \textsc{AperiodicResilient}, and \textsc{PeriodicResilient}.
All calls to the first two functions take $\cO(n\log n)$ time and $\cO(n)$ space due to 
\cref{the:2-period} and \cref{lem:aperes}.
We show that  all calls to the function \textsc{PeriodicResilient} cost $\cO(n \log n)$ time and $\cO(n)$ space.
If we simply called \textsc{PeriodicResilient} each time during the binary searches, we would obtain an algorithm running in  $\cO(n \log^2 n)$ time  (see \cref{lem:peres}).
We achieve a better running time by (i) employing the following combinatorial lemma; and (ii) avoiding the usage of heaps and rather sorting globally.

\setcounter{lemma}{14}
\begin{lemma}\label{lem:combi}
Let $(u,v)$ be an edge of $\textsf{ST}(S)$, $\ell \in [\textsf{sd}(u),\textsf{sd}(v)-1)$ be an integer, and $P = \textsf{str}(v)[0 \dd \ell]$.
If $P$ is periodic and $\mathcal{H}_P$ does not only contain elements of the form $(\star, 1)$, then $\mathcal{H}_P = \mathcal{H}_{\textsf{str}(v)}$.
\end{lemma}
\begin{proof}
If $\textsf{per}(P) < \textsf{per}(\textsf{str}(v))$, then $P$ does not have any two occurrences in $S$ at distance $\textsf{per}(P)$ as otherwise $\textsf{per}(\textsf{str}(v))$ would be at most $\textsf{per}(P)$.
Hence, all elements of $\mathcal{H}_P$ are of the form $(\star, 1)$.

In the remaining case, we have $\textsf{per}(P) = \textsf{per}(\textsf{str}(v))$.
Note that any occurrence of $\textsf{str}(v)$ implies an occurrence of $P$ at the same position.
Since $\occ_S(P) = \occ_S(\textsf{str}(v))$, the statement follows.
\end{proof}

We are now ready to explain how we can perform all calls to \textsc{PeriodicResilient} during the binary searches of Phase~2 more efficiently than by simply employing \cref{lem:peres}.
Note that \cref{alg:psurvive}
(underlying 
\cref{lem:peres})
 only incurs logarithmic factors due to the call to 
 \cref{lem:clusters}
 and to inserting/popping elements to/from a max-heap.
The $\log n$ factor from each heap's operation can be avoided by sorting all elements to be inserted to heaps (in each step of the binary search) over all edges of the cut as a batch in $\cO(n + (\tau+k) \cdot n/ (\tau+k)) = \cO(n)$ time in total using bucket sort.
As we have  $\cO(\log n)$ batches (one for each step of the binary searches), the total time required for sorting is $\cO(n \log n)$.
Moreover, due to 
\cref{lem:combi}, we need to call 
\cref{lem:clusters}
at most once for each considered edge $(u,v)$ of $\textsf{ST}(S)$ as we can reuse the computed runs in $\mathcal{H}_{\textsf{str}(v)}$ (if $\textsf{str}(v)$ is periodic) for all prefixes of $\text{str}(v)$ of lengths more than $\textsf{sd}(u)$ and with period equal to $\textsf{per}(\text{str}(v))$.
Although there might be a prefix that is periodic with period $q < \textsf{per}(\text{str}(v))$, as such a prefix does not have any two occurrences at positions that are $q$ positions apart, we can treat it as an aperiodic substring using \cref{alg:handle-aperiodic}.
We thus obtain the main result of this work.
\end{proof}

\section{Substituting with Letters from the Alphabet}\label{sec:alphabet}

In this section, we show that we can effectively destroy the occurrences of any substring $P$ of $S\in \Sigma^n$
by using only letters from $\Sigma$, with $|\Sigma|\geq 4$,
it suffices to show that we can replace substitutions with $\#$s with substitutions with letters from $\Sigma$ in a way that does not create any new occurrences of $P$.
Before proving this (\cref{lem:sigma}), we prove a few auxiliary lemmas.

\setcounter{lemma}{6}
\begin{lemma}\label{lem:1occ}
Let $T=U\cdot P\cdot V$ be a string over an alphabet $\Sigma$, such that $|\Sigma|\geq 3$,
$P$ is a non-empty string occurring only once in $T$, and $|U|,|V| < |P|$.
There exists a string $P'$ over~$\Sigma$ such that $d_H(P,P')=1$ and $T'=U\cdot P'\cdot V$ has no occurrence of $P$.
\end{lemma}
\begin{proof}
    Towards a contradiction, suppose that $P$ occurs in $T'$, for all strings $T'=U\cdot P'\cdot V$ such that $d_H(P,P')=1$. Let $m=|P|$.
    Since $|\Sigma|\geq 3$ and we can hence choose to substitute with at least two different letters in each of the $m$ positions of $T$ in $[|U|, |UP|)$,
    there are at least $2m$ different strings $T'$ of the form $U\cdot P'\cdot V$ such that $d_H(P,P')=1$.
    We have $|T'|-m \leq 2m-2$ possible starting positions for $P$ in a string such a string $T'$.
    Since $2m-2<2m$, by the pigeonhole principle, there are two such strings, say $T'_1$ and $T'_2$, with an occurrence of $P$ starting at the same position, say $i$.
    This yields a contradiction, because the occurrence of $P$ must span the substitution in each of $T'_1$ and $T'_2$ while either the position of the substitution or the substituted letter is different.
\end{proof}

\begin{lemma}\label{lem:1occ-overlap}
Let $T=U\cdot V\cdot P$ be a string over an alphabet $\Sigma$, such that $|\Sigma|\geq 3$, $P$ is a non-empty string occurring only as a prefix and as a suffix of $V\cdot P$ in $T$, $|V| \in [1,|P|)$, and $|U| < |P|$.
There exists a string~$V'$ over $\Sigma$ such that $d_H(V,V')=1$ and $T'=U\cdot V'\cdot P$ has one occurrence of~$P$.
\end{lemma}

\begin{proof}
We first prove that in any string $T'=U\cdot V'\cdot P$ such that $V'$ is over $\Sigma$ and $d_H(V,V')=1$, $P$ does not have any occurrence ending  at a position in $[|UV|,|T'|)$, i.e., overlapping the occurrence of $P$ at position $|UV|$.
Let us fix any two such strings $V'$ and $T'$ and consider a substitution at position $j\in[|U|,|UV|)$ of $T$ transforming $T$ to $T'=U\cdot V'\cdot P$,
such that $T'[j]  = b$ and $T[j]  = a$, for some letters $a \neq b$ from $\Sigma$.
Further, suppose toward a contradiction that $P$ has an occurrence in $T'$ starting at some position $i\in [0, j]$, where $i+|P| > |UV|$. Due to the overlaps of the occurrences of $P$ in each of $T$ and $T'$, $P$ has periods $|V|$ and $|UV| - i$. Since~$P$ occurs at position $i \in [0, j]$ of $T'$ and at position $|UV|$ of $T$ and $T'$, we have
$T'[j] = P[j - i] = T[|UV| + j - i] = b$.
Further, since $T[j] = a$ and $P$ also has a period $|UV| - i$, we have $T[j + (|UV|-i)] = T[|UV| + j - i] = a$; a contradiction.

We now assume that each string $T'= U\cdot V'\cdot P$, where $V' \in \Sigma^*$ and $d_H(V,V')=1$ has an occurrence of $P$ ending at some position in $[|V|,|UV|)$. Let $m = |P|$. Since $|\Sigma| \geq 3$, we can choose to substitute with each of at least two different letters in each of the $|V|$ positions of $T$ in $[|U|,|UV|)$ and we can thus generate at least $2|V|$ different strings $T'$ of the form $U\cdot V'\cdot P$ such that $d_H(V,V')=1$.
However, we have $|V| < m$ possible ending positions for~$P$ in~$T'$. Since $2|V| > m > |V|$, we have two strings $T'_1$ and $T'_2$, with an occurrence of $P$ starting at the same position, say $i$.
This yields a contradiction, because the occurrence of $P$ starting at position $i$ in each of the strings $T'_1$ and $T'_2$ must span the respective substitution while either the position of the substitution or the substituted letter is different.
\end{proof}

\begin{lemma}\label{lem:2occs}
Let $T=U\cdot X\cdot V$ be a string over an alphabet~$\Sigma$, such that $|\Sigma|\geq 4$, 
$X$ has a substring~$P$ occurring only as a prefix and as a suffix of $X$, $P$ occurs $h\geq 2$ times in~$T$, $|X|<2|P|$, and $|U|,|V|<|P|$.
There exists a string $X'$, such that $d_H(X,X')=1$ and $P$ occurs $h-2$ times in $T'=U\cdot X'\cdot V$. 
\end{lemma}
\begin{proof}
Let $m:=|P|$. Further, let $i:= |UX| - m$ be the starting position of the second occurrence of $P$ in $T[|U| \dd |UX|-1]$, and $j:=|U|+m-1$ be the ending position of the first occurrence of $P$ in $T[|U| \dd |UX|-1]$.
Then, the overlap of the two occurrences of~$P$ that are fully contained in $X=T[|U| \dd |UX|-1]$ is precisely $F:=T[i\dd j]$.

\setcounter{fact}{1}
\begin{fact}\label{fact:forall}
For any position $x \in [i\dd j]$ of $T$, there exist $b,c \in \Sigma$ such that the string $T'$ obtained via a substitution at position~$x$ of~$T$ with either of $b$ or $c$ satisfies $\occ_{T'}(P) \cap [|U| , i)= \emptyset$.
\end{fact}
\begin{proof}

Consider a string $T'=T[0 \dd i-1]\cdot F'\cdot T[j+1 \dd |T|-1]$ with $d_H(F,F')=1$ and $T[x] = a \neq T'[x] \in \Sigma$.
Suppose that $P$ has an occurrence at a position $y \in [|U|, i)$ of $T'$.
Then, the assumed occurrence of $P$ in $T'$ and the occurrences of $P$ in $X$ imply that:
\begin{itemize}
\item $T'[|U| \dd |U|+x-y-1] = P[0 \dd x-y-1] = T'[y \dd x-1]$;
\item $T'[i + x - y + 1 \dd |UX|-1] = P[x-y+1 \dd m-1] = T'[x+1 \dd y+m-1]$.
\end{itemize}
See \cref{fact:persuf} for an illustration.

Without loss of generality, let us assume that $x-y \leq m - (x-y+1)$, that is, that the string in the second bullet point is longer than the string in the first bullet point (this case is depicted in \cref{fact:persuf}). 
Let
$Y=T'[x+1 \dd y+m-1] = P[x-y+1 \dd m-1]$.
The equalities above also imply that $Y$ is a suffix of $T'[x+1 \dd |UX|-1]= P [x-i+1 \dd m-1]$.
Therefore, $Y$ is both a prefix and a suffix of
\[Z := T[x+1 \dd |UX|-1] = P[x-i+1 \dd m-1],\]  implying that $Z$ has a period $i-y$.

Now, $T'[x \dd y+m-1]$ is a suffix of $P$. It is also a suffix of~$Z$ since $y+m-x \leq i-1+m-x = |Z|$.
Then, due to $x \geq i$ and our assumption that $x-y \leq m - (x-y+1)$, we have $i-y \leq x-y \leq y+m-x-1$
and hence $T'[x \dd y+m-1]$ has a non-trivial period~$i-y$.
This is only possible if $T'[x]=T[x+(i-y)]$.
Thus, since $\Sigma \geq 4$, the claim follows.
\end{proof}

\begin{figure}[htpb!]
\centering
\includegraphics[width=.7\linewidth]{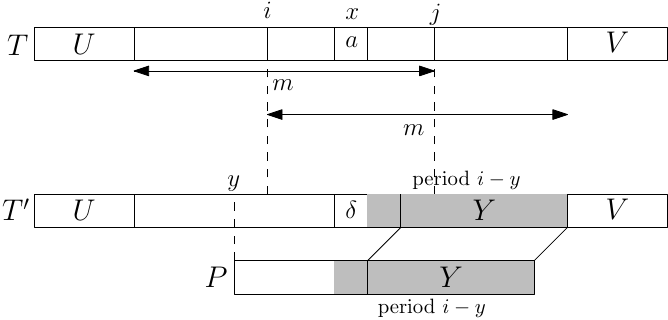}
\caption{The setting in \cref{fact:forall}, where fragments equal to $Z$ (with period $i-y$) are shaded.
In this setting, $P$ can only occur at position $y$ of $T'$ if letter $a$ at position $x$ is substituted with letter $\delta = T[x+(i-y)]$.}\label{fact:persuf}
\end{figure}

\begin{fact}\label{fact:exists}
There exists a string $T'=T[0 \dd i-1]\cdot F'\cdot T[j+1 \dd |T|-1]$ such that $d_H(F,F')=1$, $\occ_{T'}(P) \cap[0, |UX|) \subset \occ_{T}(P) \cap[0, |UX|)$.
\end{fact}
\begin{proof}
Towards a contradiction, suppose that this is not the case.
Since $|\Sigma| \geq 4$, due to \cref{fact:forall} we can choose at least two different letters to substitute each of the $|F|$ positions of $T$ in $[i,j]$ without creating a new occurrence of $P$ starting at a position in $[|U|,i)$.
We can generate at least $2|F|$ different strings $T'$ of the form specified in the claim's statement without any occurrence of~$P$ starting at a position in $[|U|,i)$.

In each of the constructed strings $T'$, each new occurrence of $P$ must then start at a position in $[|U|-|F| , |U|) \cup (i , j]$ of $T'$ and span the substitution in $F'$.
By the pigeonhole principle, we have two such strings $T'_1$ and $T'_2$, with an occurrence of $P$ starting at the same position, say $y$, in $[|U|-|F| , |U|) \cup (i , j]$, since the total size of this set is $2|F|-1$.
This yields a contradiction because the occurrences of $P$ at positions $y$ of $T'_1$ and $T'_2$ must span the substitution in each of $T'_1$ and $T'_2$ while either the position of the substitution or the substituted letter is different.
\end{proof}

To complete the proof, it suffices to note that a string $T'$, as specified in \cref{fact:exists}, has $h-2$ occurrences of $P$ since the substitution yielding string $T'$ from~$T$ destroys the two occurrences of $P$ that are fully contained in $T[|U| \dd |UX|-1]$, it does not destroy any other occurrence due to \cref{lem:aperiodic}, and it does not create any occurrence of $P$ due to \cref{fact:exists}.
\end{proof}

\begin{lemma}\label{lem:per}
Let $T$ and $T'$ be two strings in $(\Sigma \cup \{\#\})^*$; where~$\Sigma$ is an alphabet with $|\Sigma|\geq 3$ and $\# \notin \Sigma$, such that
$d_H(T, T')=1$, $T[i]\neq T'[i]=\#$ for some  $i\in[0,|T|)$, $|\occ_T(P)|- |\occ_{T'}(P)| = y>0$ for a periodic string $P$ of length~$m$,
and $|\occ_{T'[i-m+1 \dd i+m-1]}(\#)| = 1$.
We can substitute $T'[i]$ with a letter from $\Sigma$ without increasing the frequency of $P$.
\end{lemma}
\begin{proof}
We perform the substitution of $\#$ in $T'$ with any letter in $\Sigma \setminus \{T[i - \per(P)], T[i + \per(P)]\}$; let us denote the obtained string by $T''$.
This is possible because $|\Sigma| \geq 3$.
Any element $j \in \occ_{T''}(P) \setminus \occ_{T'}(P)$ must contain position $i$ and at least one of $i + \per(P)$ and $i - \per(P)$ (since $\per(P)\leq m/2$), and hence the period of $T''[j \dd j+m-1]$ does not divide $\per(P)$; a contradiction.
\end{proof}

We are now ready to put these lemmas together.

\begin{lemma}\label{lem:sigma}
Let $S$ be a string of length $n$ over $\Sigma$, with $|\Sigma|\geq 4$, and~$P$ be a substring of $S$.
If we can reduce the frequency of $P$ in~$S$ by $y \in \mathbb{Z}_+$ using $k$ substitutions with a letter $\#\notin\Sigma$, then we can also reduce the frequency of $P$ in~$S$ by at least $y$ using $k$ substitutions with letters from $\Sigma$.
\end{lemma}
\begin{proof}
Let $m=|P|$ and $\kappa \leq k$ be the \emph{smallest} integer such that we can use $\kappa$ letters $\#\notin\Sigma$ to destroy $y$ occurrences of $P$ in $S$.
Then, there exist $\kappa \leq k$ positions in $S$ where we can place $\#$s such that the distance between any pair of $\#$s is at least $m$ (and hence no two of them touch the same occurrence of $P$).
Consider the collection~$\mathcal{U}$ of all the sets of positions that satisfy the above property.
For a set $U \in \mathcal{U}$, let $B_U$ be the bit-string of length $n$ such that $B_U[i]$ is set if and only if $i \in U$.
Let $Y := \argmin_{U \in \mathcal{U}} B_U = \{i_1, \ldots, i_{\kappa} \}$, where $i_1 < i_2 < \cdots < i_{\kappa}$, and denote the set of starting positions of occurrences of $P$ in $S$ that the letter $\#$ at position $i_j$ destroys by $I_j := \occ_S(P) \cap (i_j-|P|, i_j]$.

We perform $\kappa$ substitutions to $S$ with letters from $\Sigma$, such that the $j$-the substitution destroys the occurrences of $P$ at all positions in $I_j$ (and no others), while it does not create any new occurrences of $P$.
By the minimality of $\kappa$ and $B_Y$, for each $i_j$, if $P$ is aperiodic, we do not try to destroy a single occurrence of $P$ that overlaps with an occurrence of $P$ in both sides.
Hence, for each position $i_j$, we want to destroy one of the following:
\begin{itemize}
    \item one occurrence of $P$ that does not overlap with any other occurrence of $P$ (this can be done due to \cref{lem:1occ}); or
    \item if $P$ is aperiodic, one occurrence of $P$ that overlaps with another occurrence of $P$ only on one of its two sides (this can be done due to \cref{lem:1occ-overlap}); or
   \item two overlapping occurrences of $P$ (this can be done due to \cref{lem:2occs}); or
    \item if $P$ is periodic, some number of overlapping occurrences of $P$ (this can be done due to \cref{lem:per}). 
\end{itemize}

For the remaining $k-\kappa$ substitutions that we can perform (if any), we substitute the letters of $S$ from left to right by any letter other than $P[0]$ -- we do not decrease the budget if we substitute a letter at a position in $Y$.
This guarantees that these remaining substitutions do not create a new occurrence of $P$ starting within the updated prefix.
\end{proof}

\section{Related Work}\label{sec:related}
There has been much  interest in resilience notions for many  fundamental problems. One such notion is \emph{sensitivity} (and average sensitivity) which, informally stated, measures the difference between the output of an algorithm after its input is perturbed by adding or removing \emph{a single} element. This notion has been employed on graph~\cite{sicomp23}, clustering~\cite{sensclust,sensclust2}, and learning~\cite{senslearn} problems and several algorithms based on it have been proposed. There are also resilience notions specifically developed for clustering~\cite{clustresnotion1,kddres,perturbationstability,DBLP:conf/approx/ChekuriG18}. 
For example,~ \emph{stability} in~\cite{clustresnotion1} is applied to graph clustering and requires  the interesting instances to be only those for which small perturbations in the data do not change the optimal partition of the graph. $\gamma$-\emph{Resilience}  in~\cite{kddres} is applied to $k$-clustering~\cite{kddres} and, informally, requires an algorithm to return \emph{similar}  solutions on any two inputs that are \emph{close}. 
Our resilience notion in \cref{def:resilient} differs from all existing ones in that it: (1) applies to substrings of a string; and (2) it considers a given number of possible changes to the string. 
Also, it differs from \emph{differential privacy}~\cite{diffpriv} which 
imposes a bound on how much the probability of an algorithm's output can change when applied to any two datasets that differ by an element. 
 Our work is somewhat related to works for mining frequent approximate sequential patterns~\cite{maxcfp,reputer,DBLP:conf/icdm/ZhuYHY07} and fault-tolerant patterns~\cite{DBLP:conf/kdd/PoernomoG09a}. 
The former works consider changes to the patterns while we consider changes to the input string. The latter work mines \emph{dense} sets from set-valued data, while we mine substrings using a different resilience notion. 

\section{Experimental Evaluation}\label{sec:experiments}

\paragraph*{Algorithms.} 
As no existing algorithm can deal with \RPM, we compared the suffix tree and enhanced suffix array implementations of the algorithm in \cref{sec:fast-algo}, \RPMST and \RPMESA, to \BASELINE an implementation of the algorithm in \cref{sec:Baseline}. Also, in our  clustering case study, we compared a clustering algorithm~\cite{DBLP:journals/tkde/WuZLGZFW23} using $(\tau,k)$-resilient substrings  mined by any of our algorithms as features to: (I) the same clustering algorithm that uses $\tau$-frequent substrings as features instead; and (II) 
a widely-used frequency-based clustering algorithm~\cite{DBLP:journals/bioinformatics/VingaA03}, which we will denote by \FC. 

\begin{table}[t]
\caption{Dataset characteristics and values of parameters used. The default values are in bold within brackets.}
\label{table:dataset}
\centering
\resizebox{1.01\columnwidth}{!}{
\begin{tabular}{|c||c|c|c|c|}  
 \hline
 {\bf Dataset} & {\bf Length} $n$ & {\bf Alphabet Size} $|\Sigma|$ & {\bf Frequency} $\tau$ & {\bf Number of Substitutions} $k$ \\ [0.5ex] 
 \hline\hline
 \dna & 209,715,200 & 16 & [$10^1$,$10^6$] ($\mathbf{10^4}$) & [$10^1$,$10^6$] ($\mathbf{10^2}$) \\
 \english & 209,671,447 & 225 & [$10^1$,$10^6$] ($\mathbf{10^4}$) & [$10^1$,$10^6$] ($\mathbf{10^2}$) \\
 \proteins & 209,715,200 & 25 & [$10^1$,$10^6$] ($\mathbf{10^4}$) & [$10^1$,$10^6$] ($\mathbf{10^2}$) \\
 \sources & 209,714,417 & 230 & [$10^1$,$10^6$] ($\mathbf{10^4}$) & [$10^1$,$10^6$] ($\mathbf{10^2}$) \\ 
 \xml & 209,715,200 & 96 & [$10^1$,$10^6$] ($\mathbf{10^4}$) & [$10^1$,$10^6$] ($\mathbf{10^2}$) \\ 
 \boost & $190,898,317$ & $88$ & [$4$,$50$] ($\mathbf{4}$) & [$4$,$50$] ($\mathbf{4}$) \\
 \wiki & $629,145,600$ & $194$ & [$20$,$70$] ($\mathbf{50}$) & [$20$,$70$] ($\mathbf{50}$) \\
 \hline
 \end{tabular}
 }
\end{table}

\paragraph*{Datasets.} 
We used $5$ real-world, benchmark  datasets from~\cite{pizzachili}; see \cref{table:dataset}. 
\dna contains genomic data, \english contains English text, \proteins protein data, \sources  
source code, and \xml bibliographic information on major computer science venues.  
We also considered two popular versioned datasets following~\cite{repcorpus}, which were also used in~\cite{rindex}: \boost, which consists of $10,000$ versions of the boost library in GitHub; and 
\wiki which consists of $9,271$ versions of Wikipedia’s English Einstein page. Each version was treated as a different string. 
As expected by its time complexity, \BASELINE was not able to handle these datasets. Thus, we report results on prefixes of \dna (the results on prefixes of the other datasets were analogous). We also used two genomic datasets,  
\EBOL and \COR, from~\cite{ligen} in our clustering experiments. Each dataset is a collection of viral genomes and thus has a ground truth clustering.    

\paragraph*{Setup.}~We report running time and peak memory usage.
We also report two quality measures. The first is 
the ratio between the size of the set of the $(\tau, k)$-resilient substrings and that of the set of the $\tau$-frequent substrings.
We refer to  this ratio as \RFR (for Resilient to Frequent Ratio). Since any $(\tau,k)$-resilient  substring is also $\tau$-frequent for the same $\tau$, \RFR measures the ratio of the $\tau$-frequent substrings that are also $(\tau,k)$-resilient.
The second measure is used for versioned datasets and is denoted by \LR.
For an arbitrary set~$Y$ of $\tau$-frequent substrings (note, $(\tau,k)$-resilient substrings are also $\tau$-frequent) and a dataset version $V$, we define:
\[\LR(Y,V) : = |\{F \in Y : F \text{ is not $\tau$-frequent in $V$} \}| / |Y| .\]

All experiments 
run on an AMD EPYC 7282 CPU with 252 GB RAM. All methods were implemented in \texttt{C++}. 

\begin{figure*}[!ht]
    \centering
    \begin{subfigure}{0.45\textwidth}
        \includegraphics[width=\textwidth]{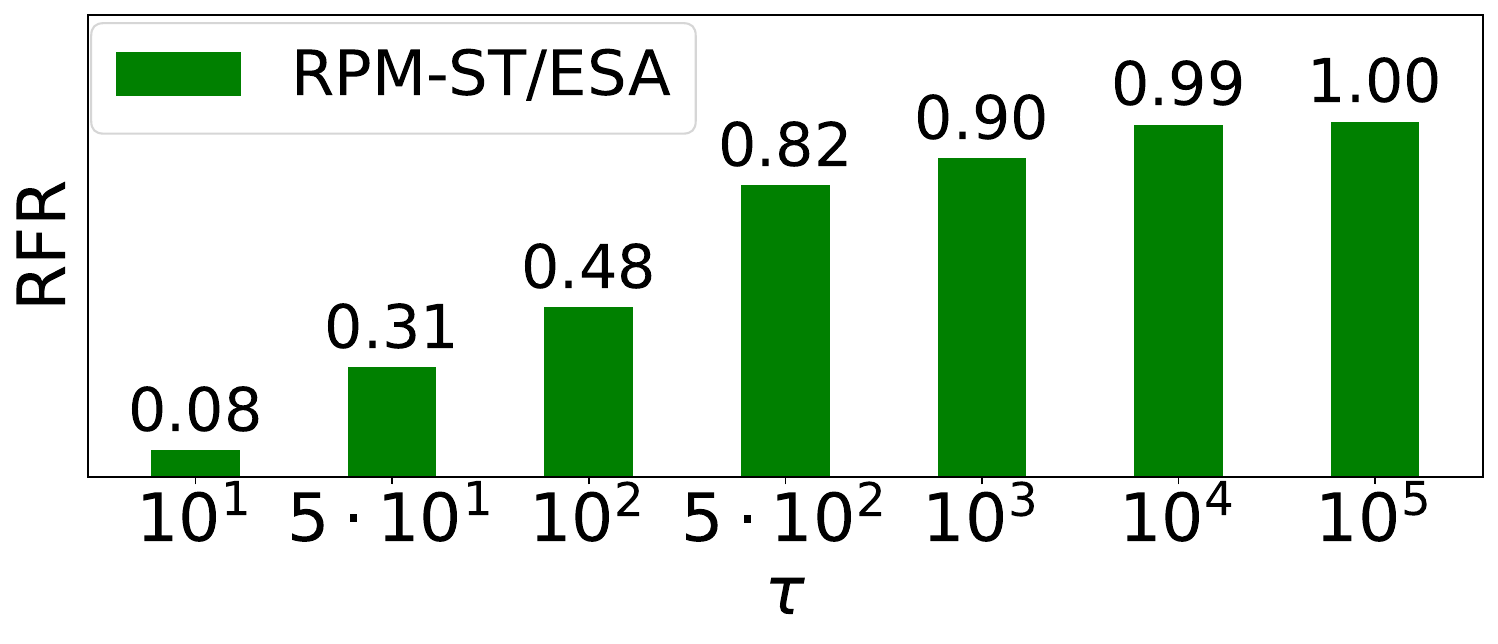}
        \caption{\dna}
    \end{subfigure}
    \begin{subfigure}{0.45\textwidth}
        \includegraphics[width=\textwidth]{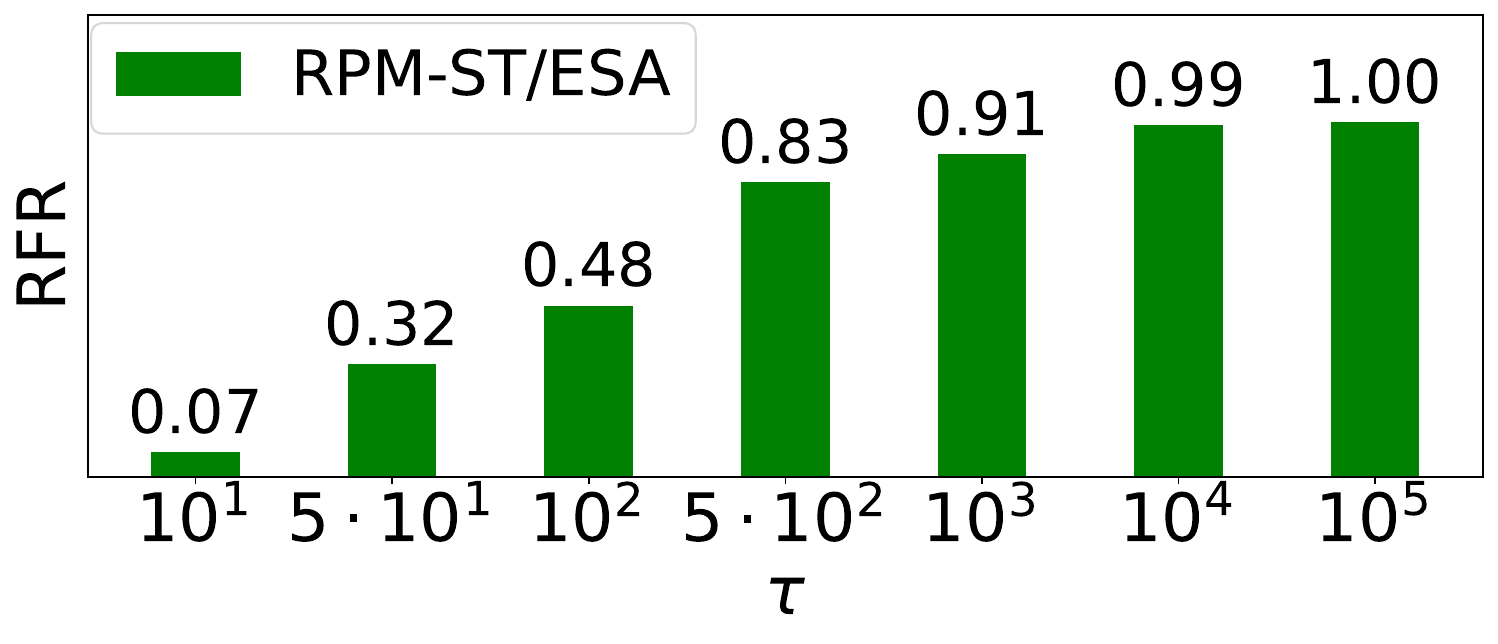}
        \caption{\english}
    \end{subfigure}
    \begin{subfigure}{0.45\textwidth}
        \includegraphics[width=\textwidth]{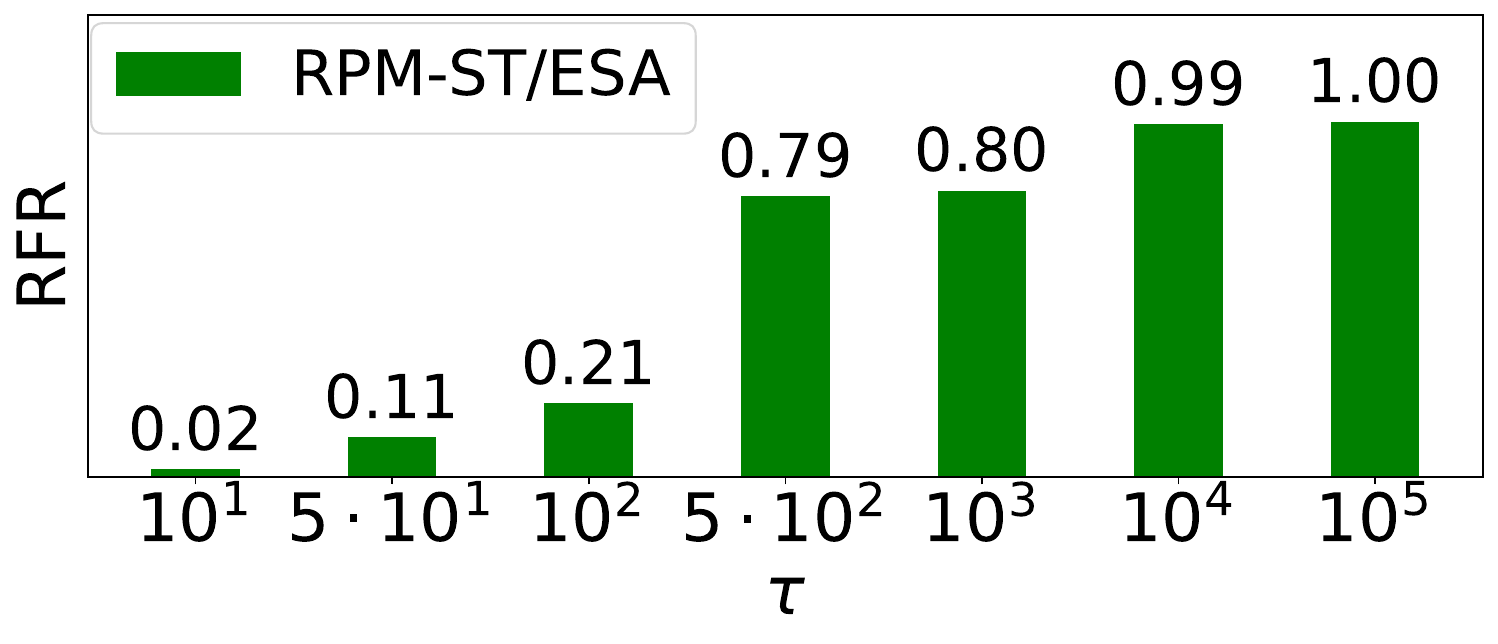}
        \caption{\proteins}
    \end{subfigure}
    \begin{subfigure}{0.45\textwidth}
        \includegraphics[width=\textwidth]{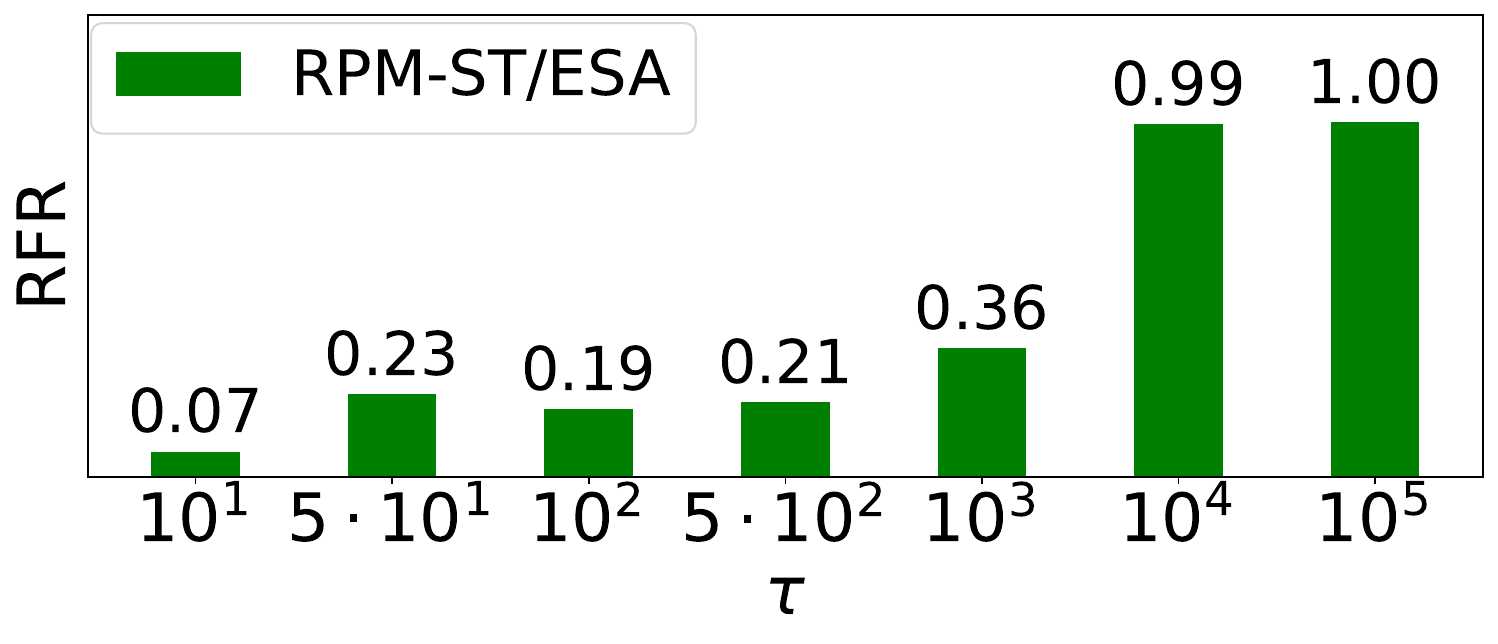}
        \caption{\sources}
    \end{subfigure}
    \begin{subfigure}{0.45\textwidth}
        \includegraphics[width=\textwidth]{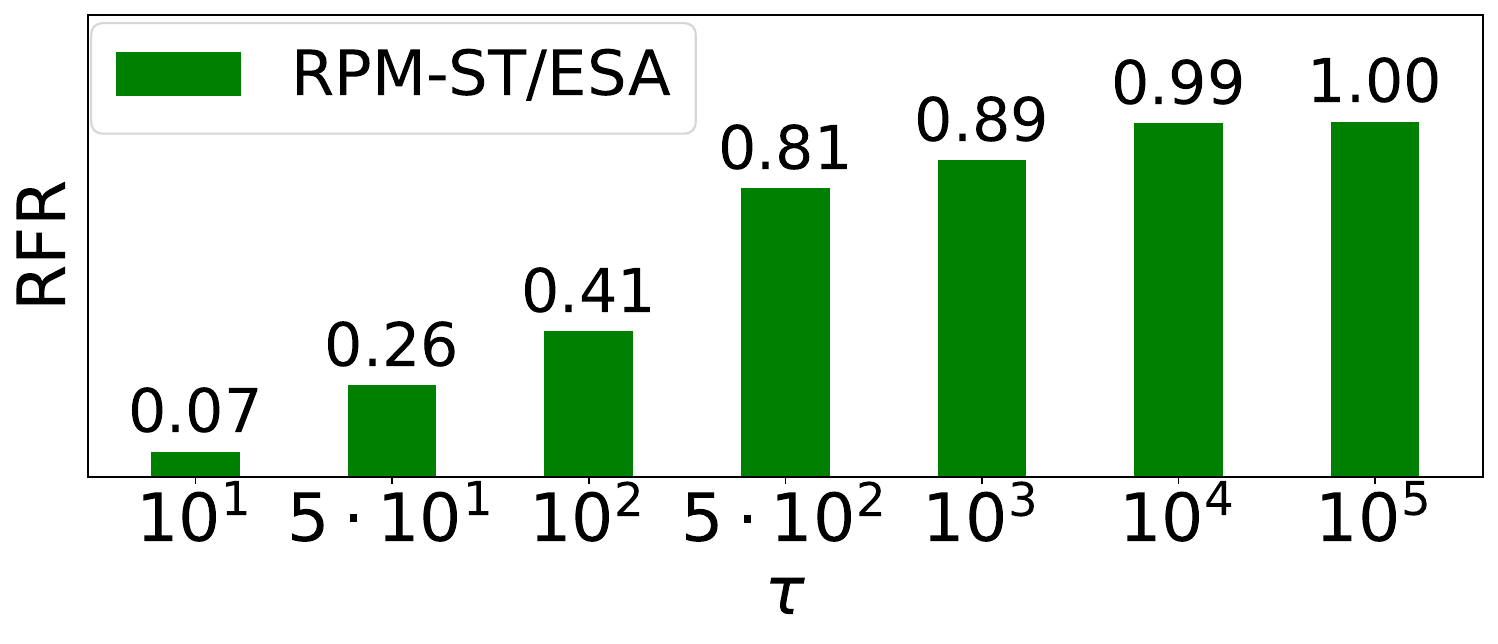}
        \caption{\xml}
    \end{subfigure}
    \caption{\RFR for varying $\tau$. The \RFR values are on the top of the bars.}
    \label{fig:effi_quality_tau}
\end{figure*}
\begin{figure*}[!ht]
    \centering
    \begin{subfigure}{0.45\textwidth}
        \includegraphics[width=\textwidth]{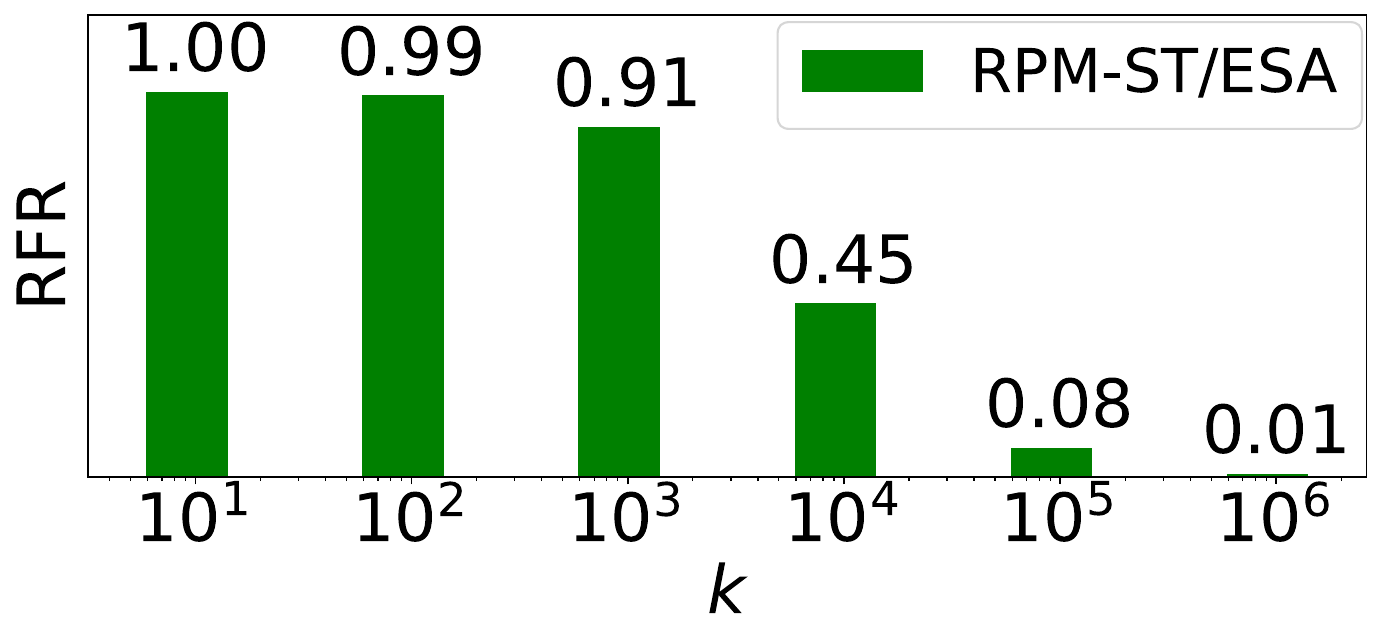}
        \caption{\dna}
    \end{subfigure}
    \begin{subfigure}{0.45\textwidth}
        \includegraphics[width=\textwidth]{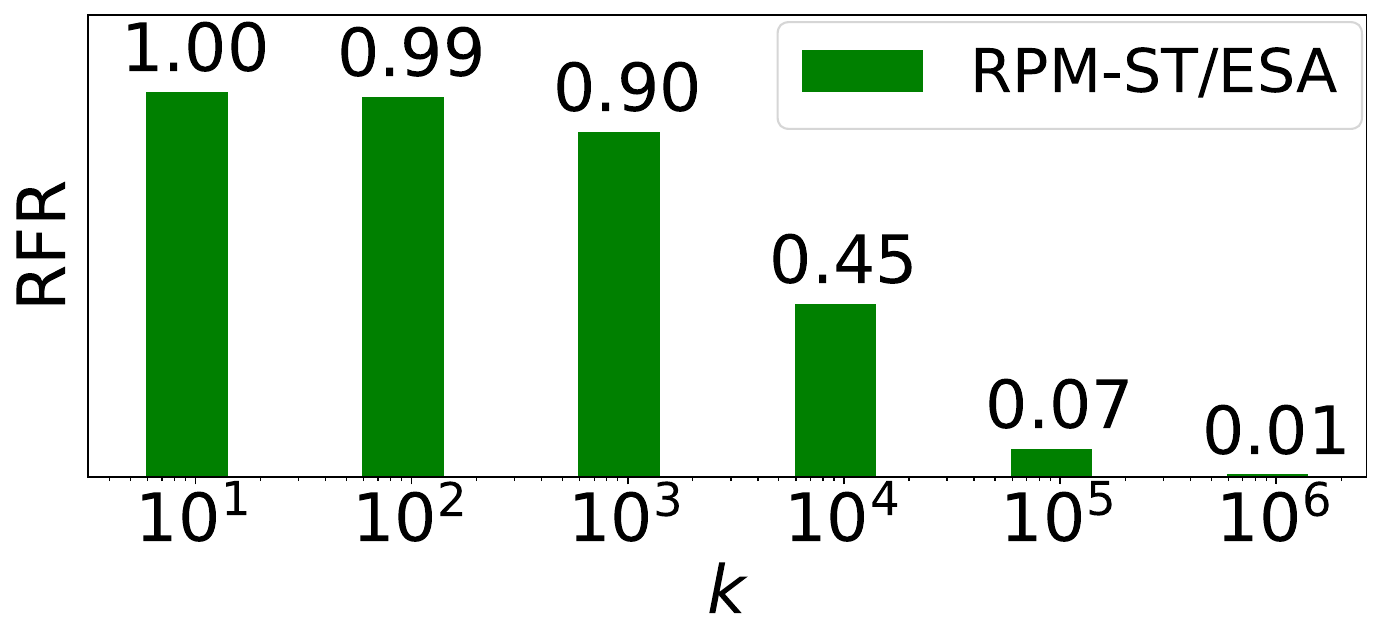}
        \caption{\english}
    \end{subfigure}
    \begin{subfigure}{0.45\textwidth}
        \includegraphics[width=\textwidth]{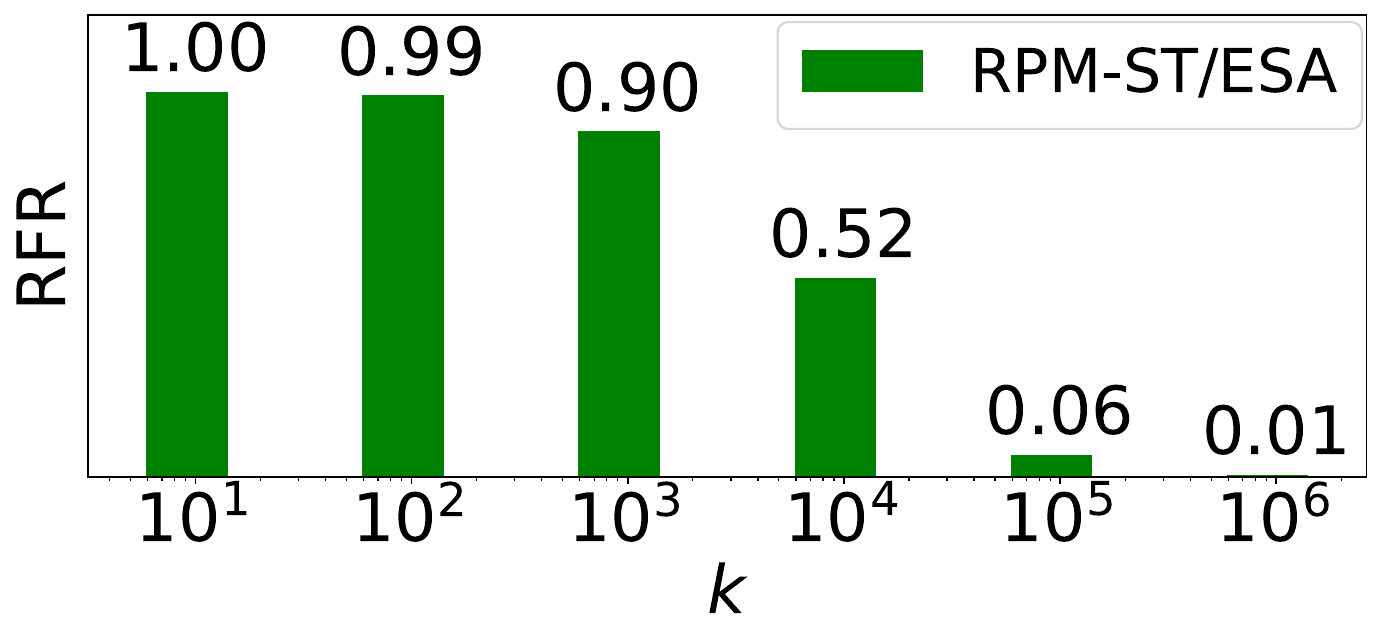}
        \caption{\proteins}
    \end{subfigure}
    \begin{subfigure}{0.45\textwidth}
        \includegraphics[width=\textwidth]{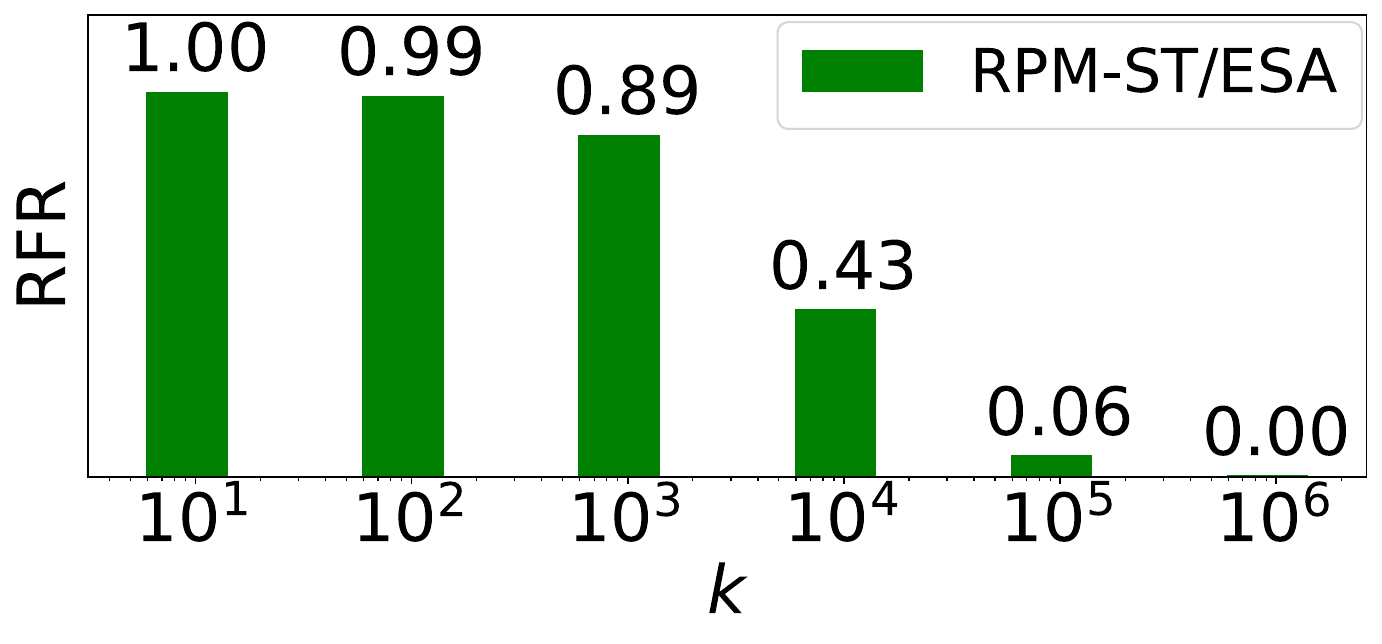}
        \caption{\sources}
    \end{subfigure}
    \begin{subfigure}{0.45\textwidth}
        \includegraphics[width=\textwidth]{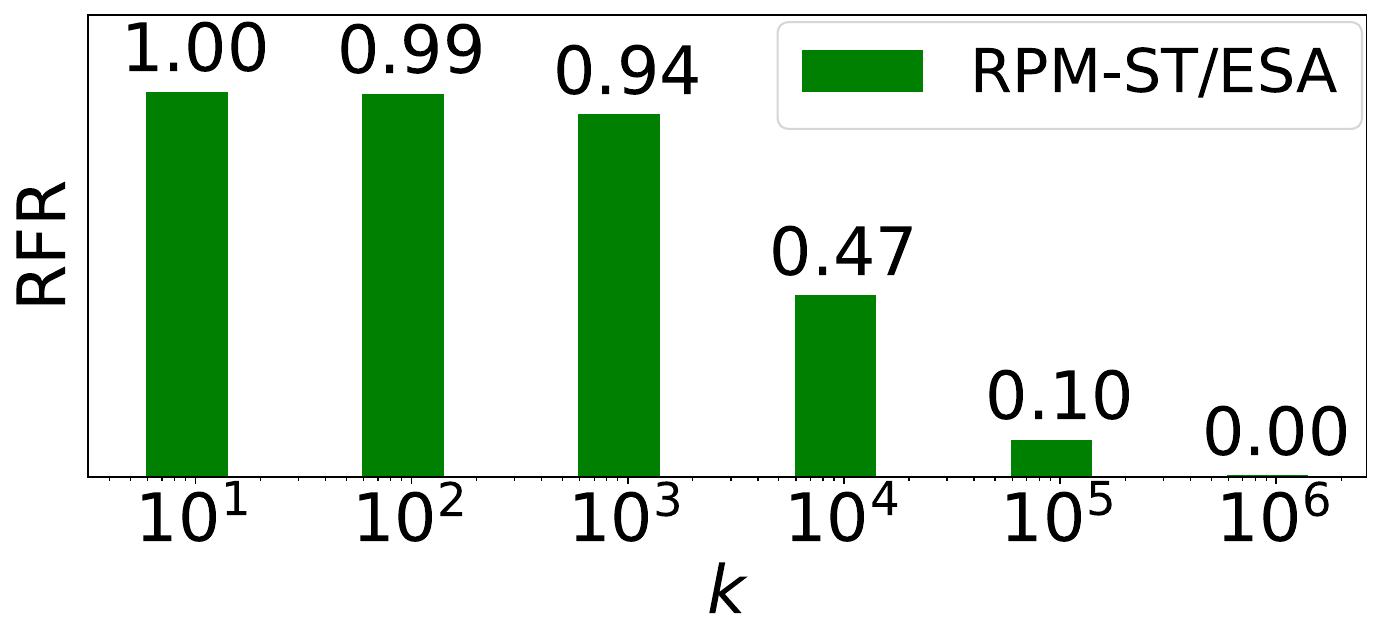}
        \caption{\xml}
    \end{subfigure}
    \caption{\RFR for varying $k$. The \RFR values are on the top of the bars.}
    \label{fig:effi_quality_k}
\end{figure*}
\begin{figure*}[!ht]
    \centering
     \begin{subfigure}{0.45\textwidth}
        \includegraphics[width=\textwidth]{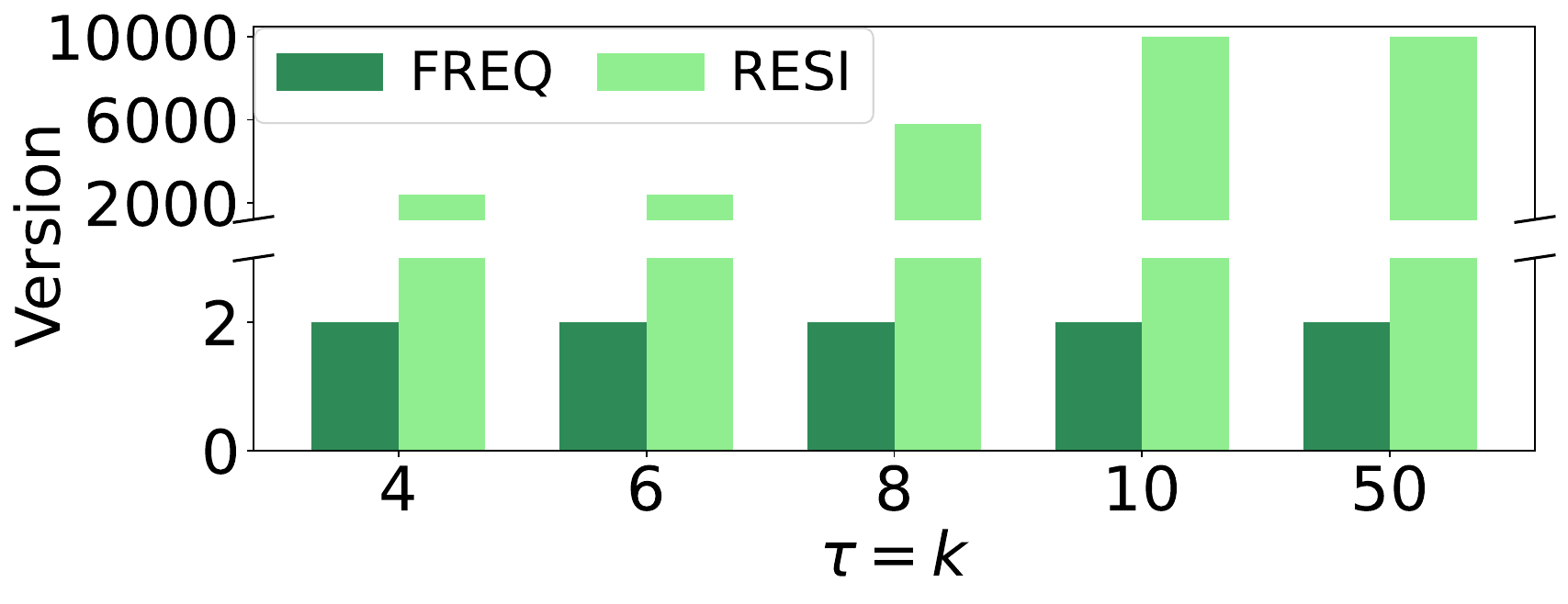}
        \caption{\boost}\label{fig:earliest_version:boost}
    \end{subfigure}
    \begin{subfigure}{0.45\textwidth}
        \includegraphics[width=\textwidth]{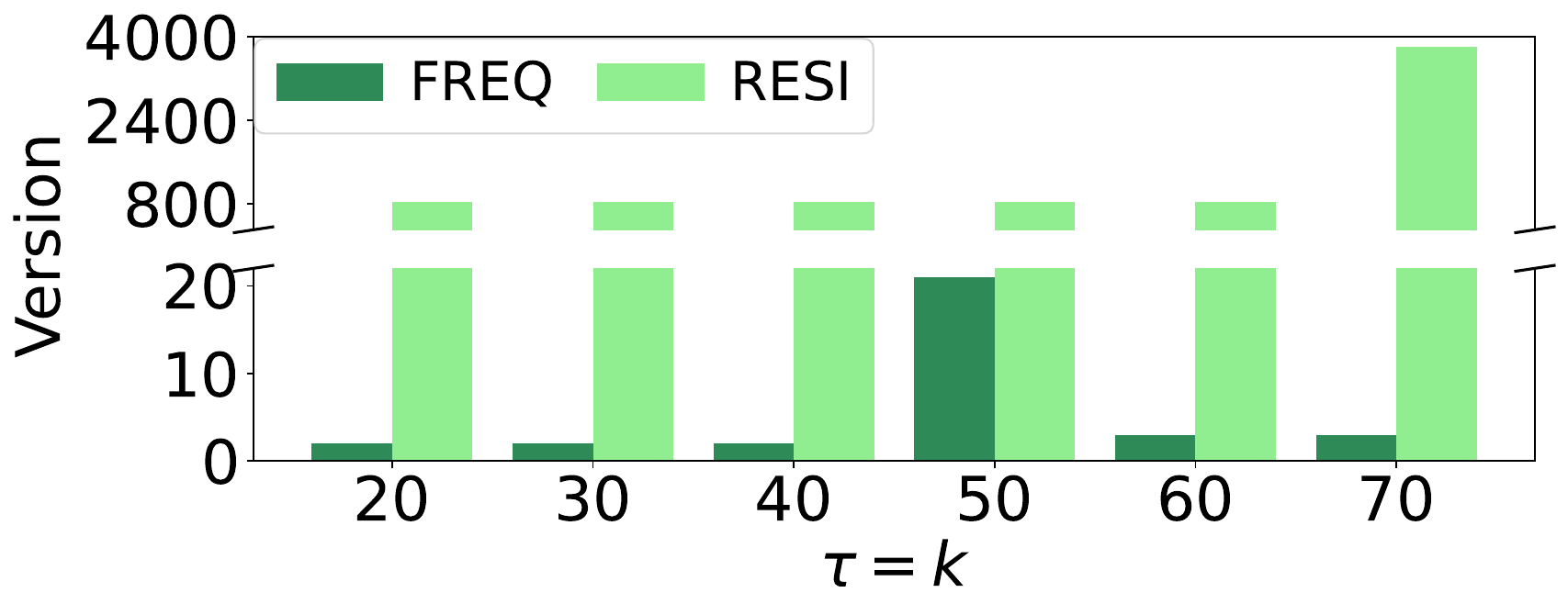}
        \caption{\wiki}\label{fig:earliest_version:wiki}
    \end{subfigure}
     \begin{subfigure}{0.45\textwidth}
        \includegraphics[width=\textwidth]{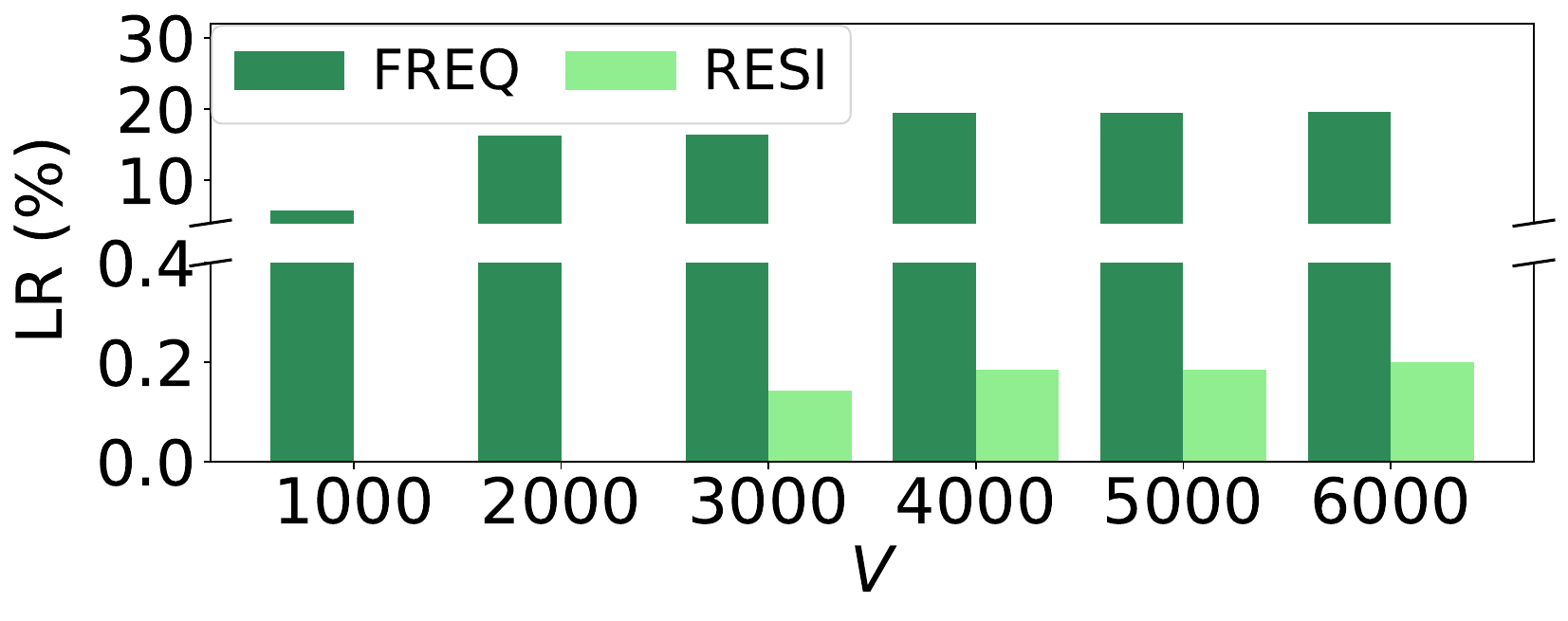}
        \caption{\boost}\label{fig:LR:boost}
    \end{subfigure}
    \begin{subfigure}{0.45\textwidth}
        \includegraphics[width=\textwidth]{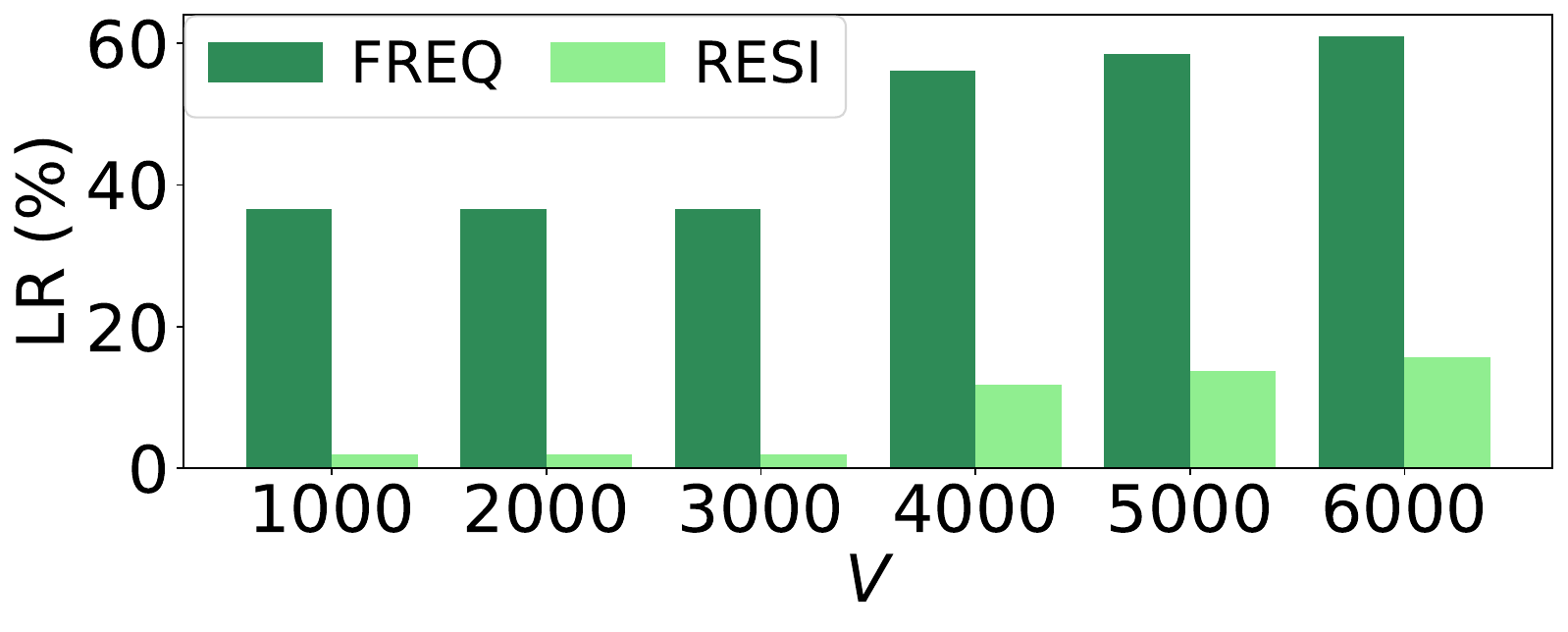}
        \caption{\wiki}\label{fig:LR:wiki}
    \end{subfigure}
    \caption{The earliest version in which at least one of the $\tau$-frequent substrings in set FREQ or of  the $(\tau,k)$-resilient substrings in set RESI that is  mined in Version $1$ is no longer $\tau$-frequent in (a) \boost and (b) \wiki.
   (c, d) \LR for varying $V$.}
\end{figure*}

\begin{figure*}[!ht]
    \centering
    \begin{subfigure}{0.45\textwidth}
        \includegraphics[width=\textwidth]{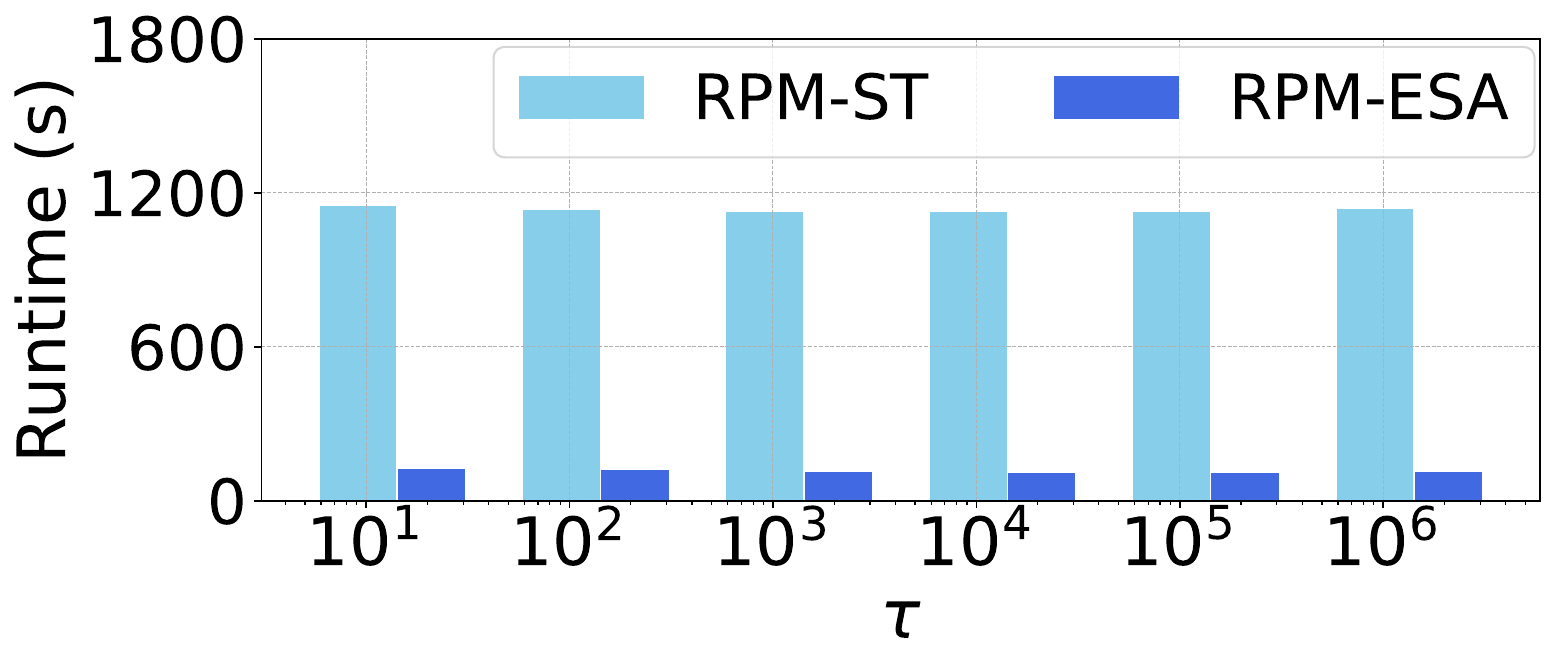}
        \caption{\dna}
    \end{subfigure}
    \begin{subfigure}{0.45\textwidth}
        \includegraphics[width=\textwidth]{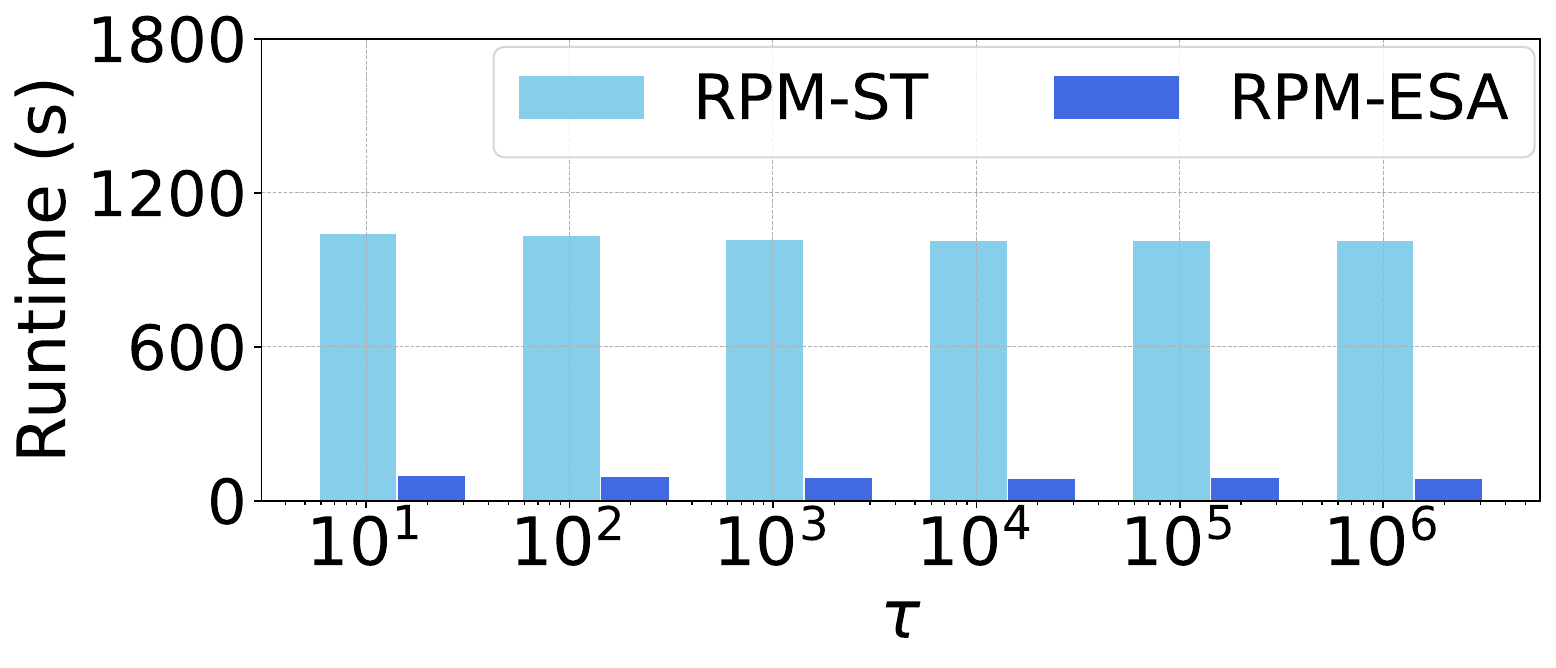}
        \caption{\english}
    \end{subfigure}
    \begin{subfigure}{0.45\textwidth}
        \includegraphics[width=\textwidth]{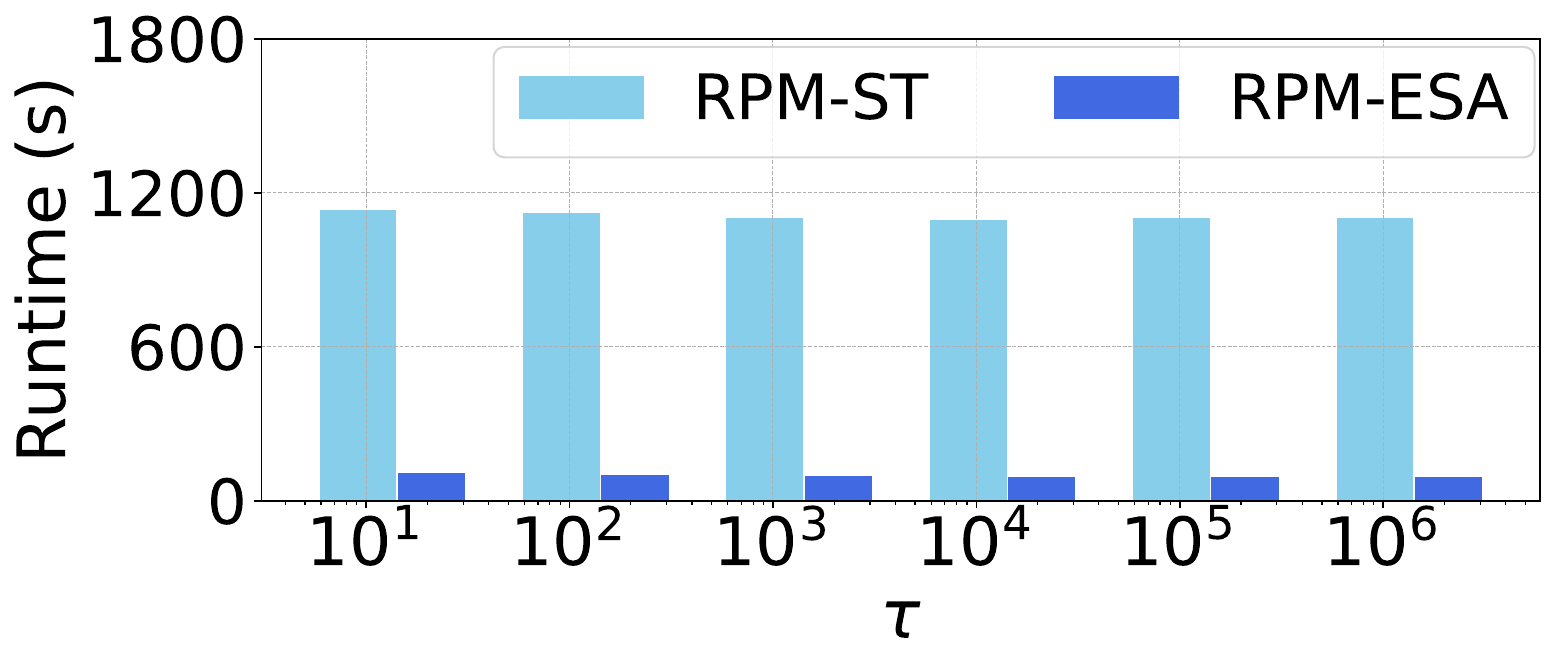}
        \caption{\proteins}
    \end{subfigure}
    \begin{subfigure}{0.45\textwidth}
        \includegraphics[width=\textwidth]{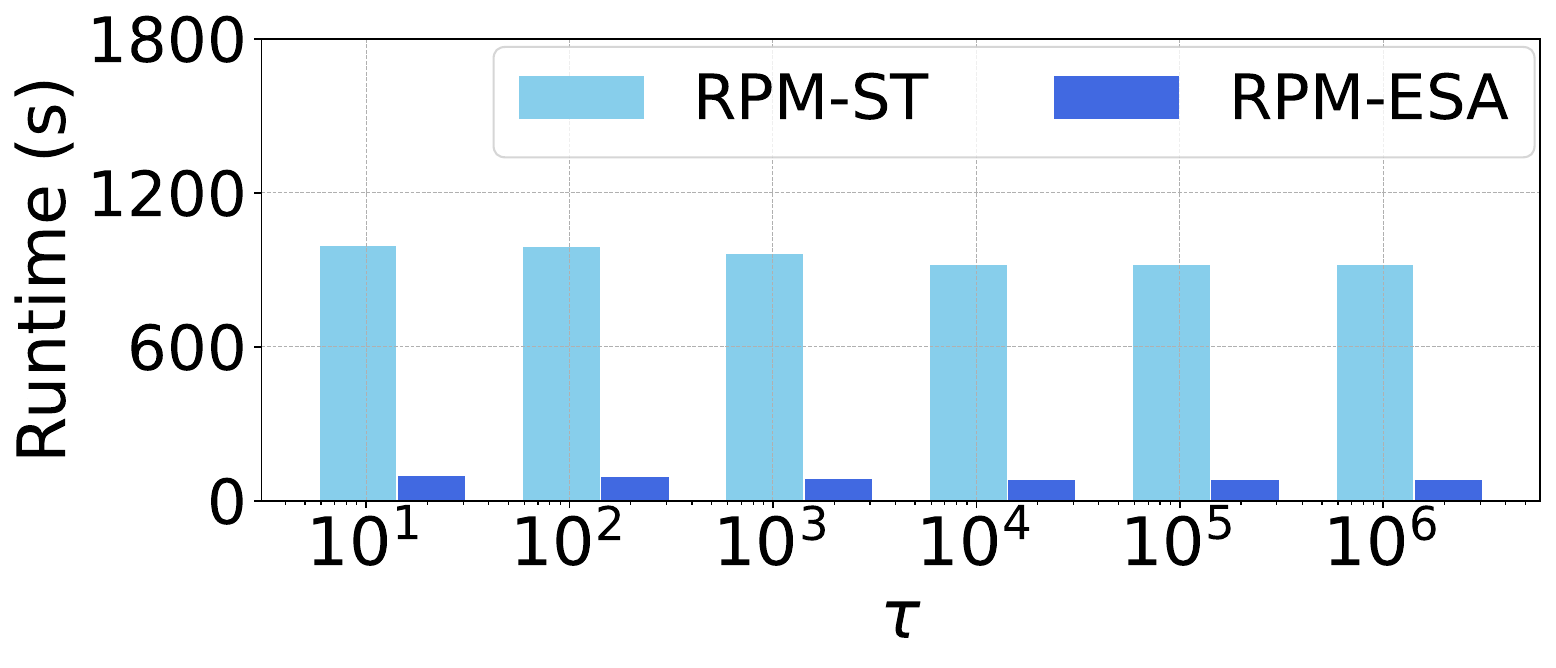}
        \caption{\sources}
    \end{subfigure}
    \begin{subfigure}{0.45\textwidth}
        \includegraphics[width=\textwidth]{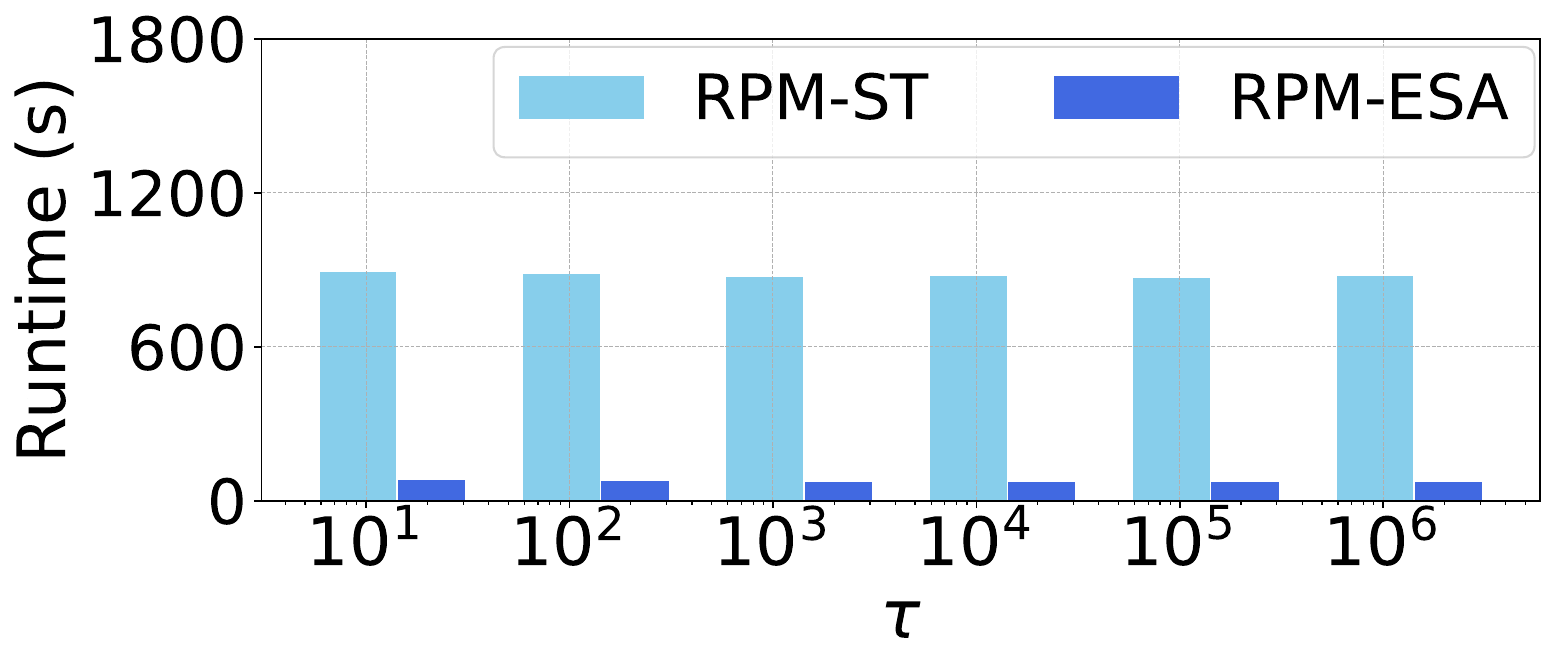}
        \caption{\xml}
    \end{subfigure}
    \caption{Running time for varying $\tau$.}
    \label{fig:effi_runtime_tau}
\end{figure*}
\begin{figure*}[!ht]
    \centering
    \begin{subfigure}{0.45\textwidth}
        \includegraphics[width=\textwidth]{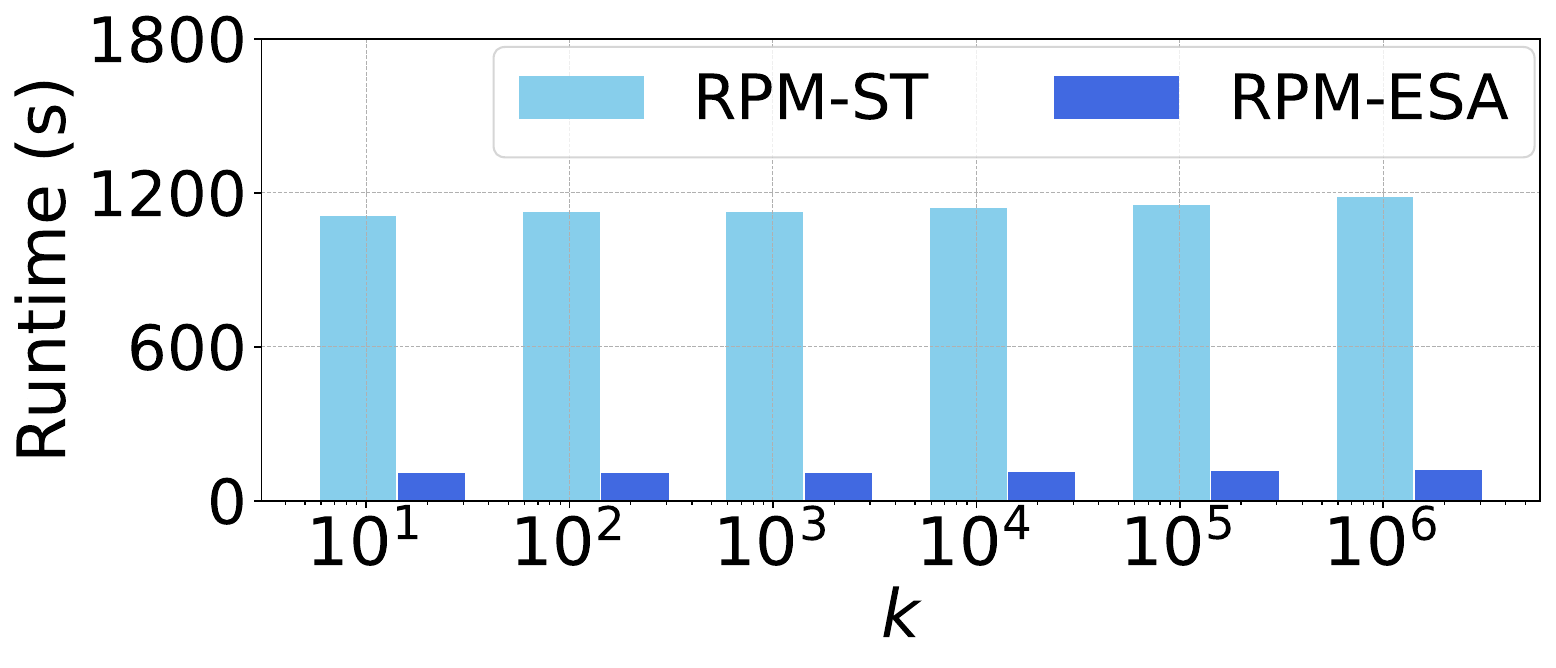}
        \caption{\dna}
    \end{subfigure}
    \begin{subfigure}{0.45\textwidth}
        \includegraphics[width=\textwidth]{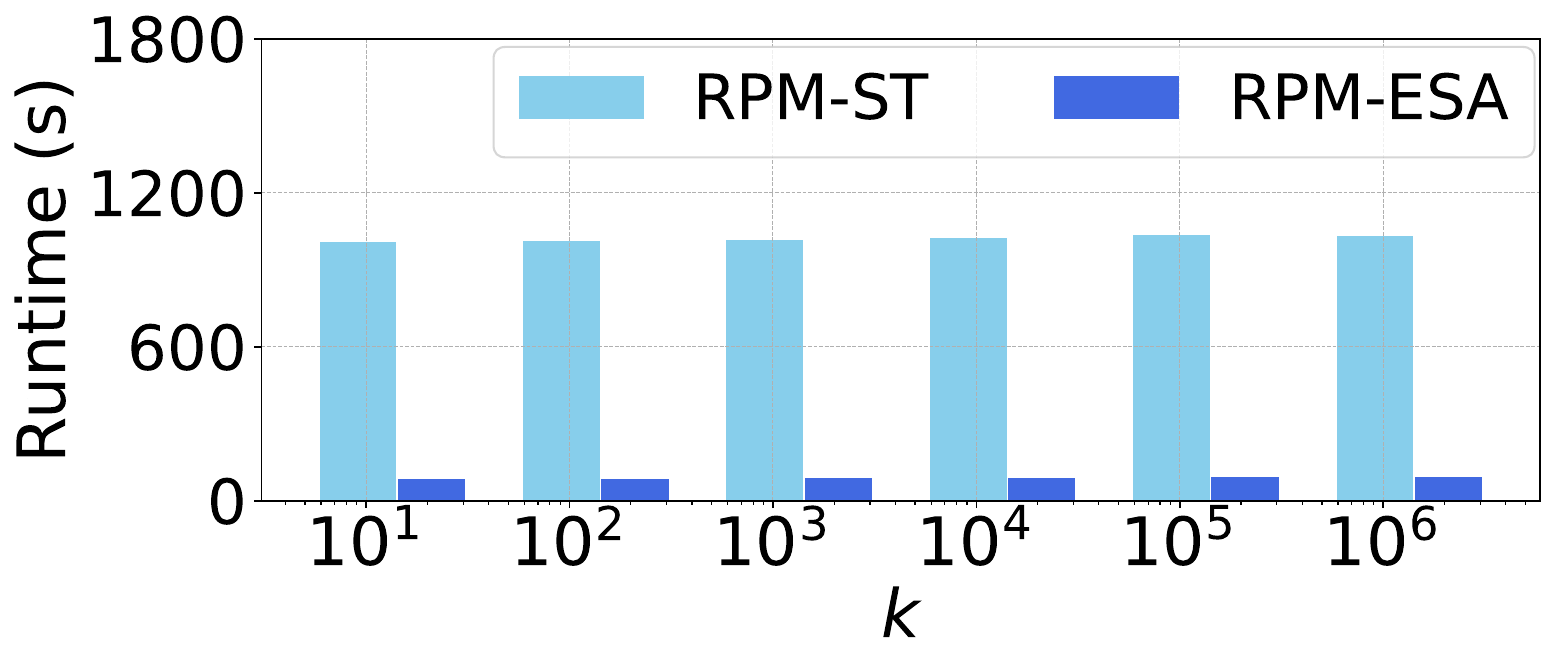}
        \caption{\english}
     \end{subfigure}
    \begin{subfigure}{0.45\textwidth}
        \includegraphics[width=\textwidth]{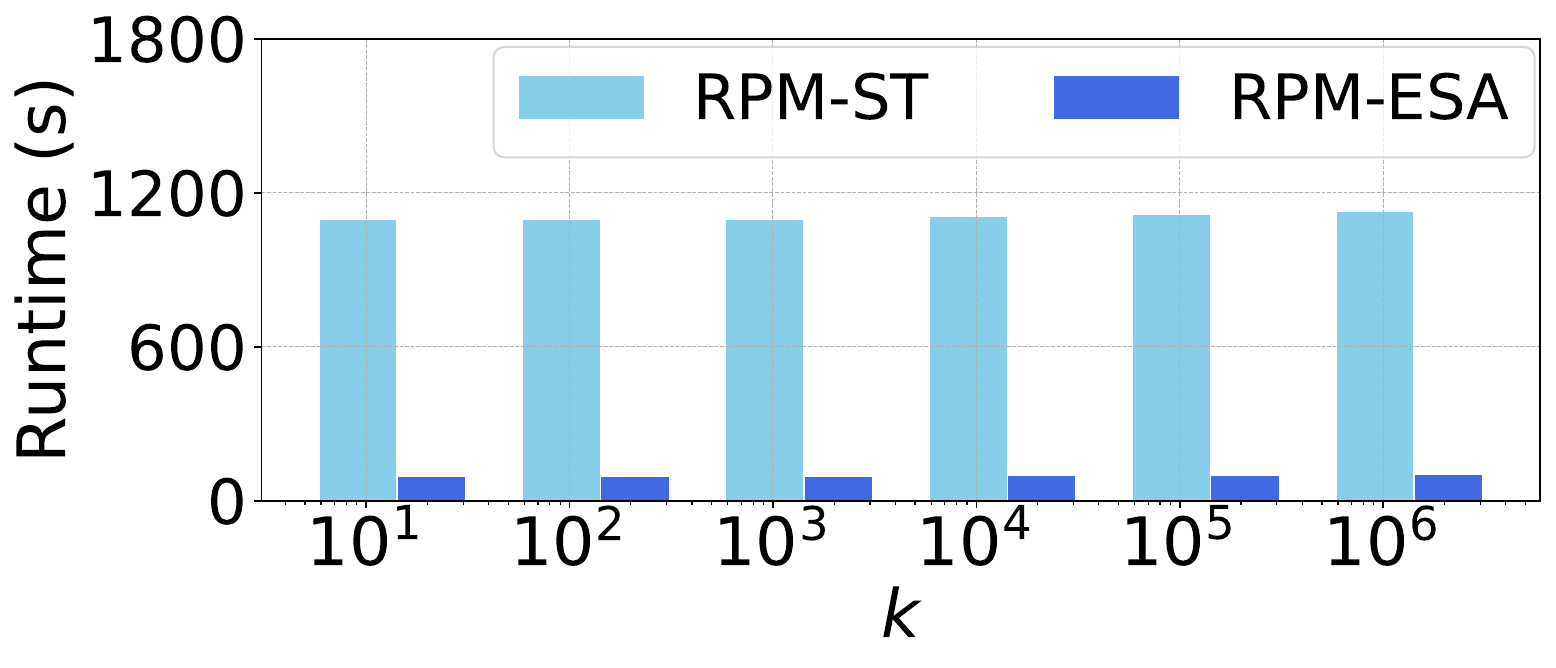}
        \caption{\proteins}
    \end{subfigure}
    \begin{subfigure}{0.45\textwidth}
        \includegraphics[width=\textwidth]{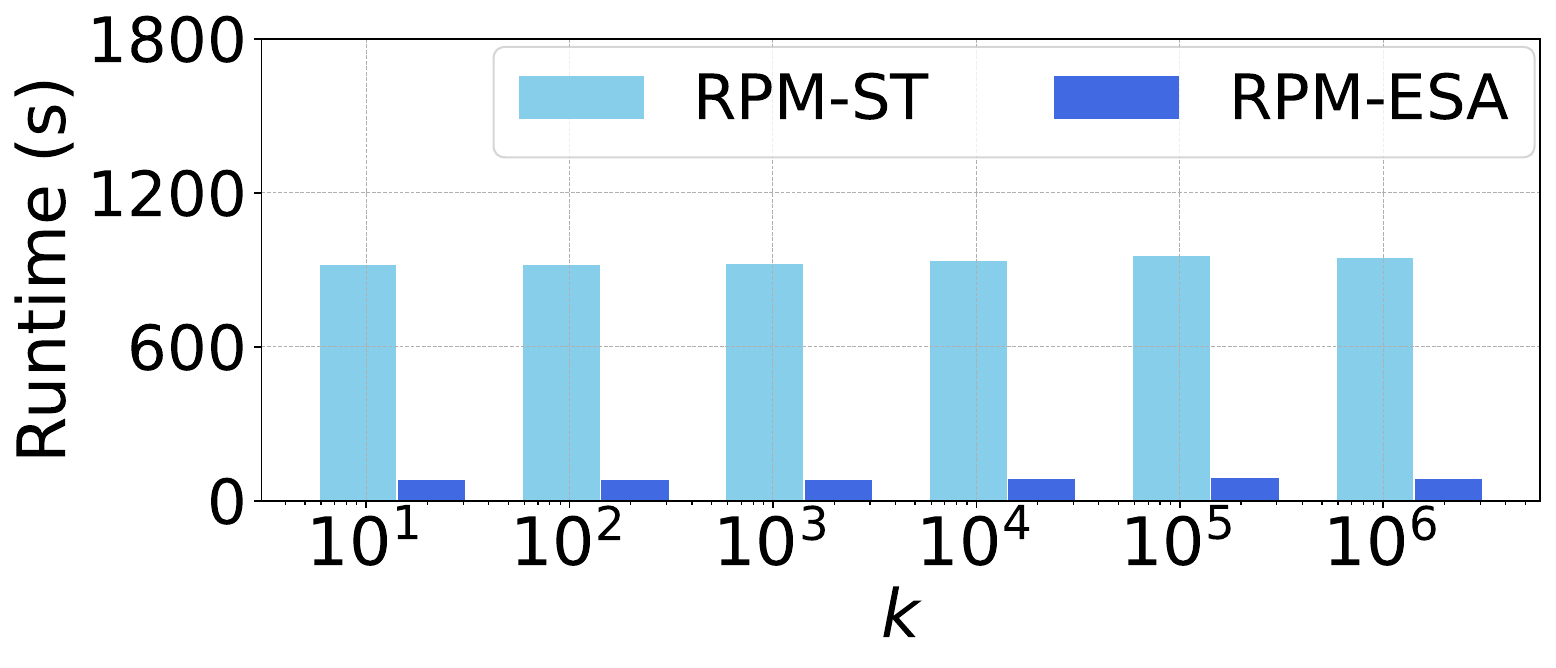}
        \caption{\sources}
    \end{subfigure}
    \begin{subfigure}{0.45\textwidth}
        \includegraphics[width=\textwidth]{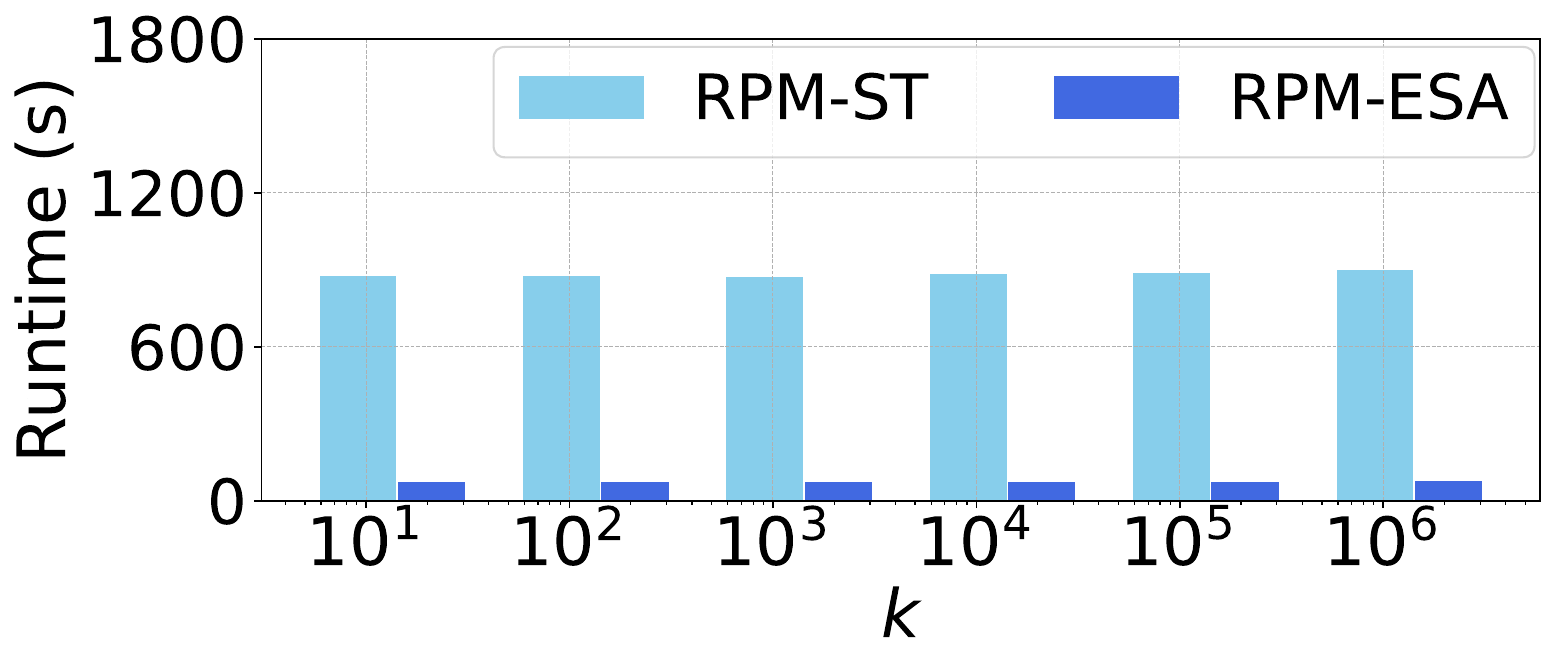}
        \caption{\xml}
    \end{subfigure}
    \caption{Running time for varying $k$.}
    \label{fig:effi_runtime_k}
\end{figure*}

\begin{figure}[!ht]
    \centering
    \begin{subfigure}{0.45\textwidth}
        \includegraphics[width=\textwidth]{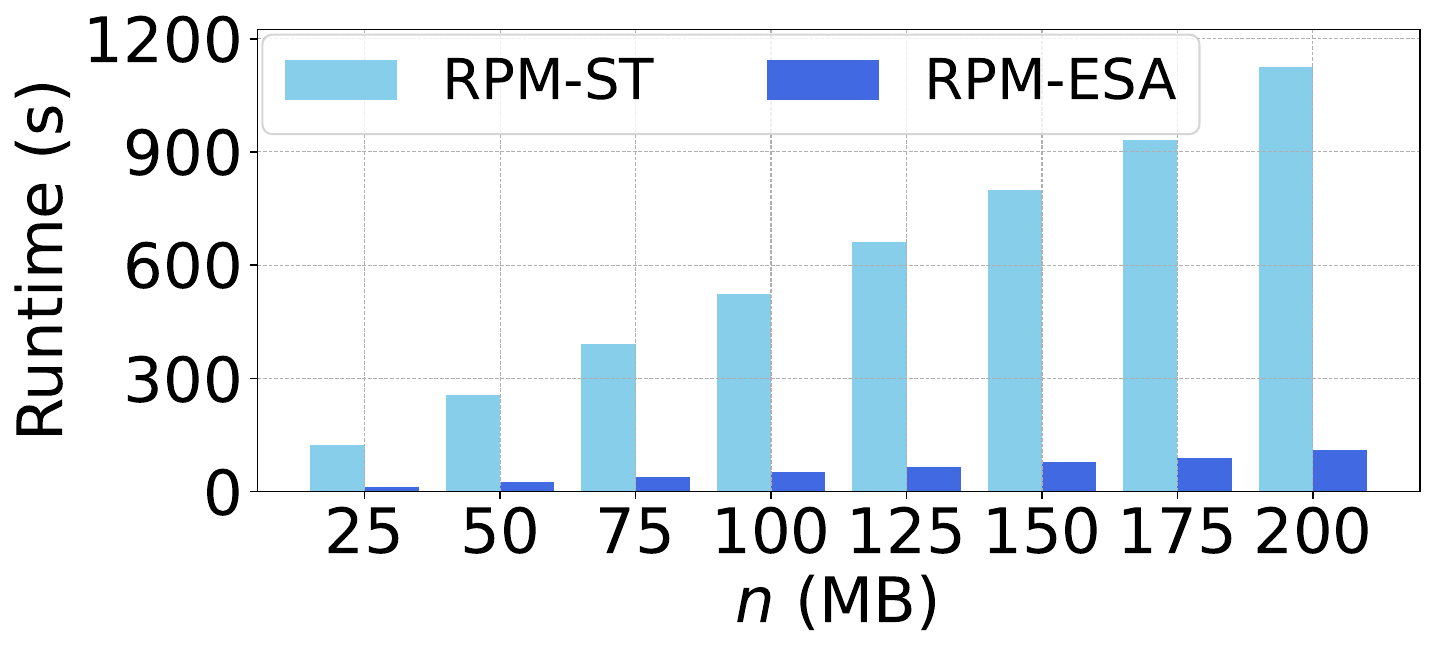}
        \caption{Prefixes of \dna}\label{fig:time:effi:n}
    \end{subfigure}
    \begin{subfigure}{0.45\textwidth}
        \includegraphics[width=\textwidth]{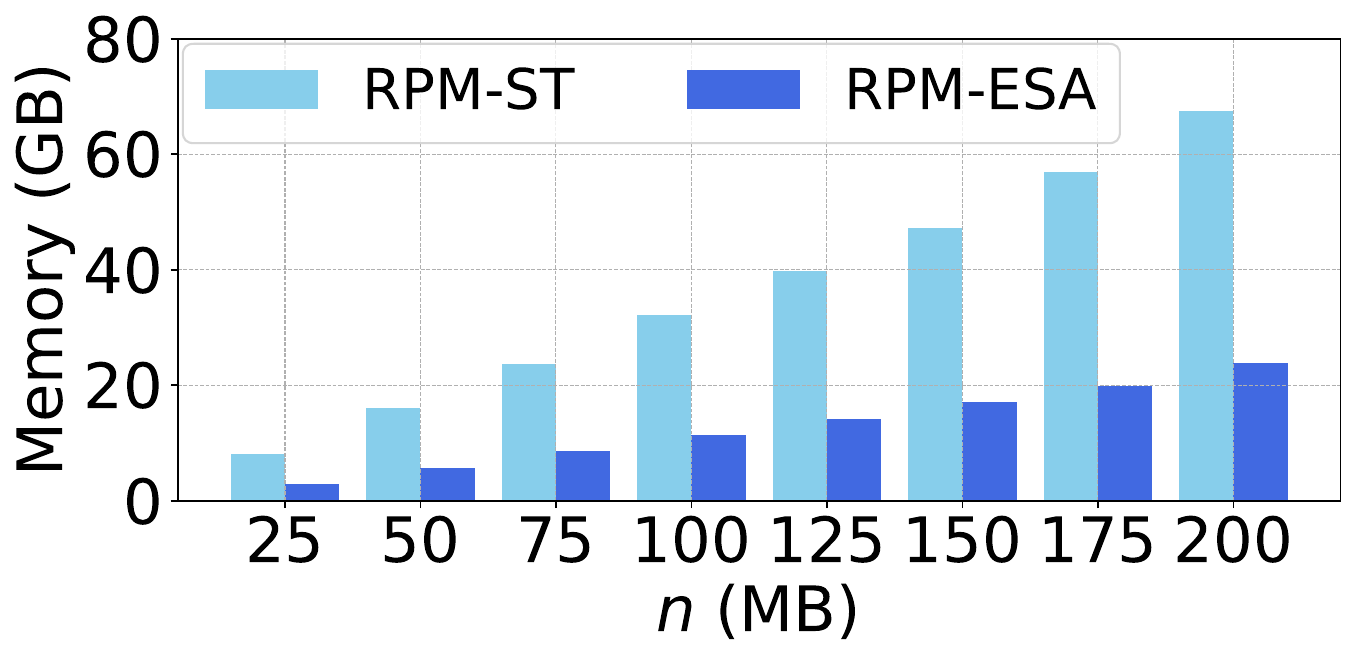}
        \caption{Prefixes of \dna}\label{fig:memory:effi:n}
    \end{subfigure}
    \caption{(a) Running time and (b) space for varying $n$.}
    \label{fig:effi_runtime_n}
\end{figure}
\begin{figure}[!ht]
    \centering
     \begin{subfigure}{0.45\textwidth}
        \includegraphics[width=\textwidth]{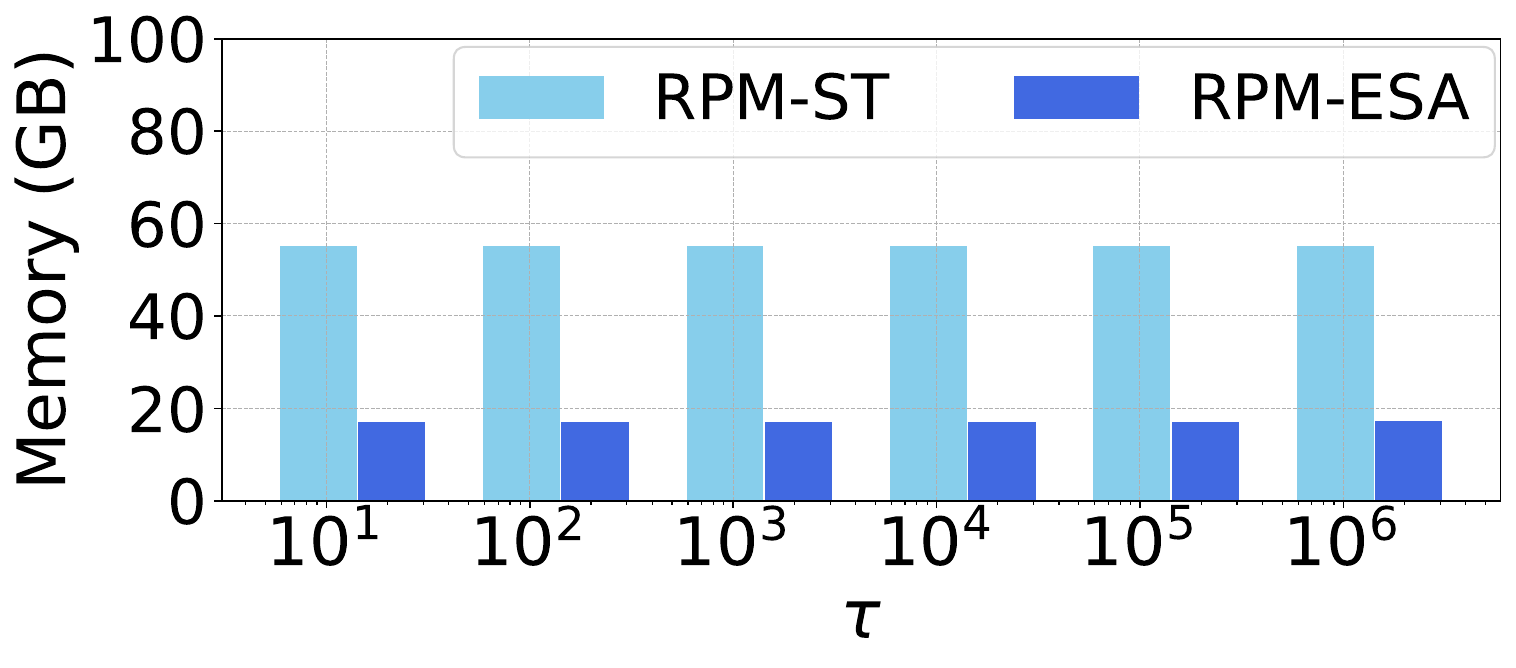}
        \caption{\xml}\label{fig:memory:effi:tau}
    \end{subfigure}
    \begin{subfigure}{0.45\textwidth}
        \includegraphics[width=\textwidth]{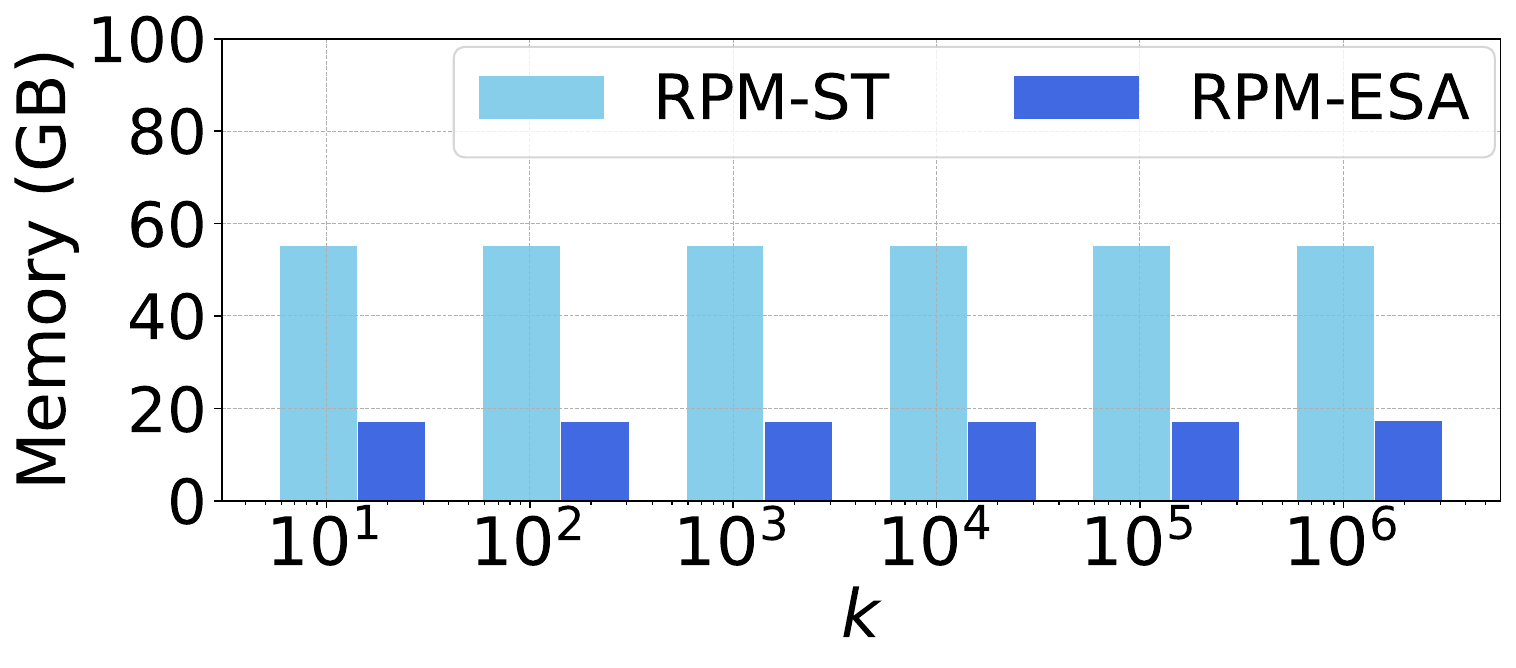}
        \caption{\xml}\label{fig:memory:effi:k}
    \end{subfigure}
    \caption{Space for varying (a) $\tau$ and (b) $k$.}
    \label{fig:effi_space}
\end{figure}

\paragraph*{\RFR.}~\cref{fig:effi_quality_tau} 
 shows that \RFR is at most $0.08$ for $\tau=10$ and $k=100$; thus most $10$-frequent substrings are \emph{not} $(10,100)$-resilient. \RFR increases with~$\tau$ since, due to the relatively small $k$ value we used, more $\tau$-frequent substrings remain $\tau$-frequent (and thus are $(\tau, 100)$-resilient) when $\tau$ increases.  \cref{fig:effi_quality_k} 
 shows that \RFR decreases (and approaches or becomes $0$) as~$k$ increases. This is  because, as $\tau$ is fixed, a larger $k$ allows more letter substitutions and thus most or all substrings that are $(\tau, k)$-resilient for one $k$ are no longer resilient for a larger one.   
  Thus, the sets of $(\tau, k)$-resilient substrings and $\tau$-frequent substrings for the same $\tau$ can be very different, which motivates the notion of $(\tau, k)$-resilient substrings.  

\paragraph*{Resilience in Versioned Datasets.}~\cref{fig:earliest_version:boost} (and~\cref{fig:earliest_version:wiki}) shows that the set of $\tau$-frequent substrings in version $1$ of the \boost (and \wiki dataset), denoted by FREQ, changes in the next version of the dataset in most cases.
For example, in \boost, the earliest  version in which it changes is $2$. 
On the other hand, the set of $(\tau,k)$-resilient substrings for the same $\tau$ and $k$ as FREQ, which is denoted by RESI, changes much later. 
For example, the earliest version in which RESI  changes is $2,412$ in \boost and $828$ in \wiki.

\cref{fig:LR:boost} and \cref{fig:LR:wiki} show that \LR for the set of $\tau$-frequent substrings is much larger compared to the \LR for the set of $(\tau,k)$-resilient substrings, for any version $V$ (on average, by more than $2$ orders of magnitude in \boost and by more than $1$ order of magnitude in \wiki). 
Thus, only a very small number of $(\tau,k)$-resilient substrings stop being $\tau$-frequent as more versions of the dataset are considered, and this happens after a very large number of versions (e.g., after version $2,000$ in \boost).  
As expected, \LR increases with $V$. 

Thus, $(\tau,k)$-resilient substrings remain truthful for longer, which is  useful in applications, as mentioned in Introduction. 

 \paragraph*{Running Time.}~Both \RPMST and \RPMESA have the same time complexity. 
 As expected, both $\tau$ and $k$ play an insignificant role in the running time of our algorithms; see \cref{fig:effi_runtime_tau} and \cref{fig:effi_runtime_k}. 
 It is important to note that \RPMESA is at least~$9$ and up to $12$ times faster than \RPMST in the experiment of \cref{fig:effi_runtime_tau} and an order of magnitude faster in the experiment of \cref{fig:effi_runtime_k}. This is because the enhanced suffix array 
 takes less memory
 and supports faster cache-friendly operations  
 compared to the suffix tree~\cite{DBLP:journals/jda/AbouelhodaKO04}. 
\cref{fig:time:effi:n} shows that both algorithms scale close to linearly in $n$, as expected by their time complexity. \RPMESA is faster (by $10.2$ times on average) for the reason outlined above. 

 \paragraph*{Space.}~\cref{fig:memory:effi:n} shows that the space occupied by both \RPMST and \RPMESA increases linearly with $n$, as expected by their $\cO(n)$ space complexity. The space is thus not affected by~$\tau$ or $k$ (see \cref{fig:memory:effi:tau} and \cref{fig:memory:effi:k} for the \xml dataset; the results for the other datasets are analogous and omitted). \RPMESA occupies much less space than \RPMST (e.g., $3.2$ times less on average in the experiments of \cref{fig:memory:effi:tau} and \cref{fig:memory:effi:k}), due to the use of the enhanced suffix array, which occupies less space than the suffix tree in practice~\cite{DBLP:journals/jda/AbouelhodaKO04}. 
 
\begin{figure}[!t]
    \centering
    \begin{subfigure}{0.45\textwidth}
        \includegraphics[width=\textwidth]{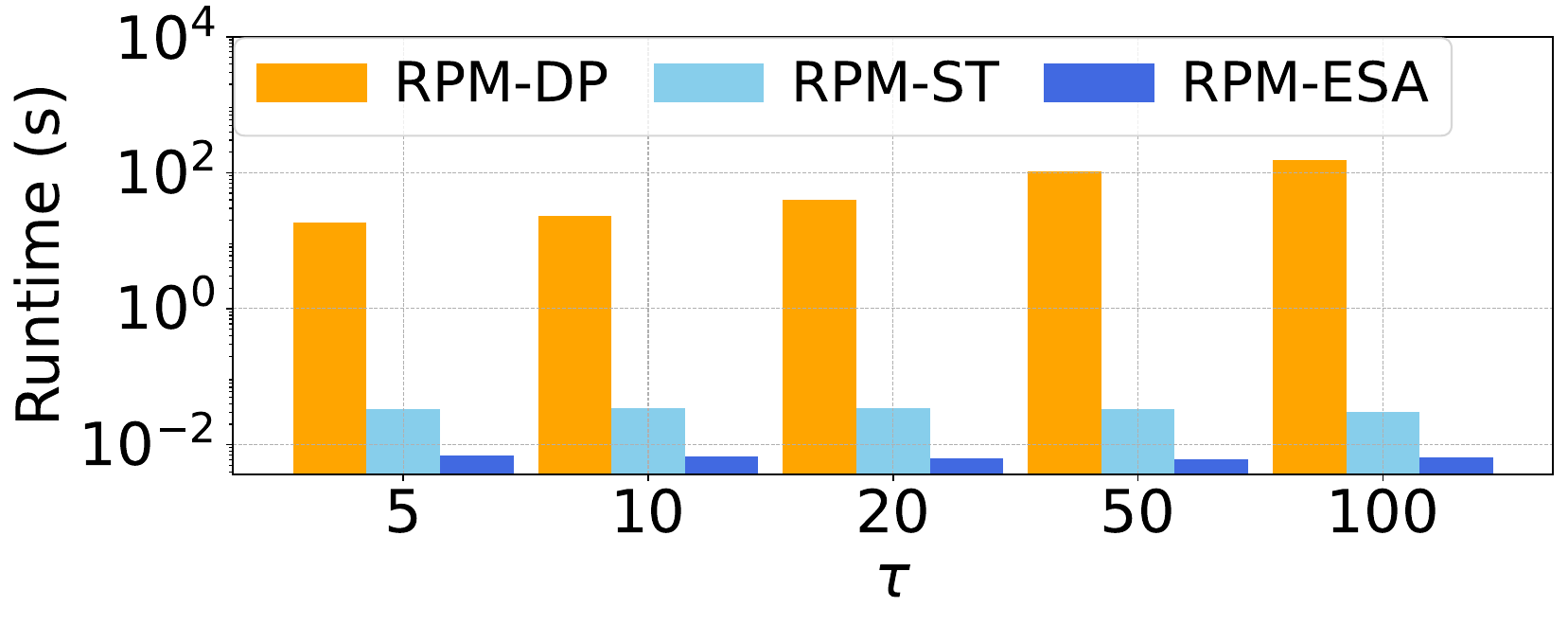}
        \caption{\dna prefix with  $n=1000$}\label{fig:runtime:baseline:tau}
    \end{subfigure}
    \begin{subfigure}{0.45\textwidth}
        \includegraphics[width=\textwidth]{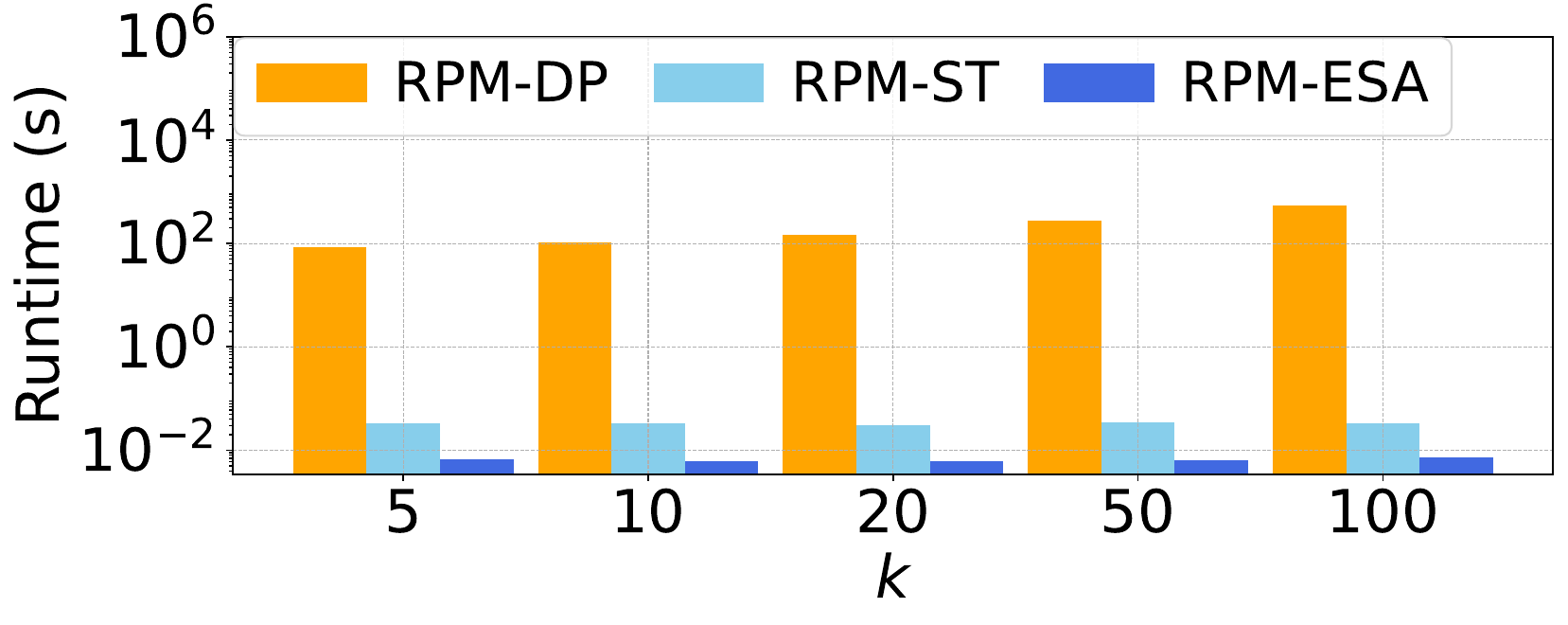}
        \caption{\dna prefix with $n=1000$}\label{fig:runtime:baseline:k}
    \end{subfigure}
    \caption{Running time of \BASELINE, \RPMST and \RPMESA for (a) varying $\tau$ and (b) varying $k$.}
    \label{fig:baseline_vs_RPM_tau_k}
\end{figure}

\begin{figure}[!ht]
    \centering
    \begin{subfigure}{0.45\textwidth}
        \includegraphics[width=\textwidth]{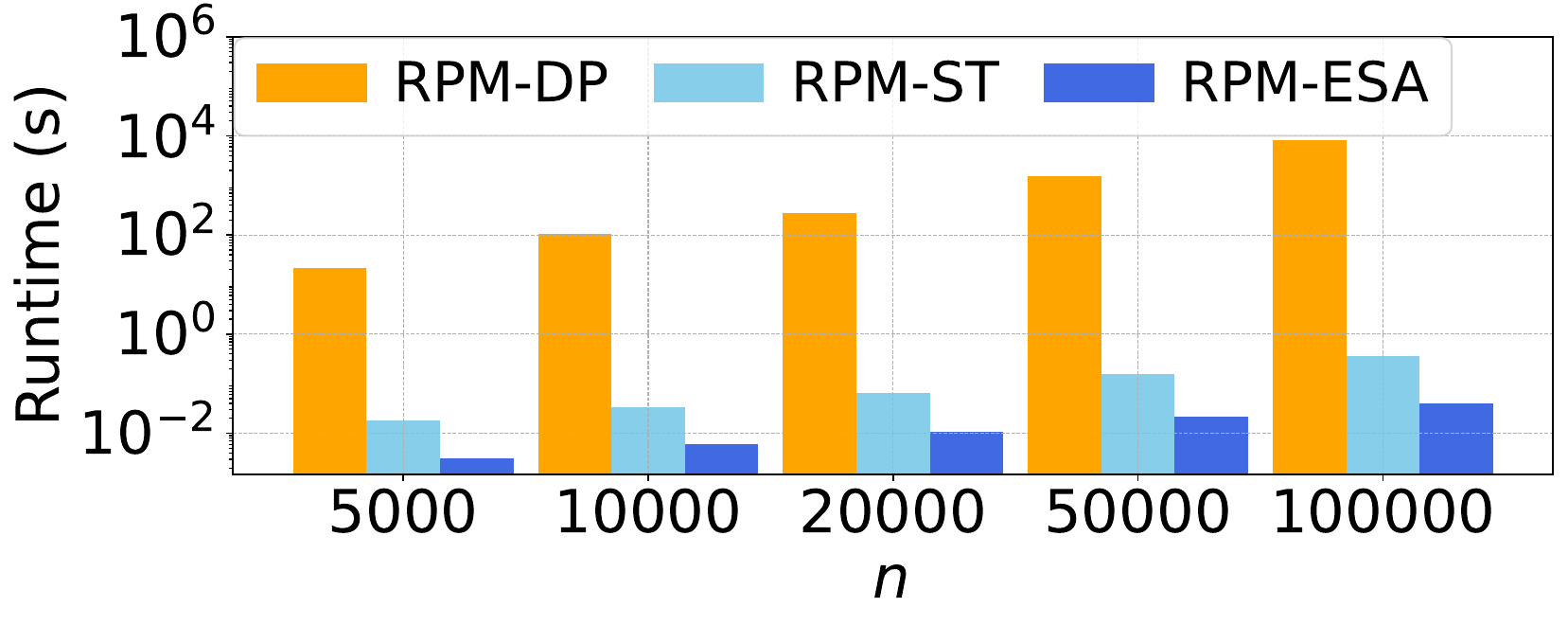}
        \caption{Prefixes of \dna}\label{fig:runtime:baseline:n}
    \end{subfigure}
    \begin{subfigure}{0.45\textwidth}
        \includegraphics[width=\textwidth]{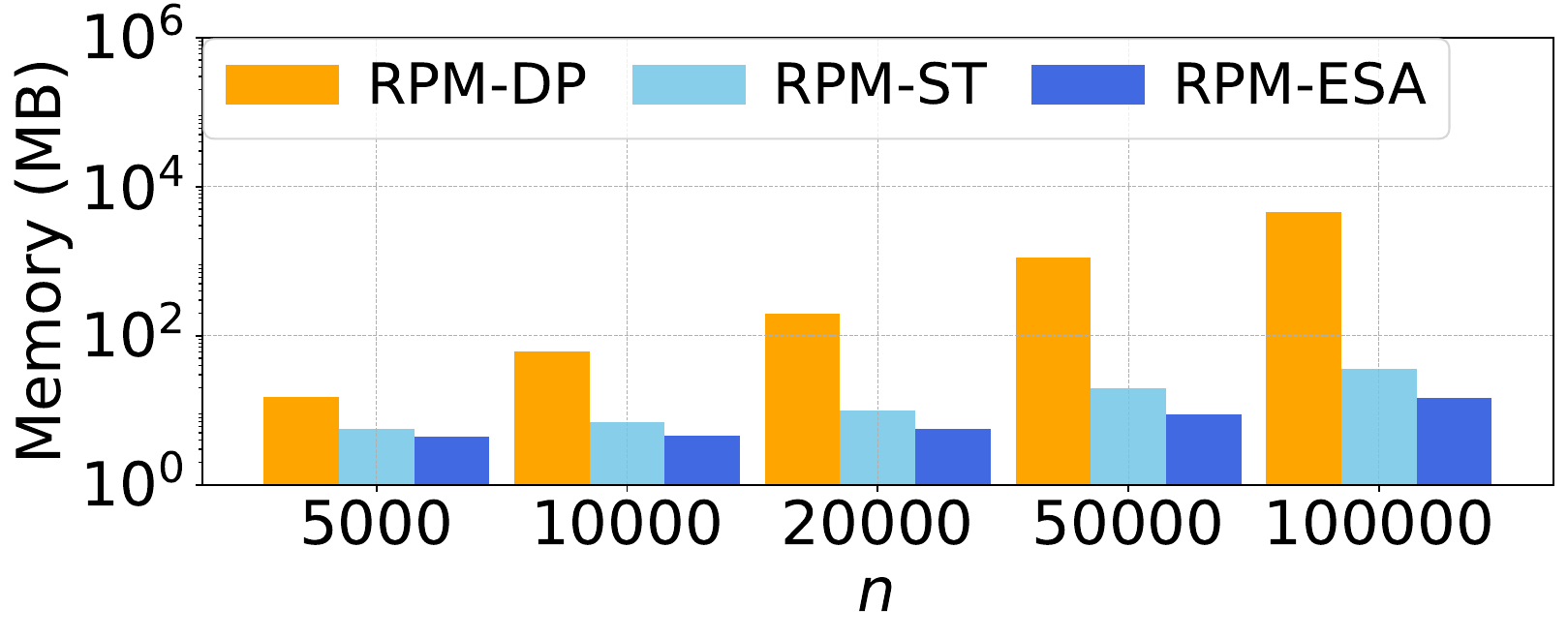}
        \caption{Prefixes of \dna}\label{fig:memory:baseline:n}
    \end{subfigure}
    \caption{(a) Running time and (b) space of \BASELINE, \RPMST and \RPMESA for varying $n$.}
    \label{fig:baseline_vs_RPM_n}
\end{figure}

\paragraph*{Comparison to \BASELINE.}~\cref{fig:runtime:baseline:tau} and \cref{fig:runtime:baseline:k} show the impact of $\tau$ and $k$, respectively, for a prefix of \dna with $n=10,000$. 
The running time of \BASELINE increases with both $\tau$ and $k$. This is because a larger $\tau$ requires more binary search iterations to locate the resilient substrings, and also slows down the DP checks due to the increased number of substring occurrences. The increase due to $k$ is expected by the $\cO(kn^3\log n)$ time complexity. 
However, the running time of \RPMST and \RPMESA does not increase, similarly to \cref{fig:effi_runtime_tau} and \cref{fig:effi_runtime_k}. 
As expected by their time complexity, 
both \RPMST and \RPMESA substantially outperform \BASELINE (\RPMST is more than~$2$ orders of magnitude faster and \RPMESA is around $4$ orders of magnitude faster).
\cref{fig:runtime:baseline:n} and \cref{fig:memory:baseline:n}
show that both the running time and the memory usage for \BASELINE increase quadratically with $n$, while those for \RPMST and \RPMESA increase linearly with $n$; in this experiment, we have set $\tau=50$ and $k=10$. 

\begin{table*}[t]
\centering
\caption{(a) Clustering with $(\tau,k)$-resilient substrings vs ground truth clustering. (b, c) Percentage increase $\frac{M_R-M_C}{M_C}\cdot 100\%$, where,  for 
a clustering quality measure (max/avg NMI and max/avg ARI),  
$M_R$ is the measure for our approach and $M_C$ is the same measure for (b) the clustering algorithm using $\tau$-frequent substrings and (c) the clustering algorithm in~\cite{DBLP:journals/bioinformatics/VingaA03}.}\label{table:clustering}
\centering
\begin{subtable}[t]{0.5\textwidth}
\caption{}\label{table:clustering_gt}
\resizebox{\columnwidth}{!}{
\begin{tabular}{|c||c|c|c|c|}  
 \hline
 {\bf Dataset} & {\bf Max NMI} & {\bf Avg NMI} & {\bf Max RI} & {\bf Avg RI} \\ [0.5ex] 
 \hline\hline
 \EBOL &  1 & 0.83 & 1 & 0.87\\
 \COR  &  0.86 & 0.80 &0.96 & 0.92\\
 \hline
 \end{tabular}
 }
\end{subtable}\hfill{}
\begin{subtable}[t]{0.5\textwidth}
\centering
\caption{}\label{table:clusteringimprovement}
\resizebox{\linewidth}{!}{
\begin{tabular}{|c||c|c|c|c|}  
 \hline
 {\bf Dataset} & {\bf Max NMI} & {\bf Avg NMI} & {\bf Max RI} & {\bf Avg RI} \\ [0.5ex] 
 \hline\hline
 \EBOL & 22.6\% & 1.5\% & 25.3\% & 9.2\% \\
 \COR  & 8.5\% & 0.1\% & 2.9\% & 0.2\% \\
 \hline
\end{tabular}
}
\end{subtable}
\hfill
\begin{subtable}[t]{0.5\textwidth}
\centering
\caption{}\label{table:clusteringimprovement2}
\resizebox{\linewidth}{!}{
\begin{tabular}{|c||c|c|c|c|}  
 \hline
 {\bf Dataset} & {\bf Max NMI} & {\bf Avg NMI} & {\bf Max RI} & {\bf Avg RI} \\ [0.5ex] 
 \hline\hline
 \EBOL & 22.6\% & 1.3\% & 25.3\% & 9.0\% \\
 \COR  & 10.2\% & 1.4\% & 11.8\% & 9.4\% \\
 \hline
\end{tabular}
}
\end{subtable}
\end{table*}

\paragraph*{Clustering Case Study.}~We show the benefit of clustering based on
$(\tau, k)$-resilient substrings  following the methodology of~\cite{DBLP:journals/tkde/WuZLGZFW23,DBLP:journals/bioinformatics/VingaA03,DBLP:conf/icdm/0012ZLPM24}. 

A clustering is a partition of a collection of strings into sets called clusters. We consider two clusterings $C$, $C'$ with sizes $|C|$ and $|C'|$, respectively, and a collection of $N$ strings. $C$ is the ground truth clustering.

We perform pattern-based clustering of a dataset comprised of multiple strings: the top-$K$ most frequent $(\tau,k)$-resilient substrings are the features, the frequency of a substring in a string is the feature value, and the $\mathrm{k}$-means clustering algorithm is used to cluster the feature matrix. We set $\mathrm{k}=5$ in $\mathrm{k}$-means, as the ground truth clustering has $5$ clusters. 

We evaluated how similar is the clustering obtained from the clustering algorithm based on $(\tau, k)$-resilient substrings to the ground truth; see \cref{table:clustering_gt}. We used 
two well-known measures: Normalized Mutual Information (NMI) \cite{DBLP:journals/jmlr/NguyenEB10} and Rand Index (RI) \cite{rand1971objective} (see \cref{app:measures} for details), 
which take values in $[0,1]$ (higher values are preferred).

NMI and RI quantify how similar is $C'$ to the ground truth clustering. The values in all these measures are in $[0,1]$, where a higher value indicates that $C'$ is closer to the ground truth clustering. 

The NMI between $C$ and $C'$ is defined as:
$$\text{NMI}(C, C') = \frac{2 \cdot \sum_{i=1}^{|C|} \sum_{j=1}^{|C'|} \frac{n_{ij}}{N} \log \left( \frac{n_{ij/N}}{n_i\cdot \hat{n}_j / N^2} \right)}{-\sum_{i =1}^{|C|} \frac{n_i}{N}  \log \frac{n_i}{N} - \sum_{j =1}^{|C'|} \frac{\hat{n}_j}{N} \log \frac{\hat{n}_j}{N}},$$

where $n_i$ denotes the number of strings in the $i$-th cluster in $C$, $\hat{n}_j$ denotes the number of strings in the $j$-th cluster in $C'$, and $n_{ij}$ denotes the number of clusters belonging both in the $i$-th cluster in $C$ and in the $j$-th cluster in $C'$. 

The Rand Index (RI)  
measures similarity between the clustering $C'$ and the ground truth clustering $C$ by 
counting the number of pairs of strings that are placed in the same clusters in $C$ and $C'$, in different clusters in $C$ and $C'$, or in the same clusters in one of the $C$ and $C'$ and in different clusters in the other.  
$\text{RI}(C, C')$ is defined as follows:

$$
\text{RI}(C, C') = \frac{a + b}{a+b+c+d},
$$

\noindent where:
\begin{itemize}
  \item $a$ is the number of pairs of strings that are in the same cluster in $C$ and $C'$. 
  \item $b$ is the number of pairs of strings that are in different clusters in both $C$ and $C'$.
  \item $c$ is is the number of pairs of strings that are in the same cluster in $C$ but in different clusters in $C'$.
  \item $d$ is is the number of pairs of strings that are in different clusters in $C$ but in the same cluster in $C'$.
\end{itemize}

The results are the maximum and average value of these measures over all combinations of $\tau$, $k$, and $K$ such that $\tau=k\in \{60,70,80,90,100\}$ and $K\in \{50,150,250\}$. The clusters constructed by our approach are very close or the same as those in the ground truth clustering. Using the same measures and setting, we also that our approach is better in terms of finding a clustering close to the ground truth  than: (1) the same clustering algorithm based on $\tau$-frequent substrings instead of $(\tau,k)$-resilient substrings (\cref{table:clusteringimprovement}); and (2) the \FC  algorithm (\cref{table:clusteringimprovement2}).

\section{Conclusions}\label{sec:conclusion}
We introduced the \RPM problem and proposed two algorithms to solve it: an $\cO(n^3k\log n)$-time and $\cO(n^2)$-space algorithm based on DP; and an $\cO(n\log n)$-time and $\cO(n)$-space algorithm based on advanced data structures and combinatorial insights. 

Our experiments show that the latter algorithm is highly efficient and that resilient substrings are fundamentally different than frequent ones and are also useful in clustering. 

\section*{Acknowledgments}
SPP is supported by the PANGAIA and ALPACA projects that have received funding from the European Union’s Horizon 2020 research and innovation programme under the Marie Skłodowska-Curie grant agreements No 872539 and 956229, respectively.

\bibliographystyle{alphaurl}
\bibliography{references}

\appendix

\newpage

\section{Listing the \boldmath{$(\tau,k)$}-Resilient Substrings}\label{app:output}  
Assume that we have access to $\textsf{SA}(S)$, $\textsf{LCP}(S)$, and
the \textsf{OUTPUT} array. We traverse the $\textsf{SA}(S)$ from top to bottom and output the set of $(\tau,k)$-resilient substrings explicitly in $\cO(n)$ time plus time proportional to their total number. In particular, if a $(\tau,k)$-resilient substring occurs in multiple positions of $S$, the latter positions will form a \emph{range} (sub-array) of $\textsf{SA}(S)$. Thus we can avoid reporting the substring again and again by traversing $\textsf{SA}(S)$ from top to bottom and checking the corresponding \textsf{LCP} values.

\section{The DP-based Algorithm}\label{app:DP}

We represent $\occ_{S}(Z)$ as a set of intervals $\mathcal{I}=\{(i_0, i_0 + |Z|-1, (i_1, i_1 + |Z|-1), \ldots, (i_{f-1}, i_{f-1} + |Z| - 1)\}$. We assume without a loss of generality that $i_0< i_1< \ldots <i_{f-1}$. The dynamic programming algorithm~\cite{DBLP:journals/ipl/Damaschke17} underlying 
\cref{the:intervals}
has three  steps: 
\begin{enumerate}
\item We compute a matrix $W$ of size $f \times f$, such that for every $i_a, i_b \in occ_S(Z)$, the element $W[a][b]$ stores the number of intervals that do not contain position $i_a + |Z| - 1$ but contain position $i_b + |Z| -1$, where $i_a <i_b$.

\item We compute a matrix $T$ of size $f\times k$ as follows.
The first column of $T$ (i.e., the element $T[b][0]$, for all $b=0, 1 \ldots, f-1$) is initialized with the number of intervals that contain $i_b + |Z| -1 $. The remaining of the first row of $T$ (i.e., the element $T[0][h]$, for all $h=1, \ldots, k-1$) is initialized with the number of intervals that contain $i_0 + |Z| - 1$. Every other element $T[b][h]$ is then computed using $\max_{a<b}\{T[a][h-1]+W[a][b]\}$. 
\item We output $\max_{b}\{T[b][k-1]\}$, as the maximum number of intervals that can be hit by a subset of $k$ points. 
\end{enumerate}

\section{Details of the Main Algorithm}

\subsection{Simulating \boldmath{$\textsf{ST}(S)$}}\label{app:practical}
Note that instead of using $\textsf{ST}(S)$, we can apply the algorithm of Abouelhoda et al.~\cite[Algorithm 4.4]{DBLP:journals/jda/AbouelhodaKO04}, which simulates a bottom-up traversal of $\textsf{ST}(S)$ using $\textsf{SA}(S)$ and the \textsf{LCP} array of $S$ in $\cO(n)$ time. This data structure is known as the \emph{enhanced suffix array} (ESA).  Simulating the $\textsf{ST}(S)$ using ESA is usually faster, due to cache-friendly operations, and more space efficient in practice.

\subsection{AVL Trees for Maintaining Sorted Lists}\label{app:AVL_trees}

The following lemmas related to the merge operations of two AVL trees  are utilized:

\setcounter{lemma}{11}
\begin{lemma}[\hspace{0.2mm}\cite{c0176eeabbb}]\label{lemmaAVL}
Two AVL trees of size at most $n_1$ and $n_2$ can be merged in time
$\cO(\log\binom{n_1+n_2}{n_1})$.
\end{lemma}

Using the {\sl smaller-half trick},  the sum over all nodes  of  an arbitrary binary tree  of the terms stated  in \cref{lemmaAVL} is given as follows:

\begin{lemma}[\hspace{0.2mm}\cite{c0176eeabbb}]
Let $T$ be an arbitrary  binary tree with $n$ leaves. The sum over all internal nodes $v$ in $T$ of terms $\log\binom{n_1+n_2}{n_1}$, where $n_1$ and $n_2$ are the numbers of leaves in the subtrees rooted at
the two children of $v$ and $n_1 \leq n_2$, is $\cO(n\log n)$.
\end{lemma}

The suffix tree $\textsf{ST}(S)$ can be converted into a binary tree via the addition of dummy nodes in $\cO(n)$ time and space, therefore,  maintaining the sorted list of occurrences of $\textsf{str}(v)$ for each node $v$ during a DFS traversal can be done in $\cO(n\log n)$ time.

\section{Evaluation Measures for Clustering}\label{app:measures}
In the following, we discuss the two well-known measures of clustering quality we used in the clustering case study, namely Normalized Mutual Information (NMI) \cite{DBLP:journals/jmlr/NguyenEB10} and the Rand Index (RI) \cite{rand1971objective}.

A clustering is a partition of a collection of strings into sets called clusters. We consider two clusterings $C$, $C'$ with sizes $|C|$ and $|C'|$, respectively, and a collection of $N$ strings. $C$ is the ground truth clustering. 
NMI and RI quantify how similar is $C'$ to the ground truth clustering. The values in all these measures are in $[0,1]$, where a higher value indicates that~$C'$ is closer to the ground truth clustering. 

The NMI between $C$ and $C'$ is defined as:
\[\text{NMI}(C, C') = \frac{2 \cdot \sum_{i=1}^{|C|} \sum_{j=1}^{|C'|} \frac{n_{ij}}{N} \log \left( \frac{n_{ij/N}}{n_i\cdot \hat{n}_j / N^2} \right)}{-\sum_{i =1}^{|C|} \frac{n_i}{N}  \log \frac{n_i}{N} - \sum_{j =1}^{|C'|} \frac{\hat{n}_j}{N} \log \frac{\hat{n}_j}{N}},\]
where $n_i$ denotes the number of strings in the $i$-th cluster in $C$, $\hat{n}_j$ denotes the number of strings in the $j$-th cluster in $C'$, and $n_{ij}$ denotes the number of clusters belonging both in the $i$-th cluster in $C$ and in the $j$-th cluster in $C'$. 

The Rand Index (RI) measures similarity between the clustering $C'$ and the ground truth clustering $C$ by 
counting the number of pairs of strings that are placed in the same clusters in $C$ and $C'$, in different clusters in $C$ and $C'$, or in the same clusters in one of the $C$ and $C'$ and in different clusters in the other.  
 
$\text{RI}(C, C')$ is defined as follows:

\[\text{RI}(C, C') = \frac{a + b}{a+b+c+d},\]

\noindent where:
\begin{itemize}
  \item $a$ is the number of pairs of strings that are in the same cluster in $C$ and $C'$. 
  \item $b$ is the number of pairs of strings that are in different clusters in both $C$ and $C'$.
  \item $c$ is is the number of pairs of strings that are in the same cluster in $C$ but in different clusters in $C'$.
  \item $d$ is is the number of pairs of strings that are in different clusters in $C$ but in the same cluster in $C'$.
\end{itemize}

\end{document}